\documentclass[11pt,a4paper]{article}

\usepackage{hyperref}
\usepackage{cite}

\usepackage{amsmath}
\usepackage{amsfonts,amsthm,amssymb}

\numberwithin{equation}{section}

\newcommand{\beqa}{\begin{eqnarray}}
\newcommand{\eeqa}{\end{eqnarray}}
\newcommand{\rf}[1]{(\ref{#1})}

\newtheorem{theorem}{Theorem}[section]
\newtheorem{proposition}{Proposition}[section]
\newtheorem{lemma}{Lemma}[section]

\newtheorem{corollary}{Corollary}[section]

\newtheorem{property}{Property}[section]
{\theoremstyle{remark}
}
\newtheorem{identity}{Identity}
\numberwithin{equation}{section}

\begin{document}

\author{\vspace{1cm}\vspace{2cm}}

\bigskip \begin{flushright}
LPENSL-TH-04-19
\end{flushright}

\bigskip \bigskip

\begin{center}
\textbf{\LARGE On Separation of Variables for Reflection Algebras}

\vspace{50pt}

{\large\textbf{J. M. Maillet}\footnote{{ Univ Lyon, Ens de Lyon, Univ
Claude Bernard, CNRS, Laboratoire de Physique, F-69342 Lyon, France;
maillet@ens-lyon.fr}}}\ \ \thinspace { \textbf{and \ G. Niccoli}\footnote{%
{ Univ Lyon, Ens de Lyon, Univ Claude Bernard, CNRS, Laboratoire de
Physique, F-69342 Lyon, France; giuliano.niccoli@ens-lyon.fr}}}
\bigskip

\vspace{50pt}

\end{center}

\begin{itemize}
\item[ ] \textbf{Abstract:}\thinspace \thinspace We implement
our new Separation of Variables (SoV) approach for open quantum
integrable models associated to higher rank representations of the
reflection algebras. We construct the 
(SoV) basis for the fundamental representations of the $Y(gl_n)$ reflection
algebra associated to general integrable boundary conditions.
Moreover, we give the conditions on the boundary $K$-matrices allowing for the transfer matrix to be diagonalizable with simple spectrum. These SoV basis  are then used to completely
characterize the transfer matrix spectrum for the rank one and two reflection
algebras. The rank one case is developed for both the rational and
trigonometric fundamental representations of the 6-vertex reflection
algebra. On the one hand, we extend the complete spectrum
characterization to representations previously lying outside the SoV
approach, e.g. those for which the standard Algebraic Bethe Ansatz applies.
On the other hand, we show that our new SoV construction can be reduced to
the generalized Sklyanin's one whenever it is applicable. The rank two
case is developed explicitly for the fundamental representations of the $Y(gl_3)$ reflection algebra associated to general integrable boundary
conditions. For both rank one and two our SoV approach leads to a complete
characterization of the transfer matrix spectrum in terms of a set of
polynomial solutions to the corresponding quantum spectral curve equation. Those are finite
difference functional equations of order equal to the rank plus one, i.e., here two and three
respectively for the $Y(gl_2)$ and $Y(gl_3)$ reflection algebras.
\end{itemize}

\newpage
\tableofcontents \newpage

\section{Introduction}

In this article we continue the development of our new method \cite{MaiN18} for constructing Separation of Variables (SoV) basis of quantum integrable lattice models. While in \cite{MaiN18,MaiN18a,MaiN18b,MaiN19} we have considered  various representations of the Yang-Baxter algebra for integrable models with quasi-periodic boundary conditions, we explore in the present article representations of the reflection algebras describing open quantum integrable models with general boundary conditions preserving integrability. We use the framework of the Quantum Inverse Scattering Method (QISM) \cite{FadS78,FadST79,FadT79,Skl79,Skl79a,FadT81,Skl82,Fad82,Fad96} extended to the reflection algebras situation \cite{Che84,Skl88} to describe these open quantum integrable lattice models.  Their exact solution  has become a subject of intense studies 
\cite{Che84,Gau71a,Bar79,Bar80,Bar80a,Sch85,AlcBBBQ87,Skl88,Bar88,MezNR90,PasS90,BatMNR90,KulS91,FriM91,MezN91,MezN91b,deVG93,deVG93a,deVG94,GhoZ94a,GhoZ94b,ArtMN95,JimKKKM95, 
JimKKMW95,Kul96,FanHSY96,Zho96, AsaS96,Zho97,GuaWY97,
ShiW97,Tsu98,Gua00,MinRS01,Nep02,NepR03,CaoLSW03,Doi03,Nep04,deGP04,ArnACDFR04,GalM05,ArnCDFR05,ArnCDFR06,MurNS06,Doi06,YanNZ06,KitKMNST07,BasK07,
YanZ07,RagS07,FraNR07,AmiFOR07,KitKMNST08,Gal08,BelR09,Nep10,FilK10,FilK11,Fil11,CraR12,CaoYSW13b,BasB13,BelC13,BelCR13,CaoYSW14,Bel15,BelP15,AvaBGP15,XuHYCYS16,
BasB17,
DerKM03b,FraSW08,AmiFOW10,FraGSW11,Nic12, FalN14, FalKN14, KitMN14,MaiNP17,KitMNT17,MaiNP18,KitMNT18}
in quantum integrability. They have attracted a large research interest both because of  their potential relevance in the description of the non-equilibrium
and transport properties of quantum integrable systems \cite{deGE05,deGE06,KinWW06,HofLFSS07,BloDZ08,CraRS10,Pro11,TroCFMSEB12,SchHRWBBBDMRR12,RonSBHLMHBS13,EisFG15}
and the fact that their solution by standard methods has represented
a longstanding challenge for the most general integrable boundary
conditions. Important instances are the integrable $XXZ$ (and $XYZ$) quantum
spin 1/2 chains with the most general integrable boundary magnetic fields.
Their Analytic or Algebraic Bethe Ansatz descriptions have been first confined
to the case of $z$-oriented boundary magnetic fields \cite{AlcBBBQ87,Skl88} and then to the case in which the left and right boundary
magnetic fields satisfy some special relation \cite{FanHSY96,Nep02,NepR03,CaoLSW03,deGP04,YanZ07}.  The problem of the most general
integrable boundaries has been
overcome in the framework of the Sklyanin's quantum version of the
Separation of Variables (SoV) \cite{Skl85,Skl85a,Skl92,Skl95}. More in detail in 
\cite{Nic12,FalKN14,FalN14} a generalized version of the Sklyanin's SoV approach has been
introduced by the combined used of SoV and of the (Vertex-IRF) Baxter's
gauge transformations \cite{Bax73a,Bax82B}. This generalized version of the
Sklyanin's SoV approach has lead to the complete characterization of the
transfer matrix spectrum with unconstrained integrable boundaries. While in 
\cite{KitMN14} it has been first proven the equivalence of this SoV
characterization with a second order difference type functional equation of
Baxter's form containing an inhomogeneous term under the most general
integrable boundary conditions\footnote{For these unconstrained integrable boundaries, see also 
\cite{Bel15,BelP15,AvaBGP15} for a
subsequent analysis by a generalized version of the Algebraic Bethe Ansatz
approach and \cite{CaoYSW13b} for another Ansatz description of
the transfer matrix eigenvalue spectrum, where in fact inhomogeneous
Baxter's type functional equations first appeared.}.

The Sklyanin's Separation of Variables method, or some minor 
generalizations of it, has by now found a large set of applications
producing cutting edge results on the construction of the full eigenvalue
and eigenvector spectrum of the transfer matrices associated to a large
class of integrable quantum models \cite{DerKM03b,FraSW08,AmiFOW10,FraGSW11,Nic12, FalN14, FalKN14, KitMN14,MaiNP17,KitMNT17,MaiNP18,KitMNT18,BabBS96,BabBS97,Smi98a,DerKM03,BytT06,vonGIPS06,vonGIPS09,vonGIPST07,
vonGIPST08,
NicT10,Nic10a,GroN12,GroMN12,Nic13,Nic13b,Nic13a,GroMN14,NicT15,NicT16,LevNT16,KitMNT16,Smi98,NieWF09,JiaKKS16,RyaV18,MaiN18,MaiN18a,MaiN18b,MaiN19}, while also providing some fundamental steps towards the computation of their 
form factors and correlation functions \cite{GroMN12,Nic13,Nic13b,Nic13a,GroMN14,NicT15,NicT16,LevNT16,KitMNT16,KitMNT17,KitMNT18}. Nevertheless, one has to
mention that some applicability issues have been encountered in particular
in relation with the quantum integrable models associated to higher rank
representations of the Yang-Baxter and reflection algebras. Moreover, these
issues appeared to exist already for the rank one quantum integrable models when
some special integrable boundary conditions are considered. This is in
particular the case for boundary conditions allowing for the application of
the standard Algebraic Bethe Ansatz approach, e.g. the XXZ quantum spin 1/2
chain with closed and open boundary conditions associated to diagonal twist
and diagonal boundary matrices, respectively. This may have generated the
wrong perception that SoV and ABA methods have disjoint applicability ranges.

Our new SoV method \cite{MaiN18} allows to overcome these problems while it is
proven to be reducible to the original Sklyanin's SoV approach when this last one is
applicable. We have already proven its efficiency  to completely characterize
the spectrum for fundamental representations of both the rational and
trigonometric Yang-Baxter algebra for any positive integer rank \cite%
{MaiN18, MaiN18a, MaiN18b} and under general closed integrable boundary
conditions. Moreover,  in \cite{MaiN19} we have described its application to
non-fundamental representations of the Yang-Baxter algebras. Already for the rank one case our analysis \cite%
{MaiN18} has proven the applicability of the SoV method beyond the limits of
the original one, being applicable as well for diagonal and non-diagonal
boundary twists, overcoming for these representations the
apparent dualism between standard ABA and SoV approaches.

In the present article, we show that our SoV approach can be applied to higher rank 
fundamental representations of the reflection algebra  and under the most general integrable boundary
conditions. Once again we get new results already for the rank one case
leading to the extension of the SoV approach to representations previously
unattainable by the generalized version of the Sklyanin's method.

More in detail the article is organized as follows. In Section 2, we develop the analysis for general fundamental representations of $Y(gl_2)$ reflection algebra. First, we analyze the applicability of SoV in a generalized Sklyanin's approach. Then we present our new SoV approach. We compare them and we show their consistence in the overlapping region of applicability, while proving that our SoV construction does not suffer the limitations of the original one. Then, we derive the transfer matrix spectrum in our SoV scheme by the quantum spectral curve equation and we define general criteria to establish the diagonalizability and simplicity of the transfer matrix. Section 3 is devoted to the generalization of our SoV analysis and results to the fundamental representations of $U_q(gl_2)$ reflection algebra. By carefully taking the rational limit of the trigonometric 6-vertex $R$ and $K$ matrices, we are able to obtain the trigonometric results from their rational counterparts. In Section 4, we develop our SoV analysis of the fundamental representations of $Y(gl_3)$ reflection algebra. We construct our SoV basis and we obtain complete characterization of the transfer matrix spectrum both by solutions to a discrete system of equations and to the quantum spectral curve equation. Finally, Section 5 is devoted to the SoV basis construction for the fundamental representations of $Y(gl_n)$ reflection algebra and to the identification of general criteria for the diagonalizability and simplicity of the transfer matrix. In Section 6 we present some conclusions. Finally, in the appendix, the scalar products of separate vectors and co-vectors are determined for the rank one case in our new SoV approach.

\section{SoV for fundamental representations of $Y(gl_{2})$ reflection algebra}

The transfer matrix associated to the fundamental representations of $gl_{2}$
reflection algebra reads: 
\begin{equation}
T(\lambda )=\text{tr}_{0}\{K_{+}(\lambda
)\,M(\lambda )\,K_{-}(\lambda )\,\hat{M}(\lambda )\}=\text{tr}_{0}\left\{
K_{+}(\lambda )\,\mathcal{U}_{-}(\lambda )\right\} .  \label{T-open++}
\end{equation}%
It defines a one-parameter family of commuting operators \cite{Skl88} on the
quantum space $\mathcal{H}=\otimes _{i-1}^{\mathsf{N}}V_{i}$, with $%
V_{i}\simeq \mathbb{C}^{2}$, of the $\mathsf{N}$ sites bidimensional
fundamental representations of the reflection algebra \cite{Che84}. The
transfer matrix is introduced in terms of the following definitions. First
we define the boundary matrices 
\begin{equation}
K_{+}(\lambda )=K_{-}(\lambda +\eta ;\zeta _{+},\kappa _{+},\tau _{+}),
\label{Kpm++}
\end{equation}%
and%
\begin{equation}
K_{-}(\lambda ;\zeta _{-},\kappa _{-},\tau _{-})=\frac{1}{\zeta _{-}}\left( 
\begin{array}{cc}
\zeta _{-}+\lambda -\eta /2 & 2\kappa _{-}e^{\tau _{-}}(\lambda -\eta /2) \\ 
2\kappa _{-}e^{-\tau _{-}}(\lambda -\eta /2) & \zeta _{-}-\lambda +\eta /2%
\end{array}%
\right) \in \mathrm{End}(\mathbb{C}^{2}),  \label{Kxxx}
\end{equation}%
where $K_{-}(\lambda)$ is the most general scalar solution \cite%
{deVG93,deVG94,GhoZ94a,GhoZ94b} of the rational 6-vertex reflection equation:%
\begin{equation}
R_{ab}(\lambda -\mu )\,K_{-,a}(\lambda )\,R_{a,b}(\lambda +\mu -\eta
)\,K_{-,b}(\mu )=K_{-,b}(\mu )\,R_{ab}(\lambda +\mu -\eta )\,K_{-,a}(\lambda
)\,R_{ab}(\lambda -\mu )\in \mathrm{End}(V_{a}\otimes V_{b}),
\end{equation}%
w.r.t. the rational 6-vertex $R$-matrix:%
\begin{equation}
R_{ab}(\lambda )=\left( 
\begin{array}{cccc}
\lambda +\eta & 0 & 0 & 0 \\ 
0 & \lambda & \eta & 0 \\ 
0 & \eta & \lambda & 0 \\ 
0 & 0 & 0 & \lambda +\eta%
\end{array}%
\right) \in \mathrm{End}(V_{a}\otimes V_{b}).
\end{equation}%
Then using it we can define the boundary monodromy matrix: 
\begin{equation}
\mathcal{U}_{-,0}(\lambda )=M_{0}(\lambda )\,K_{-,0}(\lambda )\,\hat{M}%
_{0}(\lambda )=\left( 
\begin{array}{cc}
\mathcal{A}_{-}(\lambda ) & \mathcal{B}_{-}(\lambda ) \\ 
\mathcal{C}_{-}(\lambda ) & \mathcal{D}_{-}(\lambda )%
\end{array}%
\right)_{(0)} \in \mathrm{End}(V_{0}\otimes \mathcal{H}),  \label{U+}
\end{equation}
an operator solution to the same reflection equation \cite{Skl88}%
\begin{equation}
\hspace{-0.1cm}R_{ab}(\lambda -\mu )\,\mathcal{U}_{-,a}(\lambda
)\,R_{ab}(\lambda +\mu -\eta )\,\mathcal{U}_{-,b}(\mu )=\mathcal{U}%
_{-.b}(\mu )\,R_{ab}(\lambda +\mu -\eta )\,\mathcal{U}_{-,a}(\lambda
)\,R_{ab}(\lambda -\mu )\in \mathrm{End}(V_{a}\otimes V_{b}\otimes \mathcal{H}%
),
\end{equation}%
where we have defined:%
\begin{equation}
\hat{M}_{0}(\lambda )=(-1)^{\mathsf{N}}\,\sigma
_{0}^{y}\,M_{0}^{t_{0}}(-\lambda )\,\sigma _{0}^{y},
\end{equation}%
in terms of the bulk monodromy matrix:%
\begin{equation}
M_{0}(\lambda )=R_{0\mathsf{N}}(\lambda -\xi _{\mathsf{N}}^{(0)})\dots
R_{01}(\lambda -\xi _{1}^{(0)})=\left( 
\begin{array}{cc}
A(\lambda ) & B(\lambda ) \\ 
C(\lambda ) & D(\lambda )%
\end{array}%
\right) ,  \label{mon++}
\end{equation}%
satisfying the rational 6-vertex Yang-Baxter algebra: 
\begin{equation}
R_{ab}(\lambda -\mu )\,M_{a}(\lambda )\,M_{b}(\mu )=M_{b}(\mu
)\,M_{a}(\lambda )\,R_{ab}(\lambda -\mu )\in \mathrm{End}(V_{a}\otimes
V_{b}\otimes \mathcal{H}),  \label{YB-R}
\end{equation}%
with 
\begin{equation}
\xi _{n}^{(h)}=\xi _{n}+\eta /2-h\eta ,\qquad 1\leq n\leq \mathsf{N},\quad
h\in \{0,1\}.  \label{xi-h++}
\end{equation}

The commutativity of the transfer matrix family $T^{\left( K_{+,-}\right)
}(\lambda )$ has been first proven by Sklyanin \cite{Skl88} as a consequence
of the reflection equation satisfied by $\mathcal{U}_{-}(\lambda )$ and $%
K_{+}(\lambda )$.

It is worth noticing that the most general boundary matrices are in fact of the following general form:%
\begin{equation}
K_{\pm }(\lambda )=I+\frac{\lambda \pm \eta /2}{\bar{\zeta}_{\pm }}\mathcal{M%
}^{\left( \pm \right) },  \label{Sl2-Kpm}
\end{equation}%
where%
\begin{eqnarray}
\mathcal{M}^{\left( \pm \right) 2} &=&r^{\left( \pm \right) }I,\text{ with }%
r^{\left( \pm \right) }=1-\delta _{\kappa _{\pm }^{2},-1/4}, \\
\bar{\zeta}_{\pm } &=&\zeta _{\pm }\delta _{\kappa _{\pm }^{2},-1/4}+\zeta
_{\pm }r^{\left( \pm \right) }/\sqrt{1+4\kappa _{\pm }^{2}}.
\end{eqnarray}%
Moreover, in the case $r^{\left( \pm \right) }=1$ and $K_{\pm }(\lambda )$
not proportional to the identity, there exist $S^{\left( \pm \right) }$
invertible $2\times 2$ matrices such that $\mathcal{M}^{(\pm )}=S^{\left(
\pm \right) }\sigma ^{z}\left( S^{\left( \pm \right) }\right) ^{-1}$. Then,
one can easily verify that $K_{\pm }(\lambda )$ are one-parameter families
of commuting matrices:%
\begin{equation}
\lbrack K_{+}(\lambda ),K_{-}(\mu )]=\frac{(\lambda +\eta /2)(\mu -\eta /2)}{%
\bar{\zeta}_{+}\bar{\zeta}_{-}}[\mathcal{M}^{\left( +\right) },\mathcal{M}%
^{\left( -\right) }]=0
\end{equation}%
if and only if:%
\begin{equation}
\kappa _{-}e^{\pm \tau _{-}}=\kappa _{+}e^{\pm \tau _{+}}\equiv \kappa
e^{\pm \tau },  \label{Cond-S-Diag}
\end{equation}%
so that $K_{\pm }(\lambda )$ are simultaneously diagonalizable if and only
if the above conditions are satisfied and $\kappa _{\pm }^{2}\neq -1/4$.
Finally, as already described in \cite{KitMNT16}, let us observe that if the
matrices $K_{\pm }(\lambda )$ are non-commuting, i.e., if there exists $a\in \{-1,1\}$ such that:%
\begin{equation}
\kappa _{-}e^{a\tau _{-}}\neq \kappa
_{+}e^{a\tau _{+}},  \label{Cond-NoN-S-Diag}
\end{equation}%
then there exists a couple $(\epsilon _{+},\epsilon _{-})\in \{-1,1\}^{2}$
such that the following matrix is invertible%
\begin{equation}
W^{(K_{+,-})}\equiv \left( 
\begin{array}{cc}
1 & -\frac{1-\epsilon _{+}\sqrt{1+4\kappa _{+}^{2}}}{2\kappa _{+}e^{-\tau
_{+}}} \\ 
\frac{1-\epsilon _{-}\sqrt{1+4\kappa _{-}^{2}}}{2\kappa _{-}e^{\tau _{-}}} & 
1%
\end{array}%
\right) ,
\end{equation}%
and we can define the following similarity transformation%
\begin{equation}
\bar{K}_{\mp }(\lambda )=W_{0}^{(K_{+,-})}\,K_{\mp }(\lambda
)\,(W_{0}^{(K_{+,-})-1})^{-1}=\left( 
\begin{array}{cc}
\bar{a}_{\mp }(\lambda ) & \bar{b}_{\mp }(\lambda ) \\ 
\bar{c}_{\mp }(\lambda ) & \bar{d}_{\mp }(\lambda )%
\end{array}%
\right) ,  \label{XXX-Similar}
\end{equation}%
where it holds:%
\begin{equation}
\bar{K}_{+}(\lambda )=\mathrm{I}+\frac{\lambda +\eta /2}{\bar{\zeta}_{+}}%
(\sigma ^{z}+\bar{\mathsf{c}}_{+}\sigma ^{-}),\qquad \bar{K}_{-}(\lambda )=%
\mathrm{I}+\frac{\lambda -\eta /2}{\bar{\zeta}_{-}}(\sigma ^{z}+\bar{\mathsf{%
b}}_{-}\sigma ^{+}),
\end{equation}%
with 
\begin{align}
& \bar{\zeta}_{\pm }=\epsilon _{\pm }\frac{\zeta _{\pm }}{\sqrt{1+4\kappa
_{\pm }^{2}}}\,,  \label{ZeBar} \\
& \bar{\mathsf{c}}_{+}=\epsilon _{+}\,\frac{2\kappa _{+}e^{-\tau _{+}}}{%
\sqrt{1+4\kappa _{+}^{2}}}\left[ 1+\frac{(1+\epsilon _{+}\sqrt{1+4\kappa
_{+}^{2}})(1-\epsilon _{-}\sqrt{1+4\kappa _{-}^{2}})}{4\kappa _{+}\kappa
_{-}e^{\tau _{-}-\tau _{+}}}\right] , \\
& \bar{\mathsf{b}}_{-}=\epsilon _{-}\,\frac{2\kappa _{-}e^{\tau _{-}}}{\sqrt{%
1+4\kappa _{-}^{2}}}\left[ 1+\frac{(1-\epsilon _{+}\sqrt{1+4\kappa _{+}^{2}}%
)(1+\epsilon _{-}\sqrt{1+4\kappa _{-}^{2}})}{4\kappa _{+}\kappa _{-}e^{\tau
_{-}-\tau _{+}}}\right],  \label{B_-Bar}
\end{align}
and\footnote{%
Note that the assumption that the boundary matrices are non-commuting
implies that they are not simultaneously diagonalizable. At least one of the
conditions $\bar{\mathsf{c}}_{+}\neq 0$ or $\bar{\mathsf{b}}_{-}\neq 0$
must be satisfied and with a proper choice of $(\epsilon _{+},\epsilon
_{-})\in \{-1,1\}^{2}$ we can obtain that  the second inequality holds.} $\bar{%
\mathsf{b}}_{-}\neq 0$.

\subsection{Applicability of SoV in generalized Sklyanin's approach}

The Sklyanin's approach to define SoV can be generalized to the reflection
algebra case. We can summarize the applicability of this approach\footnote{See also \cite{FraSW08,FraGSW11} for an earlier purely functional version of SoV (i.e.
without the construction of the SoV basis) for these representations.} as
developed in \cite{Nic12,FalKN14} and \cite{KitMNT17} by the following:

\begin{proposition}
Let us assume that for any $a$ and $b$ in $ \{1,...,\mathsf{N}\}$, with $a \neq b$,  the following condition on inhomogeneity parameters $\xi$'s 
\begin{equation}
 \xi_{a}\neq \xi _{b}+\epsilon\eta \text{ \ \ }\forall\,
\epsilon\in\{-1,0,1\},
\label{Inhomog-cond}
\end{equation}
holds and that the boundary matrix $K_{-}(\lambda )$ and $K_{+}(\lambda )$
are non-commuting ones, namely that $(\ref{Cond-NoN-S-Diag})$ is satisfied. Then, defining:%
\begin{equation}
\widehat{\mathcal{U}}_{-}(\lambda )=W_{0}^{(K_{+,-})}\,\mathcal{U}%
_{-}(\lambda )\,(W_{0}^{(K_{+,-})})^{-1}=\left( 
\begin{array}{cc}
\widehat{\mathcal{A}}_{-}(\lambda ) & \widehat{\mathcal{B}}_{-}(\lambda ) \\ 
\widehat{\mathcal{C}}_{-}(\lambda ) & \widehat{\mathcal{D}}_{-}(\lambda )%
\end{array}%
\right) ,
\end{equation}%
with $\mathcal{W}_{K_{+,-}}=\otimes _{a=1}^{\mathsf{N}}W_{a}^{(K_{+,-})}$,
the generalized Sklyanin's left and right SoV basis for the transfer matrix $%
T(\lambda )$ are the left and right eigenbasis of $%
\widehat{\mathcal{B}}_{-}(\lambda )$:%
\begin{equation}
\langle \,\mathbf{h}_{-}\,|\equiv \langle \,0\,|\mathcal{W}%
_{K_{+,-}}\prod_{n=1}^{\mathsf{N}}\left( \frac{\widehat{\mathcal{A}}%
_{-}(\eta /2-\xi _{n})}{\mathsf{A}_{-}(\eta /2-\xi _{n})}\right) ^{1-h_{n}},%
\text{ \ }|\,\mathbf{h}_{-}\,\rangle \equiv \prod_{n=1}^{\mathsf{N}}\left( 
\frac{\widehat{\mathcal{D}}_{-}(\xi _{n}+\eta /2)}{\text{\textsc{k}}_{n}\,%
\mathsf{A}_{-}(\eta /2-\xi _{n})}\right) ^{h_{n}}\mathcal{W}%
_{K_{+,-}}^{-1}|\,\underline{0}\,\rangle ,
\end{equation}%
with eigenvalues: 
\begin{equation*}
\mathsf{b}_{-,\mathbf{h}}(\lambda )=(-1)^{\mathsf{N}}\,\bar{\mathsf{b}}_{-}%
\frac{\lambda -\eta /2}{\bar{\zeta}_{-}}\,\prod_{n=1}^{\mathsf{N}}(\lambda
-\xi _{n}^{(h_{n})})(-\lambda -\xi _{n}^{(h_{n})}),\text{ \ }\bar{\zeta}%
_{-}=\epsilon _{-}\zeta _{-}/\,\sqrt{1+4\kappa _{-}^{2}}.
\end{equation*}%
Here $\langle \,0\,|$ is the co-vector with all spin up and $|\,\underline{0}%
\,\rangle $ is the vector with all spin down and 
\begin{equation}
\mathsf{A}_{-}(\lambda )=(-1)^{\mathsf{N}}\frac{\bar{\zeta}_{-}+\lambda
-\eta /2}{\bar{\zeta}_{-}}a(\lambda )\,d(-\lambda ),\qquad \text{\textsc{k}}%
_{n}=(\xi _{n}+\eta )/(\xi _{n}-\eta ),
\end{equation}%
and%
\begin{equation}
a(\lambda )\equiv \prod_{n=1}^{\mathsf{N}}(\lambda -\xi _{n}+\eta /2),\qquad
d(\lambda )\equiv \prod_{n=1}^{\mathsf{N}}(\lambda -\xi _{n}-\eta /2).
\label{a-d}
\end{equation}
\end{proposition}

As presented here, the Sklyanin's SoV basis can be defined only
in the case of non-commuting boundary matrices. Instead, as we will prove
in the next section, our new SoV approach works in the completely general
case. So we can use it also in the case of commuting boundary matrices for
which Algebraic Bethe Ansatz \cite{Skl88} also works in the special case of simultaneously 
diagonalizable boundary matrices.

\subsection{Our SoV approach}

Let us define:%
\begin{equation}
\mathsf{A}_{\bar{\zeta}_{+},\bar{\zeta}_{-}}(\lambda )\equiv (-1)^{\mathsf{N}%
}\frac{2\lambda +\eta }{2\lambda }\,\frac{(\lambda -\frac{\eta }{2}+\bar{%
\zeta}_{+})(\lambda -\frac{\eta }{2}+\bar{\zeta}_{-})}{\bar{\zeta}_{+}\,\bar{%
\zeta}_{-}}\,a(\lambda )\,d(-\lambda ),  \label{Coef-6v-Req-R}
\end{equation}%
then the following theorem holds:

\begin{theorem}
\label{Th-SoV-basis}i) Let $K_{-}(\lambda )$ and $K_{+}(\lambda )$ be
non-commutative boundary matrices $(\ref{Cond-NoN-S-Diag})$ and $%
T(\lambda )$ be the associated one-parameter family of transfer
matrix, then%
\begin{equation}
\langle h_{1},...,h_{\mathsf{N}}|\equiv \langle S|\prod_{n=1}^{\mathsf{N}%
}\left( \frac{T(\xi _{n}-\eta /2)}{\mathsf{A}_{\bar{\zeta}_{+},%
\bar{\zeta}_{-}}(\eta /2-\xi _{n})}\right) ^{1-h_{n}}\text{, }
\label{SoV-Basis-6v-Open}
\end{equation}%
for any $\{h_{1},...,h_{\mathsf{N}}\}\in \{0,1\}^{\otimes \mathsf{N}}$, is a
co-vector basis of $\mathcal{H}$ for almost any choice of the co-vector $%
\langle S|$ and of the inhomogeneity parameters satisfying the condition $(%
\ref{Inhomog-cond})$.

ii) Let $K_{-}(\lambda )$ and $K_{+}(\lambda )$ be commutative boundary
matrices $(\ref{Cond-S-Diag})$, moreover not both proportianl to the identity. Then, for any fixed choice of the boundary
parameters $\{\zeta _{+},\kappa ,\tau \}$ (or $\{\zeta _{-},\kappa ,\tau \}$%
), the set $(\ref{SoV-Basis-6v-Open})$ is a co-vector basis of $\mathcal{%
H}$ for almost any choice of the co-vector $\langle S|$, of the inhomogeneity
parameters satisfying the condition $(\ref{Inhomog-cond})$ and of $\zeta
_{-} $(or $\zeta _{+}$).
\end{theorem}

\begin{proof}
\textit{Let us prove i).} Let us consider the following choice on the
inhomogeneity parameters:%
\begin{equation}
\xi _{a}=a\xi \text{ \ }\forall a\in \{1,...,\mathsf{N}\},
\label{Relation-InH}
\end{equation}%
where $\xi$ is some complex parameter, then, by exactly the same steps followed in the proof of the general Proposition 2.4 of 
\cite{MaiN18}, we can prove that $T(\xi _{l}-\eta /2)$ are polynomials of
degree $2\mathsf{N}+1$ in $\xi $ for all $l\in \{1,...,\mathsf{N}\}$
with maximal degree coefficient given by:%
\begin{equation}
\frac{(-1)^{\mathsf{N}-l}\eta l(\mathsf{N}-l)!(\mathsf{N}+l)!}{\bar{\zeta}%
_{+}\bar{\zeta}_{-}}\mathcal{M}_{l}^{\left( -\right) }\mathcal{M}%
_{l}^{\left( +\right) }.  \label{Asymp-i}
\end{equation}%
Let us choose $\langle S|$ of tensor product form%
\begin{equation}
\langle S|=\otimes _{l=1}^{\mathsf{N}}\langle S,l|\text{ \ where \ \ }\langle
S,l|=(s_{+},s_{-})_{l} \ , \label{S-tensor}
\end{equation}%
then the set of co-vectors $(\ref{SoV-Basis-6v-Open})$ is a basis as soon as%
\begin{equation}
\langle S,l|\left( \frac{\mathcal{M}_{l}^{\left( -\right) }\mathcal{M}%
_{l}^{\left( +\right) }}{\bar{\zeta}_{-}\bar{\zeta}_{+}\mathsf{A}_{\bar{\zeta%
}_{+},\bar{\zeta}_{-}}(\eta /2-\xi _{l})}\right) ^{1-h}\text{ with }h\in
\{0,1\},
\end{equation}%
is a basis on the local space $V_{l}\simeq \mathbb{C}^{2}$ for all $l\in
\{1,...,\mathsf{N}\}$. This is indeed the case as it holds:%
\begin{align}
& \text{det}_{2}||\left( \langle S,l|\left( \frac{\mathcal{M}_{l}^{\left(
-\right) }\mathcal{M}_{l}^{\left( +\right) }}{\bar{\zeta}_{-}\bar{\zeta}_{+}%
\mathsf{A}_{\bar{\zeta}_{+},\bar{\zeta}_{-}}(\eta /2-\xi _{l})}\right)
^{i-1}|e_{j}(l)\rangle\right) _{i,j\in \{1,2\}}||  \notag \\
& \left. =\right. 2\frac{s_{-}^{2}(\kappa _{+}e^{-\tau _{+}}-\kappa
_{-}e^{-\tau _{-}})+s_{+}^{2}(\kappa _{+}e^{\tau _{+}}-\kappa _{-}e^{\tau
_{-}})+2s_{+}s_{-}\kappa _{-}\kappa _{+}(e^{\tau _{+}-\tau _{-}}-e^{\tau
_{-}-\tau _{+}})}{\zeta _{-}\zeta _{+}\mathsf{A}_{\bar{\zeta}_{+},\bar{\zeta}%
_{-}}(\eta /2-\xi _{l})},
\end{align}%
where $|e_{j}(l)\rangle$ is the element $j\in \{1, 2\}$ of the natural basis in $V_{l}$, which can be always chosen different from zero by an appropriate choice of $%
s_{+}$ and $s_{-}$ under the condition $(\ref{Cond-NoN-S-Diag})$.

\textit{Let us now prove ii).} Let us observe that $T(%
\lambda )$ is a polynomial of degree $1$ in $1/\zeta _{-}$ with constant
term which coincides with the transfer matrix $T^{(K_{+},I)}(\lambda )$
associated to the same $K_{+}(\lambda )$ and $K_{-}(\lambda )=I$. Here, we
show that the set of co-vectors $(\ref{SoV-Basis-6v-Open})$ generated by $%
T^{(K_{+},I)}(\lambda )$ is a basis for almost any choice of the co-vector $%
\langle S|$ and of the inhomogeneity parameters satisfying $(\ref%
{Inhomog-cond})$. This implies the statement of the theorem once we recall
that the co-vectors $(\ref{SoV-Basis-6v-Open})$ generated by $%
T(\lambda )$ are polynomials of maximal degree $\mathsf{N}$ in $%
1/\zeta _{-}$.

Let us impose on the inhomogeneity parameters the same condition $(\ref%
{Relation-InH})$ then $T^{(K_{+},I)}(\xi _{l}-\eta /2)$ are polynomials of
degree $2\mathsf{N}$ in $\xi $ for all $l\in \{1,...,\mathsf{N}\}$ with
maximal degree coefficient given by:%
\begin{equation}
\frac{(-1)^{\mathsf{N}-l}\eta l(\mathsf{N}-l)!(\mathsf{N}+l)!}{\zeta _{+}}%
\left( 
\begin{array}{cc}
1 & 2\kappa e^{\tau } \\ 
2\kappa e^{-\tau } & -1%
\end{array}%
\right) .
\end{equation}%
Then for a chosen $\langle S|$ of tensor product form $(\ref{S-tensor})$ the set of
co-vectors $(\ref{SoV-Basis-6v-Open})$ generated by $T^{(K_{+},I)}(\lambda )$
is a basis as soon as%
\begin{equation}
\langle S,l|\left( 
\begin{array}{cc}
1 & 2\kappa e^{\tau } \\ 
2\kappa e^{-\tau } & -1%
\end{array}%
\right) ^{1-h}\text{ with }h\in \{0,1\},  \label{Asymp-ii}
\end{equation}%
is a basis on the local space $V_{l}\simeq \mathbb{C}^{2}$ for all $l\in
\{1,...,\mathsf{N}\}$. This is indeed the case as it holds:
\begin{align}
& \text{det}_{2}||\left( \langle S,l|\left( 
\begin{array}{cc}
1 & 2\kappa e^{\tau } \\ 
2\kappa e^{-\tau } & -1%
\end{array}%
\right) ^{i-1}|e_{j}(l)\rangle\right) _{i,j\in \{1,2\}}||  \notag \\
& \left. =\right. 2(s_{-}^{2}\kappa e^{-\tau }+s_{+}^{2}\kappa e^{\tau
}+s_{+}s_{-}),
\end{align}%
which can be always chosen different from zero for any fixed $\kappa $ and $%
\tau $ for an appropriate choice of $s_{+}$ and $s_{-}$.
\end{proof}

\subsection{Comparison of the two SoV constructions}

Here we want to show that under some special choice of the co-vector $\langle
S|$, our SoV left basis reduces to the SoV basis associated to the
generalized Sklyanin's approach when this last one is applicable, i.e. when
the two boundary matrices are non-commuting.

Let us introduce the following gauged transformed monodromy matrix: 
\begin{equation}
\left( 
\begin{array}{cc}
\mathcal{\bar{A}}_{-}(\lambda ) & \mathcal{\bar{B}}_{-}(\lambda ) \\ 
\mathcal{\bar{C}}_{-}(\lambda ) & \mathcal{\bar{D}}_{-}(\lambda )%
\end{array}%
\right) =M(\lambda )\,\bar{K}_{-}(\lambda )\,\hat{M}(\lambda
)=W_{0}^{(K_{+,-})}\mathcal{W}_{K_{+,-}}\,\mathcal{U}_{- }(\lambda )\,%
\mathcal{W}_{K_{+,-}}^{-1\,}\,(W_{0}^{(K_{+,-})})^{-1},  \label{Gaug-U}
\end{equation}%
then the associated transfer matrix: 
\begin{align}
\bar T(\lambda )& =\text{tr}_{0}\left\{ \bar{K}_{+}(\lambda
)\,M(\lambda )\,\bar{K}_{-}(\lambda )\,\hat{M}(\lambda )\right\} =\bar{c}%
_{+}(\lambda )\,\mathcal{\bar{B}}_{-}(\lambda )  \notag \\
& +\frac{(2\lambda +\eta )\,(\lambda -\frac{\eta }{2}+\bar{\zeta}_{+})\,\bar{%
\mathcal{A}}_{-}(\lambda )+(2\lambda -\eta )\,(-\lambda -\frac{\eta }{2}+%
\bar{\zeta}_{+})\,\bar{\mathcal{A}}_{-}(-\lambda )}{2\lambda \,\bar{\zeta}%
_{+}},
\end{align}%
is related to the original transfer matrix by:%
\begin{equation}
T(\lambda )=\mathcal{W}_{K_{+,-}}^{-1\,}\bar T(\lambda )\mathcal{W}_{K_{+,-}},
\end{equation}%
and the generalized Sklyanin's SoV basis can be rewritten as:%
\begin{equation}
\langle \,\mathbf{h}_{-}\,|\equiv \langle \,0\,|\prod_{n=1}^{\mathsf{N}%
}\left( \frac{\mathcal{\bar{A}}_{-}(\eta /2-\xi _{n})}{\mathsf{A}_{-}(\eta
/2-\xi _{n})}\right) ^{1-h_{n}}\mathcal{W}_{K_{+,-}}.
\end{equation}%
Here we want to show that the co-vector $\langle \,\mathbf{h}_{-}\,|$ and the
co-vector $\langle h_{1},...,h_{\mathsf{N}}|$ of our SoV basis $(\ref%
{SoV-Basis-6v-Open})$ are proportional for any $\{h_{1},...,h_{\mathsf{N}%
}\}\in \{0,1\}^{\otimes \mathsf{N}}$ when we set:%
\begin{equation}
\langle \,S\,|=\langle \,0\,|\mathcal{W}_{K_{+,-}}.
\end{equation}%
The proof is done by induction on $l=\mathsf{N}-\sum_{a=1}^{\mathsf{N}}h_{a}$, just using the identity:%
\begin{equation}
\langle 0|\mathcal{\bar{A}}_{-}(\xi _{a}-\eta /2)=0\text{ \ }\forall a\in
\{1,...,\mathsf{N}\}  \label{0-cond}
\end{equation}%
and the following reflection algebra commutation relations:%
\begin{equation}
\mathcal{\bar{A}}_{-}\left( \mu \right) \mathcal{\bar{A}}_{-}\left( \lambda
\right) =\mathcal{\bar{A}}_{-}\left( \lambda \right) \mathcal{\bar{A}}%
_{-}\left( \mu \right) +\frac{\eta }{\lambda +\mu -\eta }(\mathcal{\bar{B}}%
_{-}\left( \lambda \right) \mathcal{\bar{C}}_{-}\left( \mu \right) -\mathcal{%
\bar{B}}_{-}\left( \mu \right) \mathcal{\bar{C}}_{-}\left( \lambda \right) ).
\label{AA-BC-Req}
\end{equation}%
First, the statement is obviously true for $l=0$. Let us assume that our statement holds for any state:%
\begin{equation}
\langle \,\mathbf{h}_{-}\,|=\mathsf{N}_{\mathbf{h}}\langle h_{1},...,h_{%
\mathsf{N}}|\text{\ \ with \ }l=\mathsf{N}-\sum_{a=1}^{\mathsf{N}}h_{a}\leq 
\mathsf{N}-1,
\end{equation}%
for some given $l$. Then, let us show it for any state with $l+1$. To this aim we fix a state in
the above set and we denote with $\pi $ a permutation on the set $\{1,...,%
\mathsf{N}\}$ such that:%
\begin{equation}
h_{\pi (a)}=0\text{ for }a\leq l\text{ \ and }h_{\pi (a)}=1\text{ for }l<a\,
.
\end{equation}%
Now, let us take $c\in \{\pi (l+1),...,\pi (\mathsf{N})\}$ and let us
compute:%
\begin{align}
\langle \,\mathbf{h}_{-}\,|T(\xi _{c}^{\left( 1\right) })&
=\langle \,0\,|\prod_{n=1}^{l}\frac{\mathcal{\bar{A}}_{-}(\eta /2-\xi _{\pi
(n)})}{\mathsf{A}_{-}(\eta /2-\xi _{\pi (n)})}\bar T(\xi
_{c}^{\left( 1\right) })\mathcal{W}_{K_{+,-}} \\
& =\langle \,0\,|\prod_{n=1}^{l}\frac{\mathcal{\bar{A}}_{-}(\eta /2-\xi
_{\pi (n)})}{\mathsf{A}_{-}(\eta /2-\xi _{\pi (n)})}\left( \frac{(2\xi
_{c}^{\left( 1\right) }+\eta )\,(\xi _{c}^{\left( 1\right) }-\frac{\eta }{2}+%
\bar{\zeta}_{+})\,\bar{\mathcal{A}}_{-}(\xi _{c}^{\left( 1\right) })}{2\xi
_{c}^{\left( 1\right) }\,\bar{\zeta}_{+}}\right.  \notag \\
& +\left. \frac{(2\xi _{c}^{\left( 1\right) }-\eta )\,(-\xi _{c}^{\left(
1\right) }-\frac{\eta }{2}+\bar{\zeta}_{+})\,\bar{\mathcal{A}}_{-}(-\xi
_{c}^{\left( 1\right) })}{2\xi _{c}^{\left( 1\right) }\,\bar{\zeta}_{+}}%
\right) \mathcal{W}_{K_{+,-}},
\end{align}%
so that we have just to prove:%
\begin{equation}
\langle \,0\,|\prod_{n=1}^{l}\frac{\mathcal{\bar{A}}_{-}(\eta /2-\xi _{\pi
(n)})}{\mathsf{A}_{-}(\eta /2-\xi _{\pi (n)})}\bar{\mathcal{A}}_{-}(\xi
_{c}^{\left( 1\right) })=0.  \label{UN-W-SoV}
\end{equation}%
From the commutation relation $(\ref{AA-BC-Req})$, the above co-vector can be
rewritten as:%
\begin{eqnarray}
&&\langle \,0\,|\prod_{n=1}^{l-1}\frac{\mathcal{\bar{A}}_{-}(\eta /2-\xi
_{\pi (n)})}{\mathsf{A}_{-}(\eta /2-\xi _{\pi (n)})}\mathsf{A}_{-}^{-1}(\eta
/2-\xi _{\pi (l)})(\bar{\mathcal{A}}_{-}(\xi _{c}^{\left( 1\right) })\bar{%
\mathcal{A}}_{-}(-\xi _{\pi (l)}^{\left( 1\right) })  \notag \\
&&+\eta (\mathcal{\bar{B}}_{-}(\xi _{c}^{\left( 1\right) })\mathcal{\bar{C}}%
_{-}(-\xi _{\pi (l)}^{\left( 1\right) })-\mathcal{\bar{B}}_{-}(-\xi _{\pi
(l)}^{\left( 1\right) })\mathcal{\bar{C}}_{-}(\xi _{c}^{\left( 1\right)
}))/(\xi _{c}+\xi _{\pi (l)}-\eta )),
\end{eqnarray}%
which reduces to:%
\begin{equation}
\langle 0|\prod_{n=1}^{l-1}\frac{\mathcal{\bar{A}}_{-}(\eta /2-\xi _{\pi
(n)})}{\mathsf{A}_{-}(\eta /2-\xi _{\pi (n)})}\bar{\mathcal{A}}_{-}(\xi
_{c}^{\left( 1\right) })\frac{\bar{\mathcal{A}}_{-}(-\xi _{\pi (l)}^{\left(
1\right) })}{\mathsf{A}_{-}^{-1}(\eta /2-\xi _{\pi (l)})},
\end{equation}%
once we observe that the co-vector on the left of $\mathcal{\bar{B}}_{-}( \xi
_{c}^{\left( 1\right) }) $ and $\mathcal{\bar{B}}_{-}(\xi _{\pi (l)}^{\left(
1\right) })$ are left eigenco-vectors of $\mathcal{\bar{B}}_{-}\left( \lambda
\right) $ with eigenvalue zeros at $\lambda =\pm \xi _{\pi (l)}^{\left(
1\right) }$, $\pm \xi _{c}^{\left( 1\right) }$. That is we can commute in
the co-vector $\bar{\mathcal{A}}_{-}(-\xi _{\pi (l)}^{\left( 1\right) })$ and 
$\bar{\mathcal{A}}_{-}(\xi _{c}^{\left( 1\right) })$ and by the same
argument $\bar{\mathcal{A}}_{-}(-\xi _{\pi (l)}^{\left( 1\right) })$ and $%
\bar{\mathcal{A}}_{-}(\xi _{c}^{\left( 1\right) })$ for any $r\leq l-1$ up
to bring $\bar{\mathcal{A}}_{-}(\xi _{c}^{\left( 1\right) })$ completely to
the left acting on $\langle 0|$ which proves $(\ref{UN-W-SoV})$ as a
consequence of $(\ref{0-cond}) $.

\subsection{Transfer matrix spectrum in our SoV approach}

In our SoV basis the separate relations are given directly by the
particularization of the fusion relations at the spectrum of the separate
variables. In the case at hand these fusion relations just reduces to the
following identities:%
\begin{equation}
T(\xi _{n}^{(0)})T(\xi _{n}^{(1)})=\mathsf{A}_{\bar{%
\zeta}_{+},\bar{\zeta}_{-}}(\xi _{n}^{(0)})\mathsf{A}_{\bar{\zeta}_{+},\bar{%
\zeta}_{-}}(-\xi _{n}^{(1)}),\text{ }\forall a\in \{1,...,\mathsf{N}\},
\label{gl_2-fusion}
\end{equation}%
which are proven by direct computations using the reduction of the rational
6-vertex $R$-matrix to the permutation operator and to a 1-dimensional
projector at $\lambda =0$ and $-\eta $, respectively. To these relations now
one has to add the knowledge of the analytic properties of the transfer
matrix that we can easily derive. In fact, $T(\lambda )$ is a
polynomial of degree $2$ in all the $\xi _{a}$ and of degree $\mathsf{N}+1$
in $\lambda ^{2}$ with the following leading central coefficient:%
\begin{equation}
\lim_{\lambda \rightarrow +\infty }\lambda ^{-2(\mathsf{N}%
+1)}T(\lambda )=t_{\mathsf{N}+1}I,\text{ with }t_{\mathsf{N}+1}=%
\frac{2(1+4\kappa _{+}\kappa _{-}\cosh (\tau _{+}-\tau _{-}))}{\zeta
_{+}\zeta _{-}}=\frac{2+\bar{\mathsf{b}}_{-}\bar{\mathsf{c}}_{+}}{\bar{\zeta}%
_{+}\,\bar{\zeta}_{-}},
\end{equation}%
whose values in $\pm \eta /2$ are central: 
\begin{equation}
T(\pm \eta /2)=2(-1)^{\mathsf{N}}\mathrm{det}_{q}M(0)\equiv
t(\eta /2),
\end{equation}%
with $\mathrm{det}_{q}M(\lambda )=a(\lambda +\eta /2)\,d(\lambda -\eta /2)$.
Let us define the following set of functions:%
\begin{eqnarray}
r_{a,\mathbf{h}}(\lambda ) &=&\frac{\lambda ^{2}-\left( \eta /2\right) ^{2}}{%
(\xi _{a}^{(h_{a})})^{2}-\left( \eta /2\right) ^{2}}\prod_{b\neq a,b=1}^{%
\mathsf{N}}\frac{\lambda ^{2}-(\xi _{b}^{(h_{b})})^{2}}{(\xi
_{a}^{(h_{a})})^{2}-(\xi _{b}^{(h_{b})})^{2}}\ , \\
s_{\mathbf{h}}(\lambda ) &=&\prod_{b=1}^{\mathsf{N}}\frac{\lambda ^{2}-(\xi
_{b}^{(h_{b})})^{2}}{(\eta /2)^{2}-(\xi _{b}^{(h_{b})})^{2}}\ , \\
u_{\mathbf{h}}(\lambda ) &=&(\lambda ^{2}-(\eta /2)^{2})\prod_{b=1}^{\mathsf{%
N}}(\lambda ^{2}-(\xi _{b}^{(h_{b})})^{2})\ ,
\end{eqnarray}%
then the following theorem holds:

\begin{theorem}
\label{SoV-Sp-Ch-gl2-R}Under the same conditions of Theorem \ref%
{Th-SoV-basis}, ensuring the existence of the left SoV basis, the spectrum
of $T(\lambda )$ is characterized by:%
\begin{equation}
\Sigma _{T}=\left\{ t(\lambda ):t(\lambda )=t_{\mathsf{N}+1}u_{%
\mathbf{h}=0}(\lambda )+t(\eta /2)s_{\mathbf{h}=0}(\lambda )+\sum_{a=1}^{%
\mathsf{N}}r_{a,\mathbf{h}=0}(\lambda )x_{a},\text{ \ \ }\forall
\{x_{1},...,x_{\mathsf{N}}\}\in S_{T}\right\} ,
\end{equation}%
$S_{T}$ is the set of solutions to the following system of \textsf{$N$}
quadratic equations:%
\begin{equation}
x_{n}[t_{\mathsf{N}+1}u_{\mathbf{h}=0}(\xi _{n}^{(1)})+t(\eta /2)s_{\mathbf{h%
}=0}(\xi _{n}^{(1)})+\sum_{a=1}^{\mathsf{N}}r_{a,\mathbf{h}=0}(\xi
_{n}^{(1)})x_{a}]=\mathsf{A}_{\bar{\zeta}_{+},\bar{\zeta}_{-}}(\xi
_{n}^{(0)})\mathsf{A}_{\bar{\zeta}_{+},\bar{\zeta}_{-}}(-\xi _{n}^{(1)}),%
\text{ }\forall n\in \{1,...,\mathsf{N}\},  \label{Quadratic System-open}
\end{equation}%
in $\mathsf{N}$ unknown $\{x_{1},...,x_{\mathsf{N}}\}$. Moreover, $%
T(\lambda )$ has simple spectrum and for any $t(\lambda )\in
\Sigma _{T}$ the associated unique (up-to normalization)
eigenvector $|t\rangle $ has the following separated wave-function in the left SoV
basis:%
\begin{equation}
\langle h_{1},...,h_{\mathsf{N}}|t\rangle =\prod_{n=1}^{\mathsf{N}}\left( 
\frac{t(\xi _{n}-\eta /2)}{\mathsf{A}_{\bar{\zeta}_{+},\bar{\zeta}_{-}}(\eta
/2-\xi _{n})}\right) ^{1-h_{n}}.  \label{SoV-Ch-T-eigenV-open}
\end{equation}
\end{theorem}

\begin{proof}
The system of $\mathsf{N}$ quadratic equations $(\ref{Quadratic System-open}%
) $ in $\mathsf{N}$ unknown $\{x_{1},...,x_{\mathsf{N}}\}$ is nothing else but 
the rewriting of the transfer matrix fusion equations for the eigenvalues.
Any transfer matrix eigenvalue is then a solution of this system and the
associated right eigenvector $|t\rangle $ admits the characterization $(\ref%
{SoV-Ch-T-eigenV-open})$ in our left SoV basis. Let us now prove the reverse statement. This is done by proving that any polynomial $t(\lambda )$
satisfying this system is an eigenvalue. For this, we prove that the vector $%
|t\rangle $ characterized by $(\ref{SoV-Ch-T-eigenV-open})$ is a transfer
matrix eigenvector, namely:%
\begin{equation}
\langle h_{1},...,h_{\mathsf{N}}|T(\lambda )|t\rangle =t(\lambda
)\langle h_{1},...,h_{\mathsf{N}}|t\rangle ,\text{ }\forall \{h_{1},...,h_{%
\mathsf{N}}\}\in \{0,1\}^{\otimes \mathsf{N}}.
\end{equation}%
Let us write the following interpolation formula for the transfer matrix:%
\begin{equation}
T(\lambda )=t_{\mathsf{N}+1}u_{\mathbf{h}}(\lambda )+t(\eta
/2)s_{\mathbf{h}}(\lambda )+\sum_{a=1}^{\mathsf{N}}r_{a,\mathbf{h}}(\lambda
)T(\xi _{a}^{\left( h_{a}\right) }),  \label{Interp-T-open}
\end{equation}%
and use it to act on the generic element of the left SoV basis. Then, we
have:%
\begin{equation}
\langle h_{1},...,h_{a},...,h_{\mathsf{N}}|T(\xi _{a}^{\left(
h_{a}\right) })|t\rangle =\left\{ 
\begin{array}{l}
\mathsf{A}_{\bar{\zeta}_{+},\bar{\zeta}_{-}}(\eta /2-\xi _{n})\langle
h_{1},...,h_{a}^{\prime }=0,...,h_{\mathsf{N}}|t\rangle \text{ \ \ if \ }%
h_{a}=1 \\ 
\frac{\mathsf{A}_{\bar{\zeta}_{+},\bar{\zeta}_{-}}(\xi _{n}^{(0)})\mathsf{A}%
_{\bar{\zeta}_{+},\bar{\zeta}_{-}}(-\xi _{n}^{(1)})}{\mathsf{A}_{\bar{\zeta}%
_{+},\bar{\zeta}_{-}}(-\xi _{n}^{(1)})}\langle h_{1},...,h_{a}^{\prime
}=1,...,h_{\mathsf{N}}|t\rangle \text{ \ if \ }h_{a}=0%
\end{array}%
\right.
\end{equation}%
which by the definition of the state $|t\rangle $ can be rewritten as:%
\begin{equation}
\langle h_{1},...,h_{a},...,h_{\mathsf{N}}|T(\xi _{a}^{\left(
h_{a}\right) })|t\rangle =\left\{ 
\begin{array}{l}
t(\xi _{a}^{\left( 1\right) })\prod_{n\neq a,n=1}^{\mathsf{N}}\left( \frac{%
t(\xi _{n}-\eta /2)}{\mathsf{A}_{\bar{\zeta}_{+},\bar{\zeta}_{-}}(\eta
/2-\xi _{n})}\right) ^{1-h_{n}}\text{ \ \ if \ }h_{a}=1 \\ 
\frac{\mathsf{A}_{\bar{\zeta}_{+},\bar{\zeta}_{-}}(\xi _{n}^{(0)})\mathsf{A}%
_{\bar{\zeta}_{+},\bar{\zeta}_{-}}(-\xi _{n}^{(1)})\prod_{n\neq a,n=1}^{%
\mathsf{N}}\left( \frac{t(\xi _{n}-\eta /2)}{\mathsf{A}_{\bar{\zeta}_{+},%
\bar{\zeta}_{-}}(\eta /2-\xi _{n})}\right) ^{1-h_{n}}}{\mathsf{A}_{\bar{\zeta%
}_{+},\bar{\zeta}_{-}}(-\xi _{n}^{(1)})}\,\text{if}\,h_{a}=0%
\end{array}%
\right.
\end{equation}%
and finally, by the fusion equation satisfied by the $t(\lambda )$, it reads:%
\begin{equation}
\langle h_{1},...,h_{a},...,h_{\mathsf{N}}|T(\xi _{a}^{\left(
h_{a}\right) })|t\rangle =\left\{ 
\begin{array}{l}
t(\xi _{a}^{\left( 1\right) })\prod_{n\neq a,n=1}^{\mathsf{N}}\left( \frac{%
t(\xi _{n}-\eta /2)}{\mathsf{A}_{\bar{\zeta}_{+},\bar{\zeta}_{-}}(\eta
/2-\xi _{n})}\right) ^{1-h_{n}}\text{ \ \ \ if \ }h_{a}=1 \\ 
t(\xi _{a}^{\left( 0\right) })\prod_{n=1}^{\mathsf{N}}\left( \frac{t(\xi
_{n}-\eta /2)}{\mathsf{A}_{\bar{\zeta}_{+},\bar{\zeta}_{-}}(\eta /2-\xi _{n})%
}\right) ^{1-h_{n}}\text{ \ \ if \ }h_{a}=0%
\end{array}%
\right. ,
\end{equation}%
and so:%
\begin{equation}
\langle h_{1},...,h_{a},...,h_{\mathsf{N}}|T(\xi _{a}^{\left(
h_{a}\right) })|t\rangle =t(\xi _{a}^{\left( h_{a}\right) })\langle
h_{1},...,h_{a},...,h_{\mathsf{N}}|t\rangle ,
\end{equation}%
from which we have:%
\begin{equation}
\langle h_{1},...,h_{\mathsf{N}}|T(\lambda )|t\rangle =\left( t_{%
\mathsf{N}+1}u_{\mathbf{h}}(\lambda )+t(\eta /2)s_{\mathbf{h}}(\lambda
)+\sum_{a=1}^{\mathsf{N}}r_{a,\mathbf{h}}(\lambda )t(\xi _{a}^{\left(
h_{a}\right) })\right) \langle h_{1},...,h_{\mathsf{N}}|t\rangle ,
\end{equation}%
proving our statement.
\end{proof}

The previous characterization of the spectrum allows to introduce a
functional equation characterization of it, the so-called quantum spectral
curve equation. This is in the current case a second order Baxter's type difference
equation. In particular, this result coincides with the Theorem 3.2 of  
\cite{KitMNT17}, the only difference being on the applicability of the
result that is now extended to the case of commuting boundary matrices.

\begin{theorem}
Under the same conditions of Theorem \ref{Th-SoV-basis}, ensuring the
existence of the left SoV basis, an entire function $t(\lambda )\in \Sigma
_{T}$ if and only if there exists a unique polynomial 
\begin{equation}
Q_{t}(\lambda )=\prod_{b=1}^{p_{K_{+,-}}}\left( \lambda ^{2}-\lambda
_{b}^{2}\right) ,\qquad \lambda _{1},\ldots ,\lambda _{p_{K_{+,-}}}\in 
\mathbb{C}\setminus \{\pm \xi _{1}^{(0)},\ldots ,\pm \xi _{\mathsf{N}%
}^{(0)}\},
\end{equation}%
such that 
\begin{equation}
t(\lambda )\,Q_{t}(\lambda )=\mathsf{A}_{\bar{\zeta}_{+},\bar{\zeta}%
_{-}}(\lambda )\,Q_{t}(\lambda -\eta )+\mathsf{A}_{\bar{\zeta}_{+},\bar{\zeta%
}_{-}}(-\lambda )\,Q_{t}(\lambda +\eta )+F(\lambda ),  \label{InH-BAX-eq}
\end{equation}%
with 
\begin{equation}
F(\lambda )=\frac{\bar{\mathsf{b}}_{-}\bar{\mathsf{c}}_{+}}{\bar{\zeta}_{-}%
\bar{\zeta}_{+}}\left( \lambda ^{2}-\left( \eta /2\right) ^{2}\right)
\,\prod_{b=1}^{\mathsf{N}}\prod_{h=0}^{1}\left( \lambda ^{2}-(\xi
_{b}^{(h)})^{2}\right) .  \label{InH-F}
\end{equation}%
Similarly, an entire function $t(\lambda )\in \Sigma _{T}$ if
and only if there exists a unique polynomial%
\begin{equation}
P_{t}(\lambda )=\prod_{b=1}^{q_{K_{+,-}}}\left( \lambda ^{2}-\mu
_{b}^{2}\right) ,\qquad \mu _{1},\ldots ,\mu _{q_{K_{+,-}}}\in \mathbb{C}%
\setminus \{\pm \xi _{1}^{(0)},\ldots ,\pm \xi _{\mathsf{N}}^{(0)}\},
\label{P-poly}
\end{equation}%
such that 
\begin{equation}
t(\lambda )\,P_{t}(\lambda )=\mathsf{A}_{-\bar{\zeta}_{+},-\bar{\zeta}%
_{-}}(\lambda )\,P_{t}(\lambda -\eta )+\mathsf{A}_{-\bar{\zeta}_{+},-\bar{%
\zeta}_{-}}(-\lambda )\,P_{t}(\lambda +\eta )+F(\lambda ).
\label{InH-BAX-eq-+}
\end{equation}%
Here, it holds%
\begin{eqnarray}
p_{K_{+,-}} &=&(1-\delta _{0,\bar{\mathsf{b}}_{-}\bar{\mathsf{c}}_{+}})%
\mathsf{N}+\delta _{0,\bar{\mathsf{b}}_{-}\bar{\mathsf{c}}_{+}}\ p \\
q_{K_{+,-}} &=&(1-\delta _{0,\bar{\mathsf{b}}_{-}\bar{\mathsf{c}}_{+}})%
\mathsf{N}++\delta _{0,\bar{\mathsf{b}}_{-}\bar{\mathsf{c}}_{+}}\ q
\end{eqnarray}%
with $p$ and $q$ non negative integers such that%
\begin{equation}
p+q=\mathsf{N}\text{\textsf{,}}
\end{equation}%
and the following Wronskian equation is satisfied in the case $\bar{\mathsf{b%
}}_{-}\bar{\mathsf{c}}_{+}=0$:%
\begin{align}
2(-1)^{\mathsf{N}}({\bar{\zeta}}_{+}+{\bar{\zeta}}_{-}+(p-q)\eta )(\lambda
-\eta /2)\,a(\lambda )\,d(\lambda )& =(\lambda -\frac{\eta }{2}+\bar{\zeta}%
_{+})(\lambda -\frac{\eta }{2}+\bar{\zeta}_{-})\,Q_{t}(\lambda -\eta
)\,P_{t}(\lambda )  \notag \\
& -(\lambda -\frac{\eta }{2}-\bar{\zeta}_{+})(\lambda -\frac{\eta }{2}-\bar{%
\zeta}_{-})\,Q_{t}(\lambda )\,P_{t}(\lambda -\eta ).  \label{W-eq}
\end{align}
\end{theorem}

\begin{proof}
Under the conditions ensuring the existence of the left SoV basis, the
equivalence of the first discrete SoV characterization with these functional
equations is proven by the standard arguments as introduced in \cite{Nic10a,KitMN14}, see e.g. the proof of Theorem 2.3 of \cite{KitMNT17}.
\end{proof}

\subsection{Diagonalizability and simplicity of the transfer matrix}

Our SoV approach implies that the transfer matrix spectrum is simple as soon as
our SoV basis can be constructed. Here, we show that the transfer matrix is
indeed diagonalizable with simple spectrum for generic values of boundary
parameters. One can adapt the general Proposition 2.5 of \cite{MaiN18} to
the present case and in fact this result is just a special
case of the general Theorem \ref{General-SoV-basis} of Section 5, derived for the
fundamental representations of the $gl_n$ reflection algebra. However, in
this section we present a slightly different proof based on the explicit
form of the transfer matrix scalar product formula, as re-derived in the appendix
A within  the current SoV framework.

\begin{theorem}
Let $\langle t|$ and $|t\rangle $\ be the unique eigenco-vector and
eigenvector associated to any fixed eigenvalue $t(\lambda )$ of $%
T(\lambda )$, then it holds:%
\begin{equation}
\langle t|t\rangle \neq 0,
\end{equation}%
and $T(\lambda )$ is diagonalizable with simple spectrum almost
for any value of $\eta $ and of the inhomogeneities satisfying $(\ref%
{Inhomog-cond})$ in the following two general cases:

i) $K_{-}(\lambda )$ and $K_{+}(\lambda )$ are non-commutative boundary
matrices $(\ref{Cond-S-Diag})$ while $\mathcal{M}_{l}^{\left( -\right) }%
\mathcal{M}_{l}^{\left( +\right) }$ is diagonalizable with simple spectrum%
\footnote{%
Note that this is the case for any fixed value of the boundary parameters $%
\kappa _{\epsilon },$ $\tau _{\epsilon }$ and $\tau _{-\epsilon }$ up two
values of $\kappa _{-\epsilon }$, for $\epsilon \in \{-1,1\}$.}.

ii) $K_{-}(\lambda )$ and $K_{+}(\lambda )$ are simultaneously
diagonalizables, i.e. it holds $(\ref{Cond-S-Diag})$ and%
\begin{equation}
\kappa ^{2}\neq -1/4,  \label{k-condi-simple}
\end{equation}%
for any fixed choice of the boundary parameters $\{\zeta _{-\epsilon
},\kappa ,\tau \}$ and for almost any value of $\zeta _{\epsilon }$, with $%
\epsilon \in \{-1,1\}$.
\end{theorem}

\begin{proof}
Let us denote with:%
\begin{equation}
K=\left\{ 
\begin{array}{l}
\mathcal{M}^{\left( -\right) }\mathcal{M}^{\left( +\right) }\text{ \ in the
case i)} \\ 
\\ 
\left( 
\begin{array}{cc}
1 & 2\kappa e^{\tau } \\ 
2\kappa e^{-\tau } & -1%
\end{array}%
\right) \text{ \ in the case ii)}%
\end{array}%
\right.
\end{equation}%
and let us denote with $\mathsf{k}_{0}$ and $\mathsf{k}_{1}$ the associated
eigenvalues. Then in the case i) the matrix $K$ is diagonalizable and with
simple spectrum by assumption while in the case ii) the requirement $(\ref%
{k-condi-simple})$ implies that $K$ is diagonalizable with simple spectrum
as it holds:%
\begin{equation}
\mathsf{k}_{1}=-\mathsf{k}_{0}=\sqrt{1+4\kappa ^{2}}\neq 0.
\end{equation}%
We can now proceed to compute the scalar product:%
\begin{align}
\langle \,t\,|\,t\,\rangle & =\sum_{h_{1},...,h_{\mathsf{N}%
}=0}^{1}\prod_{a=1}^{\mathsf{N}}\left( \frac{\xi _{a}-\eta }{\xi _{a}+\eta }%
\frac{t(\xi _{a}+\eta /2)}{\,\mathsf{A}_{\bar{\zeta}_{+},\bar{\zeta}%
_{-}}(\eta /2-\xi _{n})}\right) ^{\!h_{a}}(\frac{t(\xi _{a}-\eta /2)}{%
\mathsf{A}_{\bar{\zeta}_{+},\bar{\zeta}_{-}}(\eta /2-\xi _{n})})^{1-h_{a}}%
\frac{\hat{V}(\xi _{1}^{(h_{1})},...,\xi _{\mathsf{N}}^{(h_{\mathsf{N}})})}{%
\hat{V}(\xi _{1},...,\xi _{\mathsf{N}})} \\
& =\sum_{h_{1},...,h_{\mathsf{N}}=0}^{1}\prod_{a=1}^{\mathsf{N}}\left( \frac{%
\xi _{a}-\eta }{\xi _{a}+\eta }\frac{\mathsf{A}_{\bar{\zeta}_{+},\bar{\zeta}%
_{-}}(\xi _{a}+\eta /2)}{t(\xi _{a}-\eta /2)}\right) ^{\!h_{a}}(\frac{t(\xi
_{a}-\eta /2)}{\mathsf{A}_{\bar{\zeta}_{+},\bar{\zeta}_{-}}(\eta /2-\xi _{n})%
})^{1-h_{a}}\frac{\hat{V}(\xi _{1}^{(h_{1})},...,\xi _{\mathsf{N}}^{(h_{%
\mathsf{N}})})}{\hat{V}(\xi _{1},...,\xi _{\mathsf{N}})},
\end{align}%
then the leading coefficient of $\langle \,t\,|\,t\,\rangle $ in $\xi $ is
given by the following limit:%
\begin{equation}
\lim_{\xi \rightarrow \infty }\text{ }\langle \,t\,|\,t\,\rangle =\lim_{\xi
\rightarrow \infty }\sum_{h_{1},...,h_{\mathsf{N}}=0}^{1}\prod_{a=1}^{%
\mathsf{N}}\left( \frac{\mathsf{A}_{\bar{\zeta}_{+},\bar{\zeta}_{-}}(\xi
_{a}+\eta /2)}{t(\xi _{a}-\eta /2)}\right) ^{\!h_{a}}(\frac{t(\xi _{a}-\eta
/2)}{\mathsf{A}_{\bar{\zeta}_{+},\bar{\zeta}_{-}}(\eta /2-\xi _{n})}%
)^{1-h_{a}},
\end{equation}%
once we impose on the inhomogeneity parameters the condition $(\ref%
{Relation-InH})$. Let us now distinguish between the two cases.

In the case i), let us first observe that it holds:%
\begin{equation}
\text{det}\mathcal{M}^{\left( -\right) }\mathcal{M}^{\left( +\right) }=\text{%
det}\mathcal{M}^{\left( -\right) }\text{det}\mathcal{M}^{\left( +\right)
}=(-1)(-1)=1
\end{equation}%
and so%
\begin{equation}
\mathsf{k}_{0}\mathsf{k}_{1}=1,
\end{equation}%
taking that into account we have:%
\begin{eqnarray}
\lim_{\xi \rightarrow \infty }\text{ }\langle t|t\rangle
&=&\sum_{h_{1},...,h_{\mathsf{N}}=0}^{1}\prod_{a=1}^{\mathsf{N}}(\frac{a}{%
\mathsf{k}_{\theta _{a}}})^{h_{a}}(-a\mathsf{k}_{\theta _{a}})^{1-h_{a}}=%
\mathsf{N}!\prod_{a=1}^{\mathsf{N}}(\mathsf{k}_{\theta _{a}}^{-1}-\mathsf{k}%
_{\theta _{a}}) \\
&=&\mathsf{N}!\prod_{a=1}^{\mathsf{N}}(\mathsf{k}_{\hat{\theta}_{a}}-\mathsf{%
k}_{\theta _{a}}) \\
&=&\mathsf{N}!(\mathsf{k}_{0}-\mathsf{k}_{1})^{\mathsf{N}}\prod_{a=1}^{%
\mathsf{N}}\left( 1-2\delta _{\theta _{a},0}\right) \neq 0,
\end{eqnarray}%
where we have defined $\hat{\theta}_{a}=\{0,1\}\backslash \theta _{a}$ for
any $a\in \{1,...,\mathsf{N}\}$ and the $\{\theta _{1},...,\theta _{%
\mathsf{N}}\}\in \{0,1\}^{\mathsf{N}}$ are uniquely fixed by:%
\begin{equation}
\frac{(-1)^{\mathsf{N}-l}\eta l(\mathsf{N}-l)!(\mathsf{N}+l)!}{\bar{\zeta}%
_{+}\bar{\zeta}_{-}}\mathsf{k}_{\theta _{a}}=\lim_{\xi \rightarrow \infty
}\xi ^{-(2\mathsf{N}+1)}t(\xi _{a}-\eta /2).
\end{equation}%
Here, we have used that $T(\xi _{a}-\eta /2)$ are polynomials of
degree $2\mathsf{N}+1$ in $\xi $ for all $a\in \{1,...,\mathsf{N}\}$
with maximal degree coefficient given by $\left( \ref{Asymp-i}\right) $.

In the case ii), the scalar product $\langle \,t\,|\,t\,\rangle $ is
computed for the eigenstates associated to the transfer matrix $%
T^{(K_{+},I)}(\lambda )$. That is we have to take first the limit $\bar{\zeta%
}_{-}\rightarrow \infty $, then the leading coefficient of $\langle
\,t\,|\,t\,\rangle $ in $\xi $ is given by the following limit:%
\begin{equation}
\lim_{\xi \rightarrow \infty }\text{ }\langle \,t\,|\,t\,\rangle =\lim_{\xi
\rightarrow \infty }\sum_{h_{1},...,h_{\mathsf{N}}=0}^{1}\prod_{a=1}^{%
\mathsf{N}}\left( \frac{\mathsf{A}_{\bar{\zeta}_{+},\infty }(\xi _{a}+\eta
/2)}{t(\xi _{a}-\eta /2)}\right) ^{\!h_{a}}(\frac{t(\xi _{a}-\eta /2)}{%
\mathsf{A}_{\bar{\zeta}_{+},\infty }(\eta /2-\xi _{n})})^{1-h_{a}},
\end{equation}%
where%
\begin{equation}
\mathsf{A}_{\bar{\zeta}_{+},\infty }(\lambda )\equiv \lim_{\bar{\zeta}%
_{-}\rightarrow \infty }\mathsf{A}_{\bar{\zeta}_{+},\bar{\zeta}_{-}}(\eta
/2\pm \xi _{n})=(-1)^{\mathsf{N}}\frac{2\lambda +\eta }{2\lambda }\,\frac{%
(\lambda -\frac{\eta }{2}+\bar{\zeta}_{+})}{\bar{\zeta}_{+}\,}\,a(\lambda
)\,d(-\lambda ).
\end{equation}%
The following identities hold:%
\begin{equation}
\frac{(-1)^{\mathsf{N}-l}\eta l(\mathsf{N}-l)!(\mathsf{N}+l)!}{\zeta _{+}}%
\mathsf{k}_{\theta _{a}}=\lim_{\xi \rightarrow \infty }\xi ^{-2\mathsf{N}%
}t(\xi _{a}-\eta /2),
\end{equation}%
as $T^{(K_{+},I)}(\xi _{l}-\eta /2)$ are polynomials of degree $2\mathsf{N}$
in $\xi $ for all $l\in \{1,...,\mathsf{N}\}$ with maximal degree
coefficient given by $\left( \ref{Asymp-ii}\right) $.

So that taking now the limit $\xi \rightarrow \infty $, we get:%
\begin{eqnarray}
\lim_{\xi \rightarrow \infty }\text{ }\langle t|t\rangle
&=&\sum_{h_{1},...,h_{\mathsf{N}}=0}^{1}\prod_{a=1}^{\mathsf{N}}(\frac{\sqrt{%
1+4\kappa ^{2}}}{\mathsf{k}_{\theta _{a}}})^{h_{a}}(\frac{\mathsf{k}_{\theta
_{a}}}{\sqrt{1+4\kappa ^{2}}})^{1-h_{a}} \\
&=&\sum_{h_{1},...,h_{\mathsf{N}}=0}^{1}\prod_{a=1}^{\mathsf{N}}(\frac{%
\mathsf{k}_{1}}{\mathsf{k}_{\theta _{a}}})^{h_{a}}(\frac{\mathsf{k}_{\theta
_{a}}}{\mathsf{k}_{1}})^{1-h_{a}} \\
&=&2^{\mathsf{N}}\prod_{a=1}^{\mathsf{N}}\left( 1-2\delta _{\theta
_{a},0}\right) \neq 0.
\end{eqnarray}%
This proves that%
\begin{equation}
\langle t|t\rangle \neq 0
\end{equation}%
for almost any values of the inhomogeneities, of $\eta $ and for any choice of
the transfer matrix eigenvalue $t(\lambda )$. Finally, given an eigenvalue $%
t(\lambda )$ it is associated with a non trivial Jordan block if and only if
the eigenco-vector and eigenvector associated to $t(\lambda )$ are
orthogonal. Since we have shown that this is not the case, it implies
that the transfer matrix is diagonalizable and has simple spectrum as already proven.
\end{proof}

\section{SoV for fundamental representations of $U_{q}(\widehat{gl}_{2})$
reflection algebra}

In this section we consider the representation of the reflection algebra
associated to the trigonometric 6-vertex $R$-matrix:%
\begin{equation}
R_{ab}(\lambda )=\left( 
\begin{array}{cccc}
\sinh (\lambda +\eta ) & 0 & 0 & 0 \\ 
0 & \sinh \lambda & \sinh \eta & 0 \\ 
0 & \sinh \eta & \sinh \lambda & 0 \\ 
0 & 0 & 0 & \sinh (\lambda +\eta )%
\end{array}%
\right) \in \mathrm{End}(V_{a}\otimes V_{b}).
\end{equation}%
As in the rational case also in the trigonometric case the transfer matrix,
defined by%
\begin{equation}
T(\lambda )=\text{tr}_{a}\{K_{+,a}(\lambda
)\,M_{a}(\lambda )\,K_{-,a}(\lambda )\,\hat{M}_{a}(\lambda )\}=\text{tr}%
_{0}\left\{ K_{+,a}(\lambda )\,\mathcal{U}_{-,a}(\lambda )\right\} ,
\label{T-open+Trig}
\end{equation}%
generates a one-parameter family of commuting operators on the quantum space 
$\mathcal{H}=\otimes _{i-1}^{\mathsf{N}}V_{i}$, with $V_{i}\simeq \mathbb{C}%
^{2}$. Here, we have defined%
\begin{equation}
K_{+}(\lambda )=K_{-}(\lambda +\eta ;\zeta _{+},\kappa _{+},\tau _{+}),
\end{equation}%
and%
\begin{equation}
K_{-,a}(\lambda ;\zeta ,\kappa ,\tau )=\frac{1}{\sinh \zeta }\left( 
\begin{array}{cc}
\sinh (\lambda -\eta /2+\zeta ) & \kappa e^{\tau }\sinh (2\lambda -\eta ) \\ 
\kappa e^{-\tau }\sinh (2\lambda -\eta ) & \sinh (\zeta -\lambda +\eta /2)%
\end{array}%
\right) \in \mathrm{End}(V_{a}\simeq \mathbb{C}^{2}),
\end{equation}%
which is the most general scalar solution to the trigonometric 6-vertex
reflection equation \cite{deVG93,deVG94,GhoZ94a,GhoZ94b}. The same definitions as in the
rational case hold for the boundary monodromy matrix%
\begin{equation}
\mathcal{U}_{-,a}(\lambda )=M_{0}(\lambda )\,K_{-}(\lambda )\,\hat{M}%
_{0}(\lambda )=\left( 
\begin{array}{cc}
\mathcal{A}_{-}(\lambda ) & \mathcal{B}_{-}(\lambda ) \\ 
\mathcal{C}_{-}(\lambda ) & \mathcal{D}_{-}(\lambda )%
\end{array}%
\right) \in \mathrm{End}(V_{a}\otimes \mathcal{H}),
\end{equation}%
an operator solution to the same reflection equation, for the bulk monodromy
matrix:%
\begin{equation}
M_{0}(\lambda )=R_{0\mathsf{N}}(\lambda -\xi _{\mathsf{N}}^{(0)})\dots
R_{01}(\lambda -\xi _{1}^{(0)})=\left( 
\begin{array}{cc}
A(\lambda ) & B(\lambda ) \\ 
C(\lambda ) & D(\lambda )%
\end{array}%
\right) \in \mathrm{End}(V_{a}\otimes \mathcal{H}),
\end{equation}%
an operator solution of the trigonometric 6-vertex Yang-Baxter equation, and
for%
\begin{equation}
\hat{M}_{0}(\lambda )=(-1)^{\mathsf{N}}\,\sigma
_{0}^{y}\,M_{0}^{t_{0}}(-\lambda )\,\sigma _{0}^{y}.
\end{equation}%
For the trigonometric 6-vertex reflection algebra the fusion of transfer
matrices leads to the following quantum determinant relations:%
\begin{equation}
T(\xi _{n}^{(0)})T(\xi _{n}^{(1)})=\mathsf{A}%
_{\alpha _{\pm },\beta _{\pm }}(\xi _{n}^{(0)})\mathsf{A}_{\alpha _{\pm
},\beta _{\pm }}(-\xi _{n}^{(1)}),\text{ }\forall a\in \{1,...,\mathsf{N}\},
\end{equation}%
which are proven by direct computations using the reduction of the
trigonometric 6-vertex $R$-matrix to the permutation operator and to a
1-dimensional projector for $\lambda =0$ and $-\eta $, respectively. Here,
we have defined:%
\begin{align}
\mathsf{A}_{\alpha _{\pm },\beta _{\pm }}(\lambda )& =(-1)^{\mathsf{N}}\frac{%
\sinh (2\lambda +\eta )}{\sinh 2\lambda }g_{+}(\lambda )g_{-}(\lambda
)a(\lambda )d(-\lambda ),  \label{Coef-6v-Req-T} \\
d(\lambda )& =a(\lambda -\eta ),\text{ \ \ }a(\lambda )=\prod_{n=1}^{\mathsf{%
N}}\sinh (\lambda -\xi _{n}+\eta /2),
\end{align}%
and%
\begin{equation}
g_{\pm }(\lambda )=\left\{ 
\begin{array}{l}
\sinh (\lambda +\alpha _{\pm }-\eta /2)\cosh (\lambda \mp \beta _{\pm }-\eta
/2)/(\sinh \alpha _{\pm }\cosh \beta _{\pm })\text{ if }\kappa _{\pm }\neq 0
\\ 
\sinh (\lambda +\zeta _{\pm }-\eta /2)/\sinh \zeta _{\pm }\text{ \ \ if }%
\kappa _{\pm }=0,%
\end{array}%
\right. ,
\end{equation}%
where $\alpha _{\pm }$\ and $\beta _{\pm }$ are defined in terms of the
boundary parameters by:%
\begin{equation}
\sinh \alpha _{\pm }\cosh \beta _{\pm }=\frac{\sinh \zeta _{\pm }}{2\kappa
_{\pm }},\text{ \ \ \ \ \ }\cosh \alpha _{\pm }\sinh \beta _{\pm }=\frac{%
\cosh \zeta _{\pm }}{2\kappa _{\pm }}.
\end{equation}%
Moreover, the transfer matrix is an even function of the spectral parameter $%
\lambda $ and it is central in the following values: 
\begin{align}
\lim_{\lambda \rightarrow \pm \infty }e^{\mp 2\lambda (\mathsf{N}%
+2)}T(\lambda ) &=2^{-(2\mathsf{N}+1)}\frac{\kappa _{+}\kappa
_{-}\cosh (\tau _{+}-\tau _{-})}{\sinh \zeta _{+}\sinh \zeta _{-}},
\label{Asymp-Cent-trigo} \\
T(\pm \eta /2) &=(-1)^{\mathsf{N}}2\cosh \eta \,a(\eta
/2)d(-\eta /2),  \label{C-1} \\
T(\pm (\eta /2-i\pi /2)) &=-2\cosh \eta \coth \zeta _{-}\coth
\zeta _{+}a(i\pi /2+\eta /2)d(i\pi /2-\eta /2).  \label{C-2}
\end{align}%
Let us define the functions:%
\begin{equation}
g_{a,\mathbf{h}}(\lambda )=\frac{\cosh ^{2}2\lambda -\cosh ^{2}\eta }{\cosh
^{2}2\xi _{a}^{(h_{a})}-\cosh ^{2}\eta }\,\prod_{\substack{ b=1  \\ b\neq a}}%
^{\mathsf{N}}\frac{\cosh 2\lambda -\cosh 2\xi _{b}^{(h_{b})}}{\cosh 2\xi
_{a}^{(h_{a})}-\cosh 2\xi _{b}^{(h_{b})}}\quad \text{ \ for }a\in \{1,...,%
\mathsf{N}\},
\end{equation}%
and%
\begin{align}
f_{\mathbf{h}}(\lambda )=& \frac{\cosh 2\lambda +\cosh \eta }{2\cosh \eta }%
\prod_{b=1}^{\mathsf{N}}\frac{\cosh 2\lambda -\cosh 2\xi _{b}^{(h_{b})}}{%
\cosh \eta -\cosh 2\xi _{b}^{(h_{b})}}\mathsf{A}_{\alpha _{\pm },\beta _{\pm
}}(\eta /2)  \notag \\
& -(-1)^{\mathsf{N}}\frac{\cosh 2\lambda -\cosh \eta }{2\cosh \eta }%
\prod_{b=1}^{\mathsf{N}}\frac{\cosh 2\lambda -\cosh 2\xi _{b}^{(h_{b})}}{%
\cosh \eta +\cosh 2\xi _{b}^{(h_{b})}}\mathsf{A}_{\alpha _{\pm },\beta _{\pm
}}(\eta /2+i\pi /2)  \notag \\
& +2^{(1-\mathsf{N})}\frac{\kappa _{+}\kappa _{-}\cosh (\tau _{+}-\tau _{-})%
}{\sinh \zeta _{+}\sinh \zeta _{-}}(\cosh ^{2}2\lambda -\cosh ^{2}\eta
)\prod_{b=1}^{\mathsf{N}}(\cosh 2\lambda -\cosh 2\xi _{b}^{(h_{b})}),
\label{f-function}
\end{align}%
Then for any $\mathbf{h=}\{h_{1},...,h_{\mathsf{N}}\}\in \{0,1\}^{\mathsf{N}%
} $ the following interpolation formula for the transfer matrix holds:%
\begin{equation}
T(\lambda )=f_{\mathbf{h}}(\lambda )+\sum_{a=1}^{\mathsf{N}}g_{a,%
\mathbf{h}}(\lambda )T(\xi _{a}^{(h_{a})}).
\label{T-open-interp-T}
\end{equation}

\subsection{The rational limit of trigonometric 6-vertex $R$-matrix and $K $%
-matrix}

Let us remark that both the trigonometric 6-vertex $R$-matrix and the $K$%
-matrix general scalar solution of the trigonometric 6-vertex reflection
equation admit a well defined limit to their rational counterparts. To shorten the notations, we denote all objects related to the rational case with an index $XXX$ and while an index  $XXZ$ will refer to the same objects in the trigonometric case.   In
particular, defining:%
\begin{eqnarray}
\lambda &=&\varepsilon \hat{\lambda},\eta =\varepsilon \hat{\eta},\xi
_{n}=\varepsilon \hat{\xi}_{n},\zeta _{\pm }=\varepsilon \hat{\zeta}_{\pm },
\label{Prescription-Rational-L1} \\
\kappa _{\pm } &=&\hat{\kappa}_{\pm }+O(\varepsilon ),\text{ }\tau _{\pm }=%
\hat{\tau}_{\pm }+\varepsilon \breve{\tau}_{\pm }+O(\varepsilon ^{2}),
\label{Prescription-Rational-L1+}
\end{eqnarray}%
then it holds:%
\begin{align}
\lim_{\varepsilon \rightarrow 0}\frac{R^{(XXZ)}(\lambda |\eta )}{%
\sinh \varepsilon \text{ }}& =R^{(XXX)}(\hat{\lambda}|\hat{\eta}), \\
\lim_{\varepsilon \rightarrow 0}\frac{K_{\pm }^{(XXZ)}(\lambda
|\eta ,\zeta _{\pm },\kappa _{\pm },\tau _{\pm })}{\sinh \varepsilon \text{ }%
}& =K_{\pm }^{(XXX)}(\hat{\lambda}|\hat{\eta},\hat{\zeta}_{\pm },\hat{%
\kappa}_{\pm },\hat{\tau}_{\pm }),
\end{align}%
so that we have:%
\begin{align}
\lim_{\varepsilon \rightarrow 0}\frac{\mathcal{U}_{-}^{(XXZ)}(%
\lambda |\eta ,\{\xi \},\zeta _{-},\kappa _{-},\tau _{-})}{\sinh ^{2\mathsf{N%
}+1}\varepsilon \text{ }}& =\mathcal{U}_{-}^{(XXX)}(\hat{\lambda}|\hat{%
\eta},\{\hat{\xi}\},\hat{\zeta}_{-},\hat{\kappa}_{-},\hat{\tau}_{-}), \\
\lim_{\varepsilon \rightarrow 0}\frac{T_{(XXZ)}(%
\lambda |\eta ,\{\xi \},\zeta _{\pm },\kappa _{\pm },\tau _{\pm })}{\sinh ^{2%
\mathsf{N}+2}\varepsilon \text{ }}& =T_{(XXX)}(\hat{\lambda}%
|\hat{\eta},\{\hat{\xi}\},\hat{\zeta}_{\pm },\hat{\kappa}_{\pm },\hat{\tau}%
_{\pm }).
\end{align}%
Moreover, we get the following prescriptions on the parameters $\alpha _{\pm
}$ in the rational limit:%
\begin{equation}
\alpha _{\pm }(\varepsilon )=\varepsilon \bar{\zeta}_{\pm }\text{ with }\bar{%
\zeta}_{\pm }=\hat{\zeta}_{\pm }/\sqrt{1+4\hat{\kappa}_{\pm }^{2}},
\label{Prescription-Rational-L2}
\end{equation}%
which induces the following functional form for 
\begin{equation}
\beta _{\pm }(\varepsilon )=\hat{\beta}_{\pm }+\varepsilon \breve{\beta}%
_{\pm }
\end{equation}
with:
\begin{eqnarray}
\cosh \hat{\beta}_{\pm } &=&\lim_{\varepsilon \rightarrow 0}\frac{\sinh
\zeta _{\pm }}{2\kappa _{\pm }\sinh \alpha _{\pm }}=\frac{\sqrt{1+4\hat{%
\kappa}_{\pm }^{2}}}{2\hat{\kappa}_{\pm }}, \\
\sinh \hat{\beta}_{\pm } &=&\lim_{\varepsilon \rightarrow 0}\frac{\cosh
\zeta _{\pm }}{2\kappa _{\pm }\cosh \alpha _{\pm }}=\frac{1}{2\hat{\kappa}%
_{\pm }},  \label{Prescription-Rational-L3}
\end{eqnarray}%
and which lead to the following rational limit:%
\begin{equation}
\lim_{\varepsilon \rightarrow 0}\frac{\mathsf{A}_{\alpha _{\pm },\beta _{\pm
}}(\lambda )}{\sinh ^{2\mathsf{N}+2}\varepsilon \text{ }}=\mathsf{A}_{\bar{%
\zeta}_{+},\bar{\zeta}_{-}}(\lambda ),
\end{equation}%
where the $\mathsf{A}_{\bar{\zeta}_{+},\bar{\zeta}_{-}}(\lambda )$ are the
coefficients $(\ref{Coef-6v-Req-R})$ defined for the rational 6-vertex
reflection algebra. This is in agreement with the preservation of the transfer matrix fusion
equations under the rational limit:%
\begin{eqnarray}
0 &=&\lim_{\varepsilon \rightarrow 0}\frac{T_{(XXZ)}(%
\xi _{n}^{(0)})T_{(XXZ)}(\xi _{n}^{(1)})-\mathsf{A}%
_{\alpha _{\pm },\beta _{\pm }}(\xi _{n}^{(0)})\mathsf{A}_{\alpha _{\pm
},\beta _{\pm }}(-\xi _{n}^{(1)})}{\sinh ^{4\mathsf{N}+4}\varepsilon \text{ }%
}, \\
&=&T_{(XXX)}(\hat{\xi}_{n}^{(0)})T_{(XXX)}(%
\hat{\xi}_{n}^{(1)})-\mathsf{A}_{\bar{\zeta}_{+},\bar{\zeta}_{-}}(\hat{\xi}%
_{n}^{(0)})\mathsf{A}_{\bar{\zeta}_{+},\bar{\zeta}_{-}}(-\hat{\xi}%
_{n}^{(1)}).
\end{eqnarray}

\subsection{Applicability of our new approach and comparison with Sklyanin's
SoV}

The general Proposition 2.6 of our first paper \cite{MaiN18} applies to
these representations and it allows us to define the left SoV basis.

\begin{theorem}
\label{SoV-Basis-6v-Open-trigo}Let $T(\lambda )$ be the
one-parameter family of transfer matrix associated to a generic couple $%
K_{-}(\lambda )$ and $K_{+}(\lambda )$ of boundary matrices then%
\begin{equation}
\langle h_{1},...,h_{\mathsf{N}}|\equiv \langle S|\prod_{n=1}^{\mathsf{N}%
}\left( \frac{T(\xi _{n}-\eta /2)}{\mathsf{A}_{\alpha _{\pm
},\beta _{\pm }}(\eta /2-\xi _{n})}\right) ^{1-h_{n}}\text{\ }
\label{SoV-Basis-6v-Open-trigo-0}
\end{equation}%
for any $\{h_{1},...,h_{\mathsf{N}}\}\in \{0,1\}^{\otimes \mathsf{N}}$, is a
co-vector basis of $\mathcal{H}$ for almost any choice of the co-vector $%
\langle S|$, of the inhomogeneity parameters satisfying the condition $(\ref%
{Inhomog-cond})$ modulo $i\pi$, of the parameter $\eta $ and of the boundary
parameters.
\end{theorem}

\begin{proof}
The proof can be given as a consequence of the fact that the rational limit of the
trigonometric transfer matrix coincides with the rational transfer matrix
and then we can apply the theorem proven in the rational case.

More in detail, let us denote by $\{\hat{\eta},\{\hat{\xi}_{n}\},\hat{\zeta%
}_{\pm }$ $\hat{\kappa}_{\pm },\hat{\tau}_{\pm }\}\in \mathbb{C}^{7+\mathsf{N%
}}$ a choice of parameters such that the set of SoV co-vectors%
\begin{equation}
\langle h_{1},...,h_{\mathsf{N}};\hat{\eta},\{\hat{\xi}_{n}\},\hat{\zeta}%
_{\pm }\hat{\kappa}_{\pm },\hat{\tau}_{\pm }|_{(XXX)}\equiv \langle
S|\prod_{n=1}^{\mathsf{N}}\left( \frac{T_{(XXX)}(\hat{\xi}%
_{n}-\hat{\eta}/2|\hat{\eta},\{\hat{\xi}_{n}\},\hat{\zeta}_{\pm }\hat{\kappa}%
_{\pm },\hat{\tau}_{\pm })}{\mathsf{A}_{\bar{\zeta}_{+},\bar{\zeta}_{-}}(%
\hat{\eta}/2-\hat{\xi}_{n}|\hat{\eta},\{\hat{\xi}_{n}\})}\right) ^{1-h_{n}},
\label{SoV-rational-1}
\end{equation}
of the rational $gl_{2}$ reflection algebra $(\ref{SoV-Basis-6v-Open})$ is a
co-vector basis of $\mathcal{H}$. Then, under the prescriptions:%
\begin{equation}
\eta (\varepsilon )=\varepsilon \hat{\eta},\xi _{n}(\varepsilon
)=\varepsilon \hat{\xi}_{n},\zeta _{\pm }(\varepsilon )=\varepsilon \hat{%
\zeta}_{\pm },\kappa _{\pm }(\varepsilon )=\hat{\kappa}_{\pm }+O(\varepsilon
),\text{ }\tau _{\pm }(\varepsilon )=\hat{\tau}_{\pm }+\varepsilon \breve{%
\tau}_{\pm }+O(\varepsilon ^{2}),
\end{equation}%
the SoV co-vectors of the trigonometric $U_{q}(\widehat{gl}_{2})$ reflection algebra $(\ref%
{SoV-Basis-6v-Open-trigo-0})$ admit the following power expansions in $%
\varepsilon $%
\begin{align}
& \langle h_{1},...,h_{\mathsf{N}};\eta (\varepsilon ),\{\xi
_{n}(\varepsilon )\},\zeta _{\pm }(\varepsilon ),\kappa _{\pm }(\varepsilon
),\tau _{\pm }(\varepsilon )|_{(XXZ)}\left. =\right. \langle
h_{1},...,h_{\mathsf{N}};\hat{\eta},\hat{\xi}_{n},\hat{\zeta}_{\pm }\hat{%
\kappa}_{\pm },\hat{\tau}_{\pm }|_{(XXX)}  \notag \\
& \text{ \ \ \ \ \ \ \ \ \ \ \ \ \ \ \ \ \ \ \ \ \ \ \ \ \ \ \ \ \ \ \ \ \ \
\ \ \ \ \ \ \ \ \ \ \ \ \ \ \ \ \ \ \ \ }%
+\varepsilon ^{2(\mathsf{N}-\prod_{n=1}^{\mathsf{N}}h_{n})}\overline{\langle
h_{1},...,h_{\mathsf{N}}|}_{1}+O(\varepsilon ^{4(\mathsf{N}-\prod_{n=1}^{%
\mathsf{N}}h_{n})}),  \label{SoV-trigo-1}
\end{align}%
where $\overline{\langle h_{1},...,h_{\mathsf{N}}|}_{1}$ is some finite
co-vector, as it holds:%
\begin{align}
&\frac{T_{(XXZ)}(\xi _{n}(\varepsilon )-\eta
(\varepsilon )/2|\eta (\varepsilon ),\{\xi _{n}(\varepsilon )\},\zeta _{\pm
}(\varepsilon ),\kappa _{\pm }(\varepsilon ),\tau _{\pm }(\varepsilon ))}{%
\mathsf{A}_{\alpha _{\pm }(\varepsilon ),\beta _{\pm }(\varepsilon )}(\eta
(\varepsilon )/2-\xi _{n}(\varepsilon )|\eta (\varepsilon ),\{\xi
_{n}(\varepsilon )\})}\notag\\
& \ \ \ \ \ \ \ \ \ \ \ \ \ \ \ \ \ \ \ \ \ \ \ \ \ \ \ \ \ \ \ \ \ \ \ \ \ \ \ \ \ \ \ \ \ \ \ \ \ \ \ =\frac{T_{(XXX)}(\hat{\xi}_{n}-\hat{%
\eta}/2|\hat{\eta},\{\hat{\xi}_{n}\},\hat{\zeta}_{\pm }\hat{\kappa}_{\pm },%
\hat{\tau}_{\pm })}{\mathsf{A}_{\bar{\zeta}_{+},\bar{\zeta}_{-}}(\hat{\eta}%
/2-\hat{\xi}_{n}|\hat{\eta},\{\hat{\xi}_{n}\})} +\varepsilon ^{2}\mathsf{O}_{n}+O(\varepsilon ^{4})
\end{align}%
for some finite operator $\mathsf{O}_{n}$. Clearly, the above power
expansions in $\varepsilon $ and the fact that the co-vectors $(\ref%
{SoV-rational-1})$ form by assumption a basis imply that there exists a
positive $\bar{\varepsilon}$ such that the set of co-vectors $(\ref%
{SoV-trigo-1})$ is also a basis for any $\varepsilon $ such that $0\leq
\varepsilon \leq \bar{\varepsilon}$. The statement that $(\ref%
{SoV-Basis-6v-Open-trigo-0})$ is basis for almost any choice of the
parameters $\{\eta ,\{\xi _{n}\},\zeta _{\pm },\kappa _{\pm },\tau _{\pm
}\}\in \mathbb{C}^{7+\mathsf{N}}$ is then mainly a consequence of the fact
that these co-vectors are rational function of polynomials in the variables:%
\begin{equation}
E=e^{2\eta },\,\{X_{n}=e^{2\xi _{n}}\},\,Z_{\pm }=e^{2\zeta _{\pm }},\,\kappa
_{\pm },\,T_{\pm }=e^{2\tau _{\pm }}\text{.}
\end{equation}%
More precisely, let us define the $n^{\mathsf{N}}\times n^{\mathsf{N}}$
matrices:%
\begin{align}
\mathcal{M}_{i,j}^{(XXZ)}\left( \langle S|,\eta ,\{\xi
_{n}\},\zeta _{\pm },\kappa _{\pm },\tau _{\pm }\right) &\equiv \frac{%
\langle h_{1}(i),...,h_{\mathsf{N}}(i);\eta ,\{\xi _{n}\},\zeta _{\pm
},\kappa _{\pm },\tau _{\pm }|_{(XXZ)}e_{j}\rangle}{\prod_{n=1}^{\mathsf{N%
}}\mathsf{A}_{\alpha _{\pm },\beta _{\pm }}^{h_{a}-1}(\eta /2-\xi _{n}|\eta
,\{\xi _{n}\})}, \\
\mathcal{M}_{i,j}^{(XXX)}(\langle S|,\hat{\eta},\{\hat{\xi}_{n}\},\hat{%
\zeta}_{\pm }\hat{\kappa}_{\pm },\hat{\tau}_{\pm }) &\equiv \frac{\langle
h_{1}(i),...,h_{\mathsf{N}}(i);\hat{\eta},\{\hat{\xi}_{n}\},\hat{\zeta}_{\pm
}\hat{\kappa}_{\pm },\hat{\tau}_{\pm }|_{(XXX)}e_{j}\rangle}{\prod_{n=1}^{%
\mathsf{N}}\mathsf{A}_{\bar{\zeta}_{+},\bar{\zeta}_{-}}^{h_{a}-1}(\hat{\eta}%
/2-\hat{\xi}_{n}|\hat{\eta},\{\hat{\xi}_{n}\})},
\end{align}%
for any $i,j\in \{1,...,n^{\mathsf{N}}\}$, where we have defined uniquely
the $\mathsf{N}$-tuple $(h_{1}(i),...,h_{\mathsf{N}}(i))\in
\{1,...,n\}^{\otimes \mathsf{N}}$ by:%
\begin{equation}
1+\sum_{a=1}^{\mathsf{N}}h_{a}(i)n^{a-1}=i\in \{1,...,n^{\mathsf{N}}\},\label{Isomorph-j}
\end{equation}%
and $|e_{j}\rangle\in \mathcal{H}$ is the element $j\in \{1,...,n^{\mathsf{N}}\}$ of
the elementary basis in $\mathcal{H}$. Then the condition that the set $(\ref%
{SoV-Basis-6v-Open-trigo-0})$ form a basis of co-vector in $\mathcal{H}$ is
equivalent to the condition:%
\begin{equation}
\text{det}_{n^{\mathsf{N}}}||\mathcal{M}_{i,j}^{(XXZ)}\left(
\langle S|,\eta ,\{\xi _{n}\},\zeta _{\pm },\kappa _{\pm },\tau _{\pm
}\right) ||\neq 0.  \label{Lin-indep-SoV}
\end{equation}%
Note that the above determinant is polynomial in the variables $E,\{X_{n}\},\,Z_{\pm
},\,\kappa _{\pm }$ and Laurent polynomial in $T_{\pm }=e^{2\tau _{\pm }}$. So
to prove that $(\ref{Lin-indep-SoV})$ indeed holds for almost any values of
the parameters it is enough to prove that it holds in just one point. Now by
using the power expansion in $\varepsilon $ of the trigonometric transfer
matrices we have:%
\begin{align}
&\text{det}_{n^{\mathsf{N}}}||\mathcal{M}_{i,j}^{(XXZ)}\left(
\langle S|,\eta (\varepsilon ),\{\xi _{n}(\varepsilon )\},\zeta _{\pm
}(\varepsilon ),\kappa _{\pm }(\varepsilon ),\tau _{\pm }(\varepsilon
)\right) ||  \notag \\
&\text{ \ \ \ \ \ \ \ \ \ \ \ \ \ \ \ \ \ \ \ \ \ \ \ \ \ \ \ \ \ \ }%
\left. =\right. \varepsilon ^{2\mathsf{N}(2\mathsf{N}+2)}(\text{det}_{n^{%
\mathsf{N}}}||\mathcal{M}_{i,j}^{(XXX)}(\langle S|,\hat{\eta},\{\hat{\xi%
}_{n}\},\hat{\zeta}_{\pm }\hat{\kappa}_{\pm },\hat{\tau}_{\pm
})||+O(\varepsilon ^{2})),
\end{align}%
which is nonzero for any $\varepsilon $ such that $0<\varepsilon \leq \bar{%
\varepsilon}$. This complete the proof of the Theorem.
\end{proof}

It is important to recall that Sklyanin's SoV approach \cite{Skl85} or its
generalized version by Baxter's like gauge transformations \cite{Nic12,FalKN14} works only in the case in which at least one of the two
boundary matrices is non-diagonal and furthermore the boundary parameters
satisfy the requirements%
\begin{equation}
\tau _{+}-\tau _{-}+\left( \mathsf{N}+1-2r\right) \eta \neq \epsilon
_{-}(\alpha _{-}+\beta _{-})-\epsilon _{+}(\alpha _{+}-\beta _{+})+\frac{%
i(\epsilon _{+}+\epsilon _{-})\pi }{2},  \label{Sk-Sov-cond}
\end{equation}%
for any $(r,\epsilon _{+},\epsilon _{-})\in \{1,...,\mathsf{N}\}\times
\{-1,1\}^{2}$. In our SoV approach we can define the above basis even in the
case of both diagonal boundary matrices or in non-diagonal cases which are
forbidden in the generalized Sklyanin's SoV approach.

Let us impose that there exists $(r,\epsilon _{+},\epsilon _{-})\in \{1,...,%
\mathsf{N}\}\times \{-1,1\}^{2}$:%
\begin{equation}
\tau _{+}(\varepsilon )-\tau _{-}(\varepsilon )+\left( \mathsf{N}%
+1-2r\right) \eta (\varepsilon )=\frac{i(\epsilon _{+}+\epsilon _{-})\pi }{2}%
+\sum_{l=+,-}\epsilon _{l}(\beta _{l}(\varepsilon )-l\alpha _{l}(\varepsilon
))\text{ \ \ }\forall \varepsilon \in \mathbb{C},  \label{Non-Skly-SoV}
\end{equation}%
where $\eta (\varepsilon )$, $\tau _{+}(\varepsilon )$, $\alpha _{\pm
}(\varepsilon )$ and $\beta _{\pm }(\varepsilon )$ satisfy the
prescription on the rational limit, i.e. $\left( \ref%
{Prescription-Rational-L1}\right) $-$\left( \ref{Prescription-Rational-L1+}%
\right) $ and $\left( \ref{Prescription-Rational-L2}\right) $-$\left( \ref%
{Prescription-Rational-L3}\right) $, so the above equation is equivalent to:%
\begin{eqnarray}
\hat{\tau}_{+}-\hat{\tau}_{-} &=&\frac{i(\epsilon _{+}+\epsilon _{-})\pi }{2}%
+\epsilon _{-}\hat{\beta}_{-}+\epsilon _{+}\hat{\beta}_{+},
\label{Cond-level-0} \\
\breve{\tau}_{+}-\breve{\tau}_{-} &=&\left( \mathsf{N}+1-2r\right) \hat{\eta}%
+\sum_{l=+,-}\epsilon _{l}(\breve{\beta}_{l}-l\bar{\zeta}_{\pm }).
\label{Cond-level-1}
\end{eqnarray}%
Then taking the rational limit, we obtain $T_{(XXX)}(\hat{%
\lambda}|\hat{\eta},\{\hat{\xi}\},\hat{\zeta}_{\pm },\hat{\kappa}_{\pm },%
\hat{\tau}_{\pm })$ where $\left( \ref{Cond-level-0}\right) $ is just
imposing one condition on the parameters $\hat{\kappa}_{\pm }$ and $\hat{\tau%
}_{\pm }$ which has no effect on the definition of the SoV basis in our
approach for the rational case. So by the polynomiality argument above
developed, it follows that also in the trigonometric case our approach is
defining a basis for almost any value of $\varepsilon $ and of the boundary
parameters satisfying (\ref{Cond-level-0})-(\ref{Cond-level-1}). This
finally implies that our set of co-vector stays a co-vector basis for almost
any choice of the boundary parameters satisfying the constrain (\ref%
{Non-Skly-SoV}).

Here we want to show that under some special choice of the co-vector $\langle
S|$, our SoV left basis reduces to the SoV basis associated to Sklyanin's
approach when this last one is applicable. For simplicity we show this
statement only in the case: 
\begin{equation}
K_{+}(\lambda )=\left( 
\begin{array}{cc}
a_{+}(\lambda ) & b_{+}(\lambda ) \\ 
0 & d_{+}(\lambda )%
\end{array}%
\right)
\end{equation}%
where $b_{+}(\lambda )$ may be also zero. In this case Sklyanin's approach
directly applies and the associate co-vector basis is the eigenbasis of $%
\mathcal{B}_{-}(\lambda )$, which reads (up to normalization):%
\begin{equation}
\langle \,\mathbf{h}_{-}\,|\equiv \langle \,0\,|\prod_{n=1}^{\mathsf{N}%
}\left( \mathcal{A}_{-}(\eta /2-\xi _{n})\right) ^{1-h_{n}}.
\end{equation}%
The proof is done by induction just using the identity:%
\begin{equation}
\langle 0|\mathcal{A}_{-}(\xi _{a}-\eta /2)=0\text{ \ }\forall a\in \{1,...,%
\mathsf{N}\}
\end{equation}%
and the following reflection algebra commutation relations:%
\begin{equation}
\mathcal{A}_{-}\left( \mu \right) \mathcal{A}_{-}\left( \lambda \right) =%
\mathcal{A}_{-}\left( \lambda \right) \mathcal{A}_{-}\left( \mu \right) +%
\frac{\sinh \eta }{\sinh (\lambda +\mu -\eta )}(\mathcal{B}_{-}\left(
\lambda \right) \mathcal{C}_{-}\left( \mu \right) -\mathcal{B}_{-}\left( \mu
\right) \mathcal{C}_{-}\left( \lambda \right) ).  \label{AA-BC-Req-T}
\end{equation}%
The steps in the proof for this trigonometric case are mainly the same as
those presented in the rational case so we do not repeat them here.

\subsection{Transfer matrix spectrum in our SoV approach}

Let us show here how in our SoV schema it is characterized the transfer
matrix spectrum.

\begin{theorem}
\label{SoV-Sp-Ch-gl2-Trigo}Under the same general conditions of Theorem \ref%
{SoV-Basis-6v-Open-trigo}, ensuring the existence of the left SoV basis, the
spectrum of $T(\lambda )$ is characterized by:%
\begin{equation}
\Sigma _{T}=\left\{ t(\lambda ):t(\lambda )=f_{\mathbf{h}%
=0}(\lambda )+\sum_{a=1}^{\mathsf{N}}g_{a,\mathbf{h}=0}(\lambda )x_{a},\text{
\ \ }\forall \{x_{1},...,x_{\mathsf{N}}\}\in S_{T}\right\} ,
\end{equation}%
$S_{T}$ is the set of solutions to the following system of $\mathsf{N}$
quadratic equations:%
\begin{equation}
x_{n}[f_{\mathbf{h}=0}(\xi _{n}^{(1)})+\sum_{a=1}^{\mathsf{N}}g_{a,\mathbf{h}%
=0}(\xi _{n}^{(1)})x_{a}]=\mathsf{A}_{\alpha _{\pm },\beta _{\pm }}(\xi
_{n}^{(0)})\mathsf{A}_{\alpha _{\pm },\beta _{\pm }}(-\xi _{n}^{(1)}),\text{ 
}\forall n\in \{1,...,\mathsf{N}\},  \label{Quadratic System_open}
\end{equation}%
in $\mathsf{N}$ unknown $\{x_{1},...,x_{\mathsf{N}}\}$. Moreover, $%
T(\lambda )$ has w-simple spectrum and for any $t(\lambda )\in
\Sigma _{T}$ the associated unique (up-to normalization)
eigenvector $|t\rangle $ has the following factorized  wave-function in the left SoV
basis:%
\begin{equation}
\langle h_{1},...,h_{\mathsf{N}}|t\rangle =\prod_{n=1}^{\mathsf{N}}\left( 
\frac{t(\xi _{n}-\eta /2)}{\mathsf{A}_{\alpha _{\pm },\beta _{\pm }}(\eta
/2-\xi _{n})}\right) ^{1-h_{n}}.  \label{SoV-Ch-T-eigenV-trig-open}
\end{equation}
\end{theorem}

\begin{proof}
The system of $\mathsf{N}$ quadratic equations $(\ref{Quadratic System_open})$ in 
$\mathsf{N}$ unknown $\{x_{1},...,x_{\mathsf{N}}\}$ is also in the
trigonometric case the rewriting of the transfer matrix fusion equations for
the eigenvalues. So, any transfer matrix eigenvalue is solution of this
system and the associated eigenvector $|t\rangle $ admits the
characterization $(\ref{SoV-Ch-T-eigenV-trig-open})$ in the left SoV basis.
The reverse statement is proven following the same steps of the rational
case by proving that the vector $|t\rangle $ characterized by $(\ref%
{SoV-Ch-T-eigenV-trig-open})$ and associated to any polynomial $t(\lambda )$
satisfying the above system of equations is a transfer matrix eigenvector, i.e. we have to
show:%
\begin{equation}
\langle h_{1},...,h_{\mathsf{N}}|T(\lambda )|t\rangle =t(\lambda
)\langle h_{1},...,h_{\mathsf{N}}|t\rangle ,\text{ }\forall \{h_{1},...,h_{%
\mathsf{N}}\}\in \{0,1\}^{\otimes \mathsf{N}}.
\end{equation}%
We compute first the following matrix elements:%
\begin{equation}
\langle h_{1},...,h_{a},...,h_{\mathsf{N}}|T(\xi _{a}^{\left(
h_{a}\right) })|t\rangle =\left\{ 
\begin{array}{l}
\mathsf{A}_{\alpha _{\pm },\beta _{\pm }}(\eta /2-\xi _{n})\langle
h_{1},...,h_{a}^{\prime }=0,...,h_{\mathsf{N}}|t\rangle \text{ \ if \ }%
h_{a}=1 \\ 
\frac{\mathsf{A}_{\alpha _{\pm },\beta _{\pm }}(\xi _{n}^{(0)})\mathsf{A}%
_{\alpha _{\pm },\beta _{\pm }}(-\xi _{n}^{(1)})}{\mathsf{A}_{\alpha _{\pm
},\beta _{\pm }}(-\xi _{n}^{(1)})}\langle h_{1},...,h_{a}^{\prime }=1,...,h_{%
\mathsf{N}}|t\rangle \text{ \ if \ }h_{a}=0%
\end{array}%
\right.
\end{equation}%
which by the definition of the state $|t\rangle $ can be rewritten as:%
\begin{equation}
\langle h_{1},...,h_{a},...,h_{\mathsf{N}}|T(\xi _{a}^{\left(
h_{a}\right) })|t\rangle =\left\{ 
\begin{array}{l}
t(\xi _{a}^{\left( 1\right) })\prod_{n\neq a,n=1}^{\mathsf{N}}\left( \frac{%
t(\xi _{n}-\eta /2)}{\mathsf{A}_{\alpha _{\pm },\beta _{\pm }}(\eta /2-\xi
_{n})}\right) ^{1-h_{n}}\text{ \ \ if \ }h_{a}=1 \\ 
\frac{\mathsf{A}_{\alpha _{\pm },\beta _{\pm }}(\xi _{n}^{(0)})\mathsf{A}%
_{\alpha _{\pm },\beta _{\pm }}(-\xi _{n}^{(1)})}{\mathsf{A}_{\alpha _{\pm
},\beta _{\pm }}(-\xi _{n}^{(1)})\prod_{n\neq a,n=1}^{\mathsf{N}}\left( 
\frac{t(\xi _{n}-\eta /2)}{\mathsf{A}_{\alpha _{\pm },\beta _{\pm }}(\eta
/2-\xi _{n})}\right) ^{h_{n}-1}}\ \text{if\ }h_{a}=0%
\end{array}%
\right.
\end{equation}%
and finally, by the fusion equation satisfied by $t(\lambda )$, it reads:%
\begin{equation}
\langle h_{1},...,h_{a},...,h_{\mathsf{N}}|T(\xi _{a}^{\left(
h_{a}\right) })|t\rangle =\left\{ 
\begin{array}{l}
t(\xi _{a}^{\left( 1\right) })\prod_{n\neq a,n=1}^{\mathsf{N}}\left( \frac{%
t(\xi _{n}-\eta /2)}{\mathsf{A}_{\alpha _{\pm },\beta _{\pm }}(\eta /2-\xi
_{n})}\right) ^{1-h_{n}}\text{ \ \ \ if \ }h_{a}=1 \\ 
t(\xi _{a}^{\left( 0\right) })\prod_{n=1}^{\mathsf{N}}\left( \frac{t(\xi
_{n}-\eta /2)}{\mathsf{A}_{\alpha _{\pm },\beta _{\pm }}(\eta /2-\xi _{n})}%
\right) ^{1-h_{n}}\text{ \ \ if \ }h_{a}=0%
\end{array}%
\right. ,
\end{equation}%
and so:%
\begin{equation}
\langle h_{1},...,h_{a},...,h_{\mathsf{N}}|T(\xi _{a}^{\left(
h_{a}\right) })|t\rangle =t(\xi _{a}^{\left( h_{a}\right) })\langle
h_{1},...,h_{a},...,h_{\mathsf{N}}|t\rangle .
\end{equation}%
From these identities and by using the interpolation formula:%
\begin{equation}
T(\lambda )=f_{\mathbf{h}}(\lambda )+\sum_{a=1}^{\mathsf{N}}g_{a,%
\mathbf{h}}(\lambda )T(\xi _{a}^{(h_{a})}),
\end{equation}%
we get 
\begin{equation}
\langle h_{1},...,h_{\mathsf{N}}|T(\lambda )|t\rangle =\left( f_{%
\mathbf{h}}(\lambda )+\sum_{a=1}^{\mathsf{N}}g_{a,\mathbf{h}}(\lambda )t(\xi
_{a}^{\left( h_{a}\right) })\right) \langle h_{1},...,h_{\mathsf{N}%
}|t\rangle ,
\end{equation}%
proving our statement.
\end{proof}

The previous characterization of the spectrum allows to introduce an equivalent description in terms of a 
functional equation, the so-called quantum
spectral curve equation, which in the case at hand is a second order Baxter's type
difference equation. In particular, this result coincides with the Theorem
3.1 of \cite{KitMN14}, the only difference being that the applicability
of the result extends now to the case of both diagonal boundary matrices
and non-diagonal boundary matrices even satisfying the condition $\left( \ref%
{Sk-Sov-cond}\right) $.

\begin{theorem}
Under the same conditions of Theorem \ref{SoV-Basis-6v-Open-trigo}, ensuring
the existence of the left SoV basis, $t(\lambda )\in \Sigma _{T}$
if and only if there exists and is unique the polynomial 
\begin{equation}
Q_{t}(\lambda )=\prod_{a=1}^{p_{K_{+,-}}}\left( \cosh 2\lambda -\cosh
2\lambda _{a}\right) ,\qquad \lambda _{1},\ldots ,\lambda _{p_{K_{+,-}}}\in 
\mathbb{C}\setminus \{\pm \xi _{1}^{(0)},\ldots ,\pm \xi _{\mathsf{N}%
}^{(0)}\},
\end{equation}%
such that 
\begin{equation}
t(\lambda )\,Q_{t}(\lambda )=\mathsf{A}_{\alpha _{\pm },\beta _{\pm
}}(\lambda )\,Q_{t}(\lambda -\eta )+\mathsf{A}_{\alpha _{\pm },\beta _{\pm
}}(-\lambda )\,Q_{t}(\lambda +\eta )+F(\lambda ),  \label{InH-BAX-eq-trigo}
\end{equation}%
with 
\begin{equation}
F(\lambda )=F_{0}\,(\cosh ^{2}2\lambda -\cosh ^{2}\eta )\prod_{b=1}^{\mathsf{%
N}}\prod_{i=0}^{1}(\cosh 2\lambda -\cosh 2\xi _{b}^{(i)}),
\label{InH-F-trigo}
\end{equation}%
where%
\begin{eqnarray}
F_{0} &=&\frac{\kappa _{+}\kappa _{-}\left( \cosh (\tau _{+}-\tau
_{-})-\cosh (\alpha _{+}+\alpha _{-}-\beta _{+}+\beta _{-}-(\mathsf{N}%
+1)\eta )\right) }{2^{\mathsf{N}-1}\sinh \zeta _{+}\sinh \zeta _{-}}, \\
p_{K_{+,-}} &=&(1-\delta _{0,F_{0}})\mathsf{N}+\delta _{0,F_{0}}\ p,\text{ \
with }p\leq \mathsf{N}.
\end{eqnarray}
\end{theorem}

\begin{proof}
Under the conditions ensuring the existence of the left SoV basis, the
equivalence of the first discrete SoV characterization with this functional
equation is proven by standard arguments as introduced in \cite{Nic10a,KitMN14}, see for example the proof of the Theorem 3.1 of  \cite{KitMN14}.
\end{proof}

We have already proven that the existence of our SoV basis implies that the
transfer matrix spectrum is simple. Now following the general
Proposition 2.6 of \cite{MaiN18} we can also show that in general the
transfer matrix is diagonalizable with simple spectrum.

\begin{theorem}
For almost any couple $K_{-}(\lambda )$ and $K_{+}(\lambda )$ of boundary
matrices, any choice of the co-vector $\langle S|$, of the inhomogeneity
parameters satisfying the condition $(\ref{Inhomog-cond})$ and of the parameter 
$\eta $, we have that for any eigenvalue $t(\lambda )$ of $%
T(\lambda)$, it holds:%
\begin{equation}
\langle t|t\rangle \neq 0,
\end{equation}%
where $|t\rangle $ and $\langle t|$ are the unique eigenvector and
eigenco-vector associated to $t(\lambda )$, and $T(\lambda)$ is
diagonalizable with simple spectrum.
\end{theorem}

\begin{proof}
The proof follows taking the rational limit and using the result proven in
this case and then by using the fact that eigenvalues and eigenstates are algebraic functions in the parameter of the representations to deduce that the statement is true
for almost any values of the parameters in the trigonometric case.
\end{proof}

\section{SoV for fundamental representations of $Y(gl_{3})$ reflection algebra}
Here, we develop the SoV approach, from the construction of the SoV basis up to the functional 
equation characterization of the transfer matrix spectrum, for the most general fundamental 
representations of the $Y(gl_{3})$ reflection algebra. This spectral problem has been already studied in the 
Analytic and Nested Algebraic Bethe Ansatz framework in \cite{deVG93a,MezN91b,ArnACDFR04,GalM05,ArnCDFR05,BelR09,Nep10}, under 
some special type of boundary conditions. More recently, it has been analyzed in \cite{CaoYSW14} under general 
boundary conditions by a modified version of Analytic Bethe Ansatz producing an Ansatz for the 
transfer matrix eigenvalues.
\subsection{Fundamental representations of $Y(gl_{3})$ reflection algebra}
We consider here the reflection algebra associated to the rational $gl_{3}$ $%
R$-matrix:%
\begin{equation}
R_{a,b}(\lambda )=\lambda I_{a,b}+\eta \mathbb{P}_{a,b}=\left( 
\begin{array}{ccc}
a_{1}(\lambda ) & b_{1} & b_{2} \\ 
c_{1} & a_{2}(\lambda ) & b_{3} \\ 
c_{2} & c_{3} & a_{3}(\lambda )%
\end{array}%
\right) \in \mathrm{End}(V_{a}\otimes V_{b}),
\end{equation}%
where $V_{a}\cong V_{b}\cong \mathbb{C}^{3}$ and we have defined:%
\begin{align}
& a_{j}(\lambda )\left. =\right. \left( 
\begin{array}{ccc}
\lambda +\eta \delta _{j,1} & 0 & 0 \\ 
0 & \lambda +\eta \delta _{j,2} & 0 \\ 
0 & 0 & \lambda +\eta \delta _{j,3}%
\end{array}%
\right) ,\text{ \ \ }\forall j\in \{1,2,3\},  \notag \\
& b_{1}\left. =\right. \left( 
\begin{array}{ccc}
0 & 0 & 0 \\ 
\eta & 0 & 0 \\ 
0 & 0 & 0%
\end{array}%
\right) ,\text{ \ }b_{2}\left. =\right. \left( 
\begin{array}{ccc}
0 & 0 & 0 \\ 
0 & 0 & 0 \\ 
\eta & 0 & 0%
\end{array}%
\right) ,\text{ \ }b_{3}\left. =\right. \left( 
\begin{array}{ccc}
0 & 0 & 0 \\ 
0 & 0 & 0 \\ 
0 & \eta & 0%
\end{array}%
\right) ,  \notag \\
& c_{1}\left. =\right. \left( 
\begin{array}{ccc}
0 & \eta & 0 \\ 
0 & 0 & 0 \\ 
0 & 0 & 0%
\end{array}%
\right) ,\text{ \ }c_{2}\left. =\right. \left( 
\begin{array}{ccc}
0 & 0 & \eta \\ 
0 & 0 & 0 \\ 
0 & 0 & 0%
\end{array}%
\right) ,\text{ \ }c_{3}\left. =\right. \left( 
\begin{array}{ccc}
0 & 0 & 0 \\ 
0 & 0 & \eta \\ 
0 & 0 & 0%
\end{array}%
\right) ,
\end{align}%
which satisfies the Yang-Baxter equation:%
\begin{equation}
R_{12}(\lambda -\mu )R_{13}(\lambda )R_{23}(\mu )=R_{23}(\mu )R_{13}(\lambda
)R_{12}(\lambda -\mu )\in \mathrm{End}(V_{1}\otimes V_{2}\otimes V_{3}).
\end{equation}%
Let us introduce the following boundary matrices \cite{MezN91,deVG93,Kul96,MinRS01}:
\begin{equation}
K_{\pm }(\lambda )=I\mp \frac{\lambda -3\delta _{\pm 1,1}\eta /2}{\zeta
_{\pm }}\mathcal{M}^{\left( \pm \right) },
\end{equation}%
where%
\begin{equation}
(\mathcal{M}^{\left( \pm \right) })^{2}=r^{\left( \pm \right) }I,r^{\left(
\pm \right) }=1,0,
\end{equation}%
and in the case $r^{\left( \pm \right) }=1$ 
\begin{equation}
\mathcal{M}^{\left( \pm \right) }=W^{\left( \pm \right) }\left( 
\begin{array}{ccc}
\epsilon _{1}^{\left( \pm \right) } & 0 & 0 \\ 
0 & \epsilon _{2}^{\left( \pm \right) } & 0 \\ 
0 & 0 & \epsilon _{3}^{\left( \pm \right) }%
\end{array}%
\right) (W^{\left( \pm \right) })^{-1},
\end{equation}%
for any fixed invertible $W^{\left( \pm \right) }\in \mathrm{End}(V)$, where: 
\begin{equation}
\epsilon _{j}^{\left( \pm \right) }=1\text{ for }j\in \{1,..,p_{\pm }\},%
\text{ \ }\epsilon _{j}^{\left( \pm \right) }=-1\text{ for }j\in \{p_{\pm
}+1,..,3\},
\end{equation}%
for some $p_{\pm }\in \{0,1,2,3\}$. These $K_{\pm }$-matrices satisfy the
following reflection equations:%
\begin{equation}
R_{ab}(\lambda -\mu )\,K_{-,a}(\lambda )\,R_{ba}(\lambda +\mu )\,K_{-,b}(\mu
)=K_{-,b}(\mu )\,R_{ab}(\lambda +\mu )\,K_{-,a}(\lambda )\,R_{ba}(\lambda
-\mu ),
\end{equation}%
and%
\begin{equation}
R_{ab}(\mu -\lambda )\,K_{+,a}(\lambda )\,R_{ba}(-\lambda -\mu -3\eta
)\,K_{+,b}(\mu )=K_{+,b}(\mu )\,R_{ab}(-\lambda -\mu -3\eta
)\,K_{+,a}(\lambda )\,R_{ba}(\mu -\lambda ).
\end{equation}%
We can define the following bulk monodromy matrix:%
\begin{equation}
M_{a}(\lambda )\equiv R_{a,\mathsf{N}}(\lambda -\xi _{\mathsf{N}})\cdots
R_{a,1}(\lambda -\xi _{1})\in \mathrm{End}(V_{a}\otimes \mathcal{H}),
\end{equation}%
satisfying the Yang-Baxter algebra associated to $R$, where $\mathcal{H}=\bigotimes_{n=1}^{%
\mathsf{N}}V_{n}$ is the Hilbert space of a lattice model with $\mathsf{N}$ sites, having in each lattice site a local Hilbert space given by a  fundamental
representation. We can then define the boundary monodromy matrix: 
\begin{equation}
\mathcal{U}_{-,a}(\lambda )=M_{a}(\lambda )\,K_{-,a}(\lambda )\,\hat{M}%
_{a}(\lambda )\in \mathrm{End}(V_{a}\otimes \mathcal{H}),
\end{equation}%
satisfying the above reflection equation, where we have defined:%
\begin{equation}
\hat{M}_{a}(\lambda )\equiv R_{a,1}(\lambda +\xi _{1})\cdots R_{a,\mathsf{N}%
}(\lambda +\xi _{\mathsf{N}})\in \mathrm{End}(V_{a}\otimes \mathcal{H}).
\end{equation}%
Then, the  transfer matrix,%
\begin{equation}
T(\lambda )=\text{tr}_{V_{a}}\{K_{+,a}(\lambda )\,M_{a}(\lambda
)\,K_{-,a}(\lambda )\,\hat{M}_{a}(\lambda )\}=\text{tr}_{V_{a}}\left\{
K_{+,a}(\lambda )\,\mathcal{U}_{-,a}(\lambda )\right\} \in \mathrm{End}(%
\mathcal{H}),
\end{equation}%
defines a one-parameter family of commuting operators on $\mathcal{H}$ \cite{Skl88}.

It is interesting to remark that given a couple of integers $p_{\pm }\in \{0,1,2,3\}$,
then the following identity holds:%
\begin{eqnarray}
T(\lambda |p_{\pm },\zeta _{\pm }) &=&\text{tr}_{V_{a}}\{K_{+,a}(\lambda
|p_{+},\zeta _{+})\,M_{a}(\lambda )\,K_{-,a}(\lambda |p_{-},\zeta _{-})\,%
\hat{M}_{a}(\lambda )\} \\
&=&\mathcal{C}T(\lambda |p_{\pm }^{\prime }=3-p_{\pm },\zeta _{\pm }^{\prime
}=-\zeta _{\pm })\mathcal{C},
\end{eqnarray}%
where we have used that%
\begin{eqnarray}
K_{\pm ,a}(\lambda |p_{\pm },\zeta _{\pm }) &=&C_{a}K_{\pm ,a}(\lambda
|p_{\pm }^{\prime }=3-p_{\pm },\zeta _{\pm }^{\prime }=-\zeta _{\pm })C_{a},
\\
C_{a} &=&\left( 
\begin{array}{ccc}
0 & 0 & 1 \\ 
0 & 1 & 0 \\ 
1 & 0 & 0%
\end{array}%
\right) ,\text{ \ \ }\mathcal{C}=\otimes _{n=1}^{\mathsf{N}}C_{n}.
\end{eqnarray}%
In the following we consider only couples of $p_{\pm }\in \{0,1,2,3\}$ which
are not complementary in this sense as those complementary follows by the
above identity. Let us remark moreover that in the case $r^{\left( \pm
\right) }=1$, it holds:%
\begin{equation}
\mathcal{M}^{\left( +\right) }\mathcal{M}^{\left( -\right) }=W^{\left(
+,-\right) }\left( 
\begin{array}{ccc}
\left( -1\right) ^{(3-p_{+})(3-p_{-})} & 0 & 0 \\ 
0 & e^{-\alpha } & 0 \\ 
0 & 0 & e^{\alpha }%
\end{array}%
\right) (W^{\left( +,-\right) })^{-1},  \label{Gen-Form-B+-}
\end{equation}%
for some invertible $3\times 3$ matrix $W^{\left( +,-\right) }$ and $\alpha
\in \mathbb{C}$, being 
\begin{eqnarray}
\text{det}\mathcal{M}^{\left( +\right) }\mathcal{M}^{\left( -\right) }
&=&\left( -1\right) ^{(3-p_{+})(3-p_{-})}, \\
\text{tr}\left( \mathcal{M}^{\left( +\right) }\mathcal{M}^{\left( -\right)
}\right) ^{k} &=&\text{tr}\left( \mathcal{M}^{\left( -\right) }\mathcal{M}%
^{\left( +\right) }\right) ^{k}=\text{tr}\left( \mathcal{M}^{\left( +\right)
}\mathcal{M}^{\left( -\right) }\right) ^{-k}\text{ \ \ }\forall k\in \mathbb{%
Z}\text{.}
\end{eqnarray}
Let us here follow the standard fusion procedure of $R$-matrices \cite{KulRS81,KulR82,Res83,KirR86} and boundary $K$-matrices \cite{MezN91}. We define, the following antisymmetric projectors\footnote{Here, clearly it holds $P_{1,...,m}^{-}=0$ for $m\geq 4$
for the current case $V_{i}\cong \mathbb{C}^{3}$ for any $i\in \{1,...,%
\mathsf{N}\}$.}:
\begin{equation}
P_{1,...,m}^{-}=\frac{\sum_{\pi \in S_{m}}\left( -1\right) ^{\sigma _{\pi
}}P_{\pi }}{m!}\in \mathrm{End}(V_{1}\otimes \cdots \otimes V_{m}),
\end{equation}%
where $S_{m}$ is the set of the permutations $\pi $ of $%
\{1,...,m\}$, ${\sigma _{\pi
}}$ is the signature of $\pi$,  and we have defined%
\begin{equation}
P_{\pi }(v_{1}\otimes \cdots \otimes v_{m})=v_{\pi (1)}\otimes \cdots
\otimes v_{\pi (m)}\in V_{1}\otimes \cdots \otimes V_{m},
\end{equation}%
with $P_{1}^{-}=I$. Now, by using them we can introduce the fused transfer
matrices. In particular, the second fused transfer matrix family reads:%
\begin{equation}
T_{2}(\lambda )=\text{tr}_{V_{\langle ab\rangle }}\{K_{\langle ab\rangle
}^{+}(\lambda )\,M_{\langle ab\rangle }(\lambda )\,K_{\langle ab\rangle
}^{-}(\lambda )\,\hat{M}_{\langle ab\rangle }(\lambda )\},
\end{equation}%
where $V_{\langle ab\rangle }=P_{ab}^{-}V_{a}\otimes V_{b}$, and we have
defined the fused boundary matrices:%
\begin{eqnarray}
K_{\langle ab\rangle }^{+}(\lambda ) &=&P_{ab}^{-}K_{+,b}(\lambda -\eta
)\,R_{ab}(-2\lambda -2\eta )\,K_{+,a}(\lambda )P_{ab}^{-}\,, \\
K_{\langle ab\rangle }^{-}(\lambda ) &=&P_{ab}^{-}K_{-,a}(\lambda
)\,R_{ba}(2\lambda -\eta )\,K_{-,b}(\lambda -\eta )\,P_{ab}^{-},
\end{eqnarray}%
and the fused bulk monodromy matrices:%
\begin{eqnarray}
M_{\langle ab\rangle }(\lambda ) &=&P_{ab}^{-}M_{a}(\lambda )M_{b}(\lambda
-\eta )P_{ab}^{-}, \\
\hat{M}_{\langle ab\rangle }(\lambda ) &=&P_{ab}^{-}\hat{M}_{a}(\lambda )%
\hat{M}_{b}(\lambda -\eta )P_{ab}^{-}.
\end{eqnarray}%
Then, we can define the further fused boundary matrices:%
\begin{eqnarray}
K_{\langle abc\rangle }^{+}(\lambda ) &=&P_{abc}^{-}K_{\langle bc\rangle
}^{+}(\lambda -\eta )\,R_{ac}(-2\lambda -\eta )\,R_{ab}(-2\lambda -2\eta
)\,K_{+,a}(\lambda )P_{abc}^{-}\,, \\
K_{\langle abc\rangle }^{-}(\lambda ) &=&P_{abc}^{-}K_{\langle bc\rangle
}^{-}(\lambda )\,R_{ba}(2\lambda -\eta )\,R_{ca}(2\lambda -2\eta
)\,K_{-,a}(\lambda -\eta )P_{abc}^{-}\,,
\end{eqnarray}%
and the further fused bulk monodromy matrices:%
\begin{eqnarray}
M_{\langle abc\rangle }(\lambda ) &=&P_{abc}^{-}M_{a}(\lambda )M_{b}(\lambda
-\eta )M_{c}(\lambda -2\eta )P_{abc}^{-}, \\
\hat{M}_{\langle abc\rangle }(\lambda ) &=&P_{abc}^{-}\hat{M}_{a}(\lambda )%
\hat{M}_{b}(\lambda -\eta )\hat{M}_{c}(\lambda -2\eta )P_{abc}^{-},
\end{eqnarray}%
and from them the third family of quantum spectral invariants:%
\begin{equation}
T_{3}(\lambda )=\text{tr}_{V_{\langle abc\rangle }}\{K_{\langle abc\rangle
}^{+}(\lambda )\,M_{\langle abc\rangle }(\lambda )\,K_{\langle abc\rangle
}^{-}(\lambda )\,\hat{M}_{\langle abc\rangle }(\lambda )\}\,,
\end{equation}%
where $V_{\langle abc\rangle }=P_{abc}^{-}V_{a}\otimes V_{b}\otimes V_{c}$,
which is also called the quantum determinant.

\subsection{Properties of the transfer matrices}

In this subsection we regroup some properties satisfied by the transfer 
matrices associated to the fundamental representations of the the rank two reflection algebras which play an important role in our SoV construction.

\begin{property}
The transfer matrices $T(\lambda )$ and $T_{2}(\lambda )$ defines two one
parameter families of mutually commuting operators:%
\begin{equation}
\left[ T(\lambda ),T(\mu )\right] =\left[ T(\lambda ),T_{2}(\mu )\right] =%
\left[ T_{2}(\lambda ),T_{2}(\mu )\right] =0.
\end{equation}%
Moreover, the quantum determinant $T_{3}(\lambda )$ is a central element of
the reflection algebra, i.e.%
\begin{equation}
\lbrack T_{3}(\lambda ),\mathcal{U}_{-,a}(\mu )]=0.
\end{equation}
\end{property}

The transfer matrix properties described in the following can be directly derived by using the known $R$-matrix properties like its reduction to the permutation operator and to the antisymmetric and symmetric projectors for special values of its arguments ($0,-\eta$ and $\eta$, respectively). These explicit computations have been presented recently in \cite{CaoYSW14}.

\begin{property}
The quantum spectral invariants have the following polynomial form:

i) $T(\lambda )$ is a degree $2\mathsf{N}+2$ polynomial in $\lambda $ with
the  central asymptotics:%
\begin{equation}
\lim_{\lambda \rightarrow \infty }\lambda ^{-(2+2\mathsf{N})}T(\lambda )=-%
\text{tr}_{a}\mathcal{M}_{a}^{\left( +\right) }\mathcal{M}_{a}^{\left(
-\right) },
\end{equation}%
and the following two central values:%
\begin{equation}
T(0)=\zeta _{-}d(\eta )\text{tr}_{a}K_{a}^{\left( +\right) }(0),\text{ \ }%
T(-3\eta /2)=\zeta _{+}d(3\eta /2)\text{tr}_{a}K_{a}^{\left( -\right)
}(-3\eta /2).  \label{Central-T1}
\end{equation}%

ii) $T_{2}^{\left( K\right) }(\lambda )$ is a degree $4\mathsf{N}+6$
polynomial in $\lambda $ with the  $2\mathsf{N}+2$ central zeros:%
\begin{equation}
T_{2}(\lambda )=(\lambda -\eta )(\lambda +3\eta /2)d(\lambda -\eta )\bar{T}%
_{2}(\lambda ),  \label{Central-zeros-T2}
\end{equation}%
and the central asymptotic behaviour:%
\begin{equation}
\text{ }\lim_{\lambda \rightarrow \infty }\lambda ^{-(6+4\mathsf{N}%
)}T_{2}(\lambda )=\text{tr}_{ab}P_{ab}^{-}\mathcal{M}_{a}^{\left( +\right) }%
\mathcal{M}_{a}^{\left( -\right) }\mathcal{M}_{b}^{\left( +\right) }\mathcal{%
M}_{b}^{\left( -\right) }P_{ab}^{-},
\end{equation}%
furthermore it has the following two central values:%
\begin{eqnarray}
T_{2}(\eta /2) &=&\eta (\eta ^{2}/4-\zeta _{-}^{2})d(\eta /2)d(3\eta /2)%
\text{tr}_{\langle ab\rangle }K_{\langle ab\rangle }^{+}(\eta /2),
\label{Central-T2a} \\
T_{2}(-\eta ) &=&\eta (\eta ^{2}/4-\zeta _{+}^{2})d(\eta )d(2\eta )\text{tr}%
_{\langle ab\rangle }K_{\langle ab\rangle }^{-}(-\eta )\ . \label{Central-T2b}
\end{eqnarray}%
It has also two known values in terms of the transfer matrix:%
\begin{eqnarray}
T_{2}(0) &=&r(-\eta )T(0)T(-\eta ),  \label{T2byT1a} \\
T_{2}(-\eta /2) &=&r(-2\eta )T(-\eta /2)T(-3\eta /2).  \label{T2byT1b}
\end{eqnarray}%
iii) The quantum determinant explicitly reads:%
\begin{align}
T_{3}(\lambda )& =\left( -1\right) ^{p_{+}+p_{-}}(2\lambda -2\eta )(2\lambda
+\eta )(2\lambda -3\eta )(2\lambda +2\eta )(2\lambda -4\eta )(2\lambda
+3\eta )  \notag \\
& \times d(\lambda +\eta )d(\lambda -\eta )\prod_{h=0}^{2-p_{+}}(\eta
/2-\zeta _{+}-h\eta -\lambda )\prod_{h=0}^{p_{+}-1}(\eta /2+\zeta _{+}-h\eta
-\lambda ) \\
& \times d(\lambda -2\eta )\prod_{h=0}^{2-p_{-}}(\lambda -\zeta _{-}-h\eta
)\prod_{h=0}^{p_{-}-1}(\lambda +\zeta _{-}-h\eta ).
\end{align}%
Moreover, the following fusion identities holds:%
\begin{eqnarray}
r(\pm 2\xi _{a}-\eta )r(\pm 2\xi _{a}-2\eta )T(\pm \xi _{a})T_{2}(\pm \xi
_{a}-\eta ) &=&T_{3}(\pm \xi _{a}), \\
r(\pm 2\xi _{a}-\eta )T(\pm \xi _{a})T(\pm \xi _{a}-\eta ) &=&T_{2}(\pm \xi
_{a}),
\end{eqnarray}%
where we have defined%
\begin{equation}
d(\lambda )=\prod_{a=1}^{\mathsf{N}}(\lambda -\xi _{a})(\lambda +\xi _{a}),%
\text{ }a(\lambda )=d(\lambda +\eta ),\text{\ \ }r(\lambda )=-\lambda
(\lambda +3\eta ).
\end{equation}
\end{property}

Moreover, the transfer matrix satisfies the following important set of inversion relations:
\begin{lemma}
The following identities holds:%
\begin{align}
T(\xi _{l})& =R_{l,l-1}(\xi _{l}-\xi _{l-1})\cdots R_{l,1}(\xi _{l}-\xi
_{1})K_{-,l}(\xi _{l})R_{l,1}(\xi _{l}+\xi _{1})\cdots R_{l,l-1}(\xi
_{l}+\xi _{l-1})  \notag \\
& \cdot R_{l,l+1}(\xi _{l}+\xi _{l+1})\cdots R_{l,\mathsf{N}}(\xi _{l}+\xi _{%
\mathsf{N}})tr_{V_{a}}\left[ K_{+,a}(\xi _{l})R_{a,l}(0)R_{a,l}(2\xi _{l})%
\right]  \notag \\
& \cdot R_{l,\mathsf{N}}(\xi _{l}-\xi _{\mathsf{N}})\cdots R_{l,l+1}(\xi
_{l}-\xi _{l+1}),  \label{T+ID}
\end{align}%
and%
\begin{align}
T(-\xi _{l})& =R_{l,l+1}(\xi _{l+1}-\xi _{l})\cdots R_{a,\mathsf{N}}(\xi _{%
\mathsf{N}}-\xi _{l})tr_{V_{a}}\left[ K_{+,a}(-\xi
_{l})R_{a,l}(0)R_{a,l}(-2\xi _{l})\right]  \notag \\
& \cdot R_{l,\mathsf{N}}(-(\xi _{l}+\xi _{\mathsf{N}}))\cdots
R_{l,l+1}(-(\xi _{l}+\xi _{l+1}))R_{l,1-1}(-(\xi _{l}+\xi _{l-1}))\cdots
R_{l,1}(-(\xi _{l}+\xi _{1}))  \notag \\
& \cdot K_{-,l}(-\xi _{l})R_{l,1}(\xi _{1}-\xi _{l})\cdots R_{l,l-1}(\xi
_{l-1}-\xi _{l}),  \label{T-ID}
\end{align}%
for any $l\in \{1,...,\mathsf{N}\}$. Moreover, the following $\mathsf{N}$
centrality conditions hold:%
\begin{equation}
T(\xi _{a})T(-\xi _{a})=r_{a}  \label{Inversion-symmetry}
\end{equation}%
where%
\begin{equation}
r_{a}=\frac{(\xi _{a}-3\eta /2)(\xi _{a}+3\eta /2)}{(\xi _{a}-\eta /2)(\xi
_{a}+\eta /2)}a(\xi _{a})a(-\xi _{a})(\left( \zeta _{+}+\eta /2\right)
^{2}-\xi _{a}^{2})(\zeta _{-}^{2}-\xi _{a}^{2}).
\end{equation}
\end{lemma}

\begin{proof}
Let us prove the above identities for the transfer matrix evaluated  at the
inhomogeneities values. Let us introduce the following short notations:%
\begin{eqnarray}
K_{+,a}^{(l)} &=&K_{+,a}(\xi _{l}),\text{ }K_{-,a}^{(l)}=K_{-,a}(\xi _{l}),%
\text{ }K_{-,l}^{(l)}=K_{-,l}(\xi _{l}), \\
R_{a,h}^{(-)} &=&R_{a,h}(\xi _{l}-\xi _{h}),\text{ }R_{l,h}^{(-)}=R_{l,h}(%
\xi _{l}-\xi _{h}), \\
R_{a,h}^{(+)} &=&R_{a,h}(\xi _{l}+\xi _{h}),\text{ }R_{l,h}^{(+)}=R_{l,h}(%
\xi _{l}+\xi _{h}),
\end{eqnarray}%
and%
\begin{eqnarray}
\hat{K}_{+,a}^{(l)} &=&K_{+,a}(-\xi _{l}),\text{ }\hat{K}%
_{-,a}^{(l)}=K_{-,a}(-\xi _{l}),\text{ }\hat{K}_{-,l}^{(l)}=K_{-,l}(-\xi _{l}),
\\
\hat{R}_{a,h}^{(-)} &=&R_{a,h}(\xi _{h}-\xi _{l}),\text{ }\hat{R}%
_{l,h}^{(-)}=R_{l,h}(\xi _{h}-\xi _{l}), \\
\hat{R}_{a,h}^{(+)} &=&R_{a,h}(-(\xi _{l}+\xi _{h})),\text{ }\hat{R}%
_{l,h}^{(+)}=R_{l,h}(-(\xi _{l}+\xi _{h})).
\end{eqnarray}%
By definition, then $T(\xi _{l})$ reads:%
\begin{align}
& \eta \,tr_{V_{a}}\left[ K_{+,a}^{(l)}R_{a,\mathsf{N}}^{(-)}\cdots
R_{a,l+1}^{(-)}P_{a,l}(0)R_{a,l-1}^{(-)}\cdots
R_{a,1}^{(-)}K_{-,a}^{(l)}R_{a,1}^{(+)}\cdots R_{a,\mathsf{N}}^{(+)}\right] 
\notag \\
& =\eta R_{l,l-1}^{(-)}\cdots R_{l,1}^{(-)}K_{-,l}^{(l)}R_{l,1}^{(+)}\cdots
R_{l,l-1}^{(+)}\,tr_{V_{a}}\left[ K_{+,a}^{(l)}R_{a,\mathsf{N}}^{(-)}\cdots
R_{a,l+1}^{(-)}P_{a,l}(0)R_{a,l}^{(+)}\cdots R_{a,\mathsf{N}}^{(+)}\right] 
\notag \\
& =\eta R_{l,l-1}^{(-)}\cdots R_{l,1}^{(-)}K_{-,l}^{(l)}R_{l,1}^{(+)}\cdots
R_{l,l-1}^{(+)}\,tr_{V_{a}}\left[ K_{+,a}^{(l)}P_{a,l}(0)R_{l,\mathsf{N}%
}^{(-)}\cdots R_{l,l+1}^{(-)}R_{a,l}^{(+)}R_{a,l+1}^{(+)}\cdots R_{a,\mathsf{%
N}}^{(+)}\right] .  \label{T+Step1}
\end{align}%
In the following, we use the Yang-Baxter equation:%
\begin{equation}
R_{l,l+1}^{(-)}R_{a,l}^{(+)}R_{a,l+1}^{(+)}=R_{a,l+1}^{(+)}R_{a,l}^{(+)}R_{l,l+1}^{(-)},
\end{equation}%
and:%
\begin{align}
R_{l,\mathsf{N}}^{(-)}\cdots R_{l,l+2}^{(-)}R_{a,l+1}^{(+)}
&=R_{a,l+1}^{(+)}R_{l,\mathsf{N}}^{(-)}\cdots R_{l,l+2}^{(-)}, \\
R_{l,l+1}^{(-)}R_{a,l+2}^{(+)}\cdots R_{a,\mathsf{N}}^{(+)}
&=R_{a,l+2}^{(+)}\cdots R_{a,\mathsf{N}}^{(+)}R_{l,l+1}^{(-)}.
\end{align}%
to rewrite $\left( \ref{T+Step1}\right) $ as it follows:%
\begin{align}
& \eta R_{l,l-1}^{(-)}\cdots R_{l,1}^{(-)}K_{-,l}^{(l)}R_{l,1}^{(+)}\cdots
R_{l,l-1}^{(+)}\,tr_{V_{a}}\left[ P_{a,l}(0)R_{a,l+1}^{(+)}K_{+,l}^{(l)}R_{l,%
\mathsf{N}}^{(-)}\cdots R_{l,l+2}^{(-)}R_{a,l}^{(+)}R_{a,l+2}^{(+)}\cdots
R_{a,\mathsf{N}}^{(+)}\right] R_{l,l+1}^{(-)}  \notag \\
& =\eta R_{l,l-1}^{(-)}\cdots R_{l,1}^{(-)}K_{-,l}^{(l)}R_{l,1}^{(+)}\cdots
R_{l,l-1}^{(+)}R_{l,l+1}^{(+)}\,tr_{V_{a}}\left[ P_{a,l}(0)K_{+,l}^{(l)}R_{l,%
\mathsf{N}}^{(-)}\cdots R_{l,l+2}^{(-)}R_{a,l}^{(+)}R_{a,l+2}^{(+)}\cdots
R_{a,\mathsf{N}}^{(+)}\right] R_{l,l+1}^{(-)},
\end{align}%
now making the same steps for $R_{l,j}^{(-)}R_{a,l}^{(+)}R_{a,j}^{(+)}$ for
all the $j$ from $l+2$ up to $\mathsf{N}$, we end up with our formula $%
\left( \ref{T+ID}\right) $. Similarly, we have that $T(-\xi _{l})$ reads:%
\begin{align}
& \eta \,tr_{V_{a}}\left[ \hat{K}_{+,a}^{(l)}\hat{R}_{a,\mathsf{N}%
}^{(+)}\cdots \hat{R}_{a,1}^{(+)}\hat{K}_{-,a}^{(l)}\hat{R}%
_{a,1}^{(-)}\cdots \hat{R}_{a,l-1}^{(-)}P_{a,l}(0)\hat{R}_{a,l+1}^{(-)}%
\cdots \hat{R}_{a,\mathsf{N}}^{(-)}\right]  \notag \\
& =\eta tr_{V_{a}}\left[ \hat{K}_{+,a}^{(l)}\hat{R}_{a,\mathsf{N}%
}^{(+)}\cdots \hat{R}_{a,l}^{(+)}P_{a,l}(0)\hat{R}_{a,l+1}^{(-)}\cdots \hat{R%
}_{a,\mathsf{N}}^{(-)}\right] \hat{R}_{l,l-1}^{(+)}\cdots \hat{R}_{l,1}^{(+)}%
\hat{K}_{-,l}^{(l)}\hat{R}_{l,1}^{(-)}\cdots \hat{R}_{l,l-1}^{(-)}\,,  \notag
\\
& =\eta tr_{V_{a}}\left[ \hat{K}_{+,a}^{(l)}P_{a,l}(0)\hat{R}_{l,\mathsf{N}%
}^{(+)}\cdots \hat{R}_{l,l+1}^{(+)}\hat{R}_{a,l}^{(+)}\hat{R}%
_{a,l+1}^{(-)}\cdots \hat{R}_{a,\mathsf{N}}^{(-)}\right] \hat{R}%
_{l,l-1}^{(+)}\cdots \hat{R}_{l,1}^{(+)}\hat{K}_{-,l}^{(l)}\hat{R}%
_{l,1}^{(-)}\cdots \hat{R}_{l,l-1}^{(-)}\,,  \label{T-Step1}
\end{align}%
where in the last line we have used:%
\begin{equation}
\hat{R}_{a,l}^{(+)}P_{a,l}(0)=P_{a,l}(0)\hat{R}_{a,l}^{(+)}.
\end{equation}%
In the following, we use the Yang-Baxter equation:%
\begin{equation}
\hat{R}_{l,l+1}^{(+)}\hat{R}_{a,l}^{(+)}\hat{R}_{a,l+1}^{(-)}=\hat{R}%
_{a,l+1}^{(-)}\hat{R}_{a,l}^{(+)}\hat{R}_{l,l+1}^{(+)},
\end{equation}%
and the commutativities:%
\begin{eqnarray}
\hat{R}_{l,\mathsf{N}}^{(+)}\cdots \hat{R}_{l,l+2}^{(+)}\hat{R}%
_{a,l+1}^{(-)} &=&\hat{R}_{a,l+1}^{(-)}\hat{R}_{l,\mathsf{N}}^{(+)}\cdots 
\hat{R}_{l,l+2}^{(+)}, \\
\hat{R}_{l,l+1}^{(+)}\hat{R}_{a,l+2}^{(-)}\cdots \hat{R}_{a,\mathsf{N}%
}^{(-)} &=&\hat{R}_{a,l+2}^{(-)}\cdots \hat{R}_{a,\mathsf{N}}^{(-)}\hat{R}%
_{l,l+1}^{(+)}.
\end{eqnarray}%
to rewrite $\left( \ref{T-Step1}\right) $ as it follows:%
\begin{align}
& \eta tr_{V_{a}}\left[ \hat{K}_{+,a}^{(l)}P_{a,l}(0)\hat{R}_{a,l+1}^{(-)}%
\hat{R}_{l,\mathsf{N}}^{(+)}\cdots \hat{R}_{l,l+2}^{(+)}\hat{R}_{a,l}^{(+)}%
\hat{R}_{a,l+2}^{(-)}\cdots \hat{R}_{a,\mathsf{N}}^{(-)}\right] \hat{R}%
_{l,l+1}^{(+)}\hat{R}_{l,l-1}^{(+)}\cdots \hat{R}_{l,1}^{(+)}\hat{K}%
_{-,l}^{(l)}\hat{R}_{l,1}^{(-)}\cdots \hat{R}_{l,l-1}^{(-)}\,,  \notag \\
& =\eta \hat{R}_{l,l+1}^{(-)}tr_{V_{a}}\left[ \hat{K}_{+,a}^{(l)}P_{a,l}(0)%
\hat{R}_{l,\mathsf{N}}^{(+)}\cdots \hat{R}_{l,l+2}^{(+)}\hat{R}_{a,l}^{(+)}%
\hat{R}_{a,l+2}^{(-)}\cdots \hat{R}_{a,\mathsf{N}}^{(-)}\right] \hat{R}%
_{l,l+1}^{(+)}\hat{R}_{l,l-1}^{(+)}\cdots \hat{R}_{l,1}^{(+)}\hat{K}%
_{-,l}^{(l)}\hat{R}_{l,1}^{(-)}\cdots \hat{R}_{l,l-1}^{(-)},
\end{align}
now making the same steps for $\hat{R}_{l,j}^{(+)}\hat{R}_{a,l}^{(+)}\hat{R}%
_{a,j}^{(-)}$ for all the $j$ from $l+2$ up to $\mathsf{N}$, we end up with
our formula $\left( \ref{T-ID}\right) $.

Let us now prove the inversion relations. By direct computation one can prove
that the following identities hold:
\begin{eqnarray}
&&K_{-,l}(\xi _{l})tr_{V_{a}}\left[ K_{+,a}(\xi _{l})R_{a,l}(0)R_{a,l}(2\xi
_{l})\right] tr_{V_{a}}\left[ K_{+,a}(-\xi _{l})R_{a,l}(0)R_{a,l}(-2\xi _{l})%
\right] K_{-,l}(-\xi _{l})  \notag \\
&=&(\xi
_{l}^{2}-\left(3 \eta /2\right) ^{2})\eta (\eta -2\xi _{l})(\left( \zeta
_{+}+\eta /2\right) ^{2}-\xi _{l}^{2})(\zeta _{-}^{2}-\xi _{l}^{2})/(\xi
_{l}^{2}-\left( \eta /2\right) ^{2})
\end{eqnarray}%
and by using the unitarity property of the $R$-matrix:%
\begin{equation}
R_{a,b}(x)R_{a,b}(-x)=-(x+\eta )(x-\eta )
\end{equation}%
we get:%
\begin{eqnarray}
R_{l,j}^{(-)}\hat{R}_{l,j}^{(-)} &=&-(\xi _{l}-\xi _{j}+\eta )(\xi _{l}-\xi
_{j}-\eta ), \\
R_{l,j}^{(+)}\hat{R}_{l,j}^{(+)} &=&-(\xi _{l}+\xi _{j}+\eta )(\xi _{l}+\xi
_{j}-\eta ).
\end{eqnarray}%
Finally, by taking the product of the r.h.s. of formulae $\left( \ref{T+ID}%
\right) $ and $\left( \ref{T-ID}\right) $ and by using the above identities we derive $\left( \ref{Inversion-symmetry}\right) $.
\end{proof}

These properties imply that the second transfer matrix $T_2(\lambda)$ can be written in terms
of the transfer matrix $T(\lambda)$. Let us introduce the functions%
\begin{eqnarray}
g_{a,\epsilon }(\lambda ) &=&\frac{\lambda (\lambda +3\eta /2)}{\xi _{a}(\xi
_{a}+3\epsilon \eta /2)}\prod_{b\neq a,b=1}^{\mathsf{N}}\frac{(\lambda
+\epsilon \xi _{a})(\lambda ^{2}-\xi _{b}^{2})}{2\epsilon \xi _{a}(\xi
_{a}^{2}-\xi _{b}^{2})}, \\
f_{a,\epsilon }(\lambda ) &=&\frac{(\lambda ^{2}-\eta ^{2})(\lambda
^{2}-(\eta /2)^{2})d(\lambda -\eta )}{(\xi _{a}^{2}-\eta ^{2})(\xi
_{a}^{2}-(\eta /2)^{2})d(\epsilon \xi _{a}-\eta )}g_{a,\epsilon }(\lambda ),
\end{eqnarray}%
and%
\begin{align}
T^{(\infty )}(\lambda )& =-\lambda (\lambda +3\eta /2)d(\lambda )\,\text{tr}%
_{a}\mathcal{M}_{a}^{\left( +\right) }\mathcal{M}_{a}^{\left( -\right) }%
\text{ } \\
T_{2}^{(\infty )}(\lambda )& =4\lambda (\lambda ^{2}-\left( \eta /2\right)
^{2})(\lambda ^{2}-\eta ^{2})(\lambda +3\eta /2)d(\lambda )  \notag \\
& \times d(\lambda -\eta )\text{tr}_{ab}P_{ab}^{-}\mathcal{M}_{a}^{\left(
+\right) }\mathcal{M}_{a}^{\left( -\right) }\mathcal{M}_{b}^{\left( +\right)
}\mathcal{M}_{b}^{\left( -\right) }P_{ab}^{-},
\end{align}%
then the following corollary holds:

\begin{corollary}
The transfer matrix $T_{2}(\lambda )$ is completely characterized in terms
of the fundamental transfer matrix $T(\lambda )$ by the fusion equations, and the
following interpolation formulae hold:
\begin{equation}
T_{2}(\lambda )=T_{2}^{(\infty )}(\lambda )+\sum_{\epsilon =\pm
1}\sum_{a=1}^{\mathsf{N}}f_{a,\epsilon }(\lambda )T(\epsilon \xi _{a}-\eta
)T(\epsilon \xi _{a})+V(\lambda |T)
\end{equation}%
where%
\begin{align}
V(\lambda |T)& =\frac{8(\lambda ^{2}-\eta ^{2})(\lambda ^{2}-(\eta
/2)^{2})(\lambda +3\eta /2)d(\lambda -\eta )d(\lambda )}{3\eta ^{5}d(\eta
)d(0)}T_{2}(0)  \notag \\
& -\sum_{\epsilon =\pm 1}\frac{16\lambda (\lambda ^{2}-\eta ^{2})(\lambda
+\epsilon \eta /2)(\lambda +3\eta /2)d(\lambda -\eta )d(\lambda )}{3\epsilon
(3+\epsilon )\eta ^{5}d((\epsilon -2)\eta /2)d(\eta /2)}T_{2}(\epsilon \eta
/2)  \notag \\
& +\frac{4\lambda (\lambda -\eta )(\lambda ^{2}-(\eta /2)^{2})(\lambda
+3\eta /2)d(\lambda -\eta )d(\lambda )}{3\eta ^{5}d(2\eta )d(\eta )}%
T_{2}(-\eta ),
\end{align}%
and%
\begin{align}
T(\lambda )& =T^{(\infty )}(\lambda )+\sum_{\epsilon =\pm 1}\sum_{a=1}^{%
\mathsf{N}}g_{a,\epsilon }(\lambda )r_{a}^{(1-\epsilon )/2}(T(\xi
_{a}))^{\epsilon }+\frac{(\lambda +3\eta /2)d(\lambda )}{(3\eta /2)d(0)}T(0)
\notag \\
& -\frac{\lambda d(\lambda )}{(3\eta /2)d(0)}T(-3\eta /2).
\end{align}
\end{corollary}

\begin{proof}
The known central zeros and asymptotics imply the above interpolation
formulae once we use the fusion equations to write $T_{2}(\pm \xi _{a}-\eta
) $ and the inversion relations to write $T(-\xi _{a})$ in terms of $T(\xi
_{a})$.
\end{proof}

\subsection{Our SoV co-vector basis}

\begin{theorem}
Let us assume that both $K_{+,a}(\lambda )$ and $K_{-,a}(\lambda )$ are
non-proportional to the identity\footnote{%
Note this means that $p_{\pm }\in \{1,2\}$ and that we have just two
independent cases here $(p_{+}=1,p_{-}=2)$ (with the equivalent
complementary one $(p_{+}=2,p_{-}=1)$) and $(p_{+}=1,p_{-}=1)$ (with the
equivalent complementary one $(p_{+}=2,p_{-}=2)$).} and that one of the
following requirements holds:

i) $K_{+,a}(\lambda )$ and $K_{-,a}(\lambda )$ are non-commuting, ii) $%
K_{+,a}(\lambda )$ and $K_{-,a}(\lambda )$ are commuting matrices and either 
$r^{\left( +\right) }r^{\left( -\right) }=0$ or $r^{\left( +\right)
}r^{\left( -\right) }=1$, where in this last case, moreover it holds:%
\begin{equation}
\mathcal{M}^{\left( -\right) }=\pm W\left( 
\begin{array}{ccc}
1 & 0 & 0 \\ 
0 & -1 & 0 \\ 
0 & 0 & -1%
\end{array}%
\right) W^{-1}, \ \ \ with\ \ W\in End(\mathbb{C}^3),
\end{equation}%
and%
\begin{equation}
\mathcal{M}^{\left( +\right) }=\pm W\left( 
\begin{array}{ccc}
1 & 0 & 0 \\ 
0 & \left( -1\right) ^{a} & 0 \\ 
0 & 0 & \left( -1\right) ^{a+1}%
\end{array}%
\right) W^{-1}\text{ \ with }a\in \{0,1\},
\end{equation}%
then for almost any choice of $\langle S|$ and of the inhomogeneities under
the condition $(\ref{Inhomog-cond})$, the following set of co-vectors:%
\begin{equation}
\langle h_{1},...,h_{\mathsf{N}}|\equiv \langle S|\prod_{n=1}^{\mathsf{N}%
}(T(\xi _{n}))^{h_{n}}\text{ \ for any }\{h_{1},...,h_{\mathsf{N}}\}\in
\{0,1,2\}^{\otimes \mathsf{N}},
\end{equation}%
forms basis of $\mathcal{H}^*$. In particular, we can take the state $%
\langle S|$ of the following tensor product form:%
\begin{equation}
\langle S|=\bigotimes_{a=1}^{\mathsf{N}}(x,y,z)_{a}\Gamma _{W}^{-1},\text{ \
\ }\Gamma _{W}=\bigotimes_{a=1}^{\mathsf{N}}W_{K,a}\,,
\end{equation}%
simply asking $x\,y\,z\neq 0$.
\end{theorem}

\begin{proof}
This is a special case of the Theorem \ref{General-SoV-basis} presented in the next section for the case $%
n=3$. Following the proof there, we obtain that a sufficient
condition to get the theorem is to prove that there exist  $\alpha _{\pm
}\in \mathbb{C}$ such that the following matrices%
\begin{equation}
K_{a}^{(+,-)}(\alpha _{\pm },\mathcal{M}^{\left( \pm \right) })=(\alpha
_{-}+a\mathcal{M}_{a}^{\left( -\right) })(\alpha _{+}+a\mathcal{M}%
_{a}^{\left( +\right) })\text{ \ }\forall a\in \{1,...,\mathsf{N}\}
\end{equation}%
have simple spectrum. Then, it is simple to observe that the set of
conditions considered above just imply this property.
\end{proof}

\subsection{Transfer matrix spectrum in our SoV approach}

The following characterization of the transfer matrix spectrum holds:

\begin{theorem}
Under the same assumptions ensuring that the set of SoV co-vectors form a
basis, the spectrum of $T(\lambda )$ is characterized by:%
\begin{equation}
\Sigma _{T^{(K)}}=\{ t_{1}(\lambda ):t_{1}(\lambda |\{x_{i}\})=\left\{ 
\begin{array}{l}
T^{(\infty )}(\lambda )-\frac{\lambda d(\lambda )}{(3\eta /2)d(0)}T(-3\eta
/2) \\ 

+\sum_{\epsilon =\pm 1}\sum_{a=1}^{\mathsf{N}}g_{a,\epsilon }(\lambda
)r_{a}^{(1-\epsilon )/2}x_{a}^{\epsilon } \\ 
+\frac{(\lambda +3\eta /2)d(\lambda )}{(3\eta /2)d(0)}T(0)%
\end{array}%
\right. ,\text{ \ }\forall \{x_{1},...,x_{\mathsf{N}}\}\in S_{T}\} .
\label{SET-T-open}
\end{equation}%
Here, $S_{T}$ is the set of solutions to the following system of $\mathsf{N}$
equations in $\mathsf{N}$ unknowns $\{x_{1},...,x_{\mathsf{N}}\}$:%
\begin{equation}
r_{n}^{(1-\mu )/2}x_{n}^{\epsilon }[T_{2}^{(\infty )}(\mu \xi _{n}-\eta
)+\sum_{\epsilon =\pm 1}\sum_{a=1}^{\mathsf{N}}f_{a,\epsilon }(\mu \xi
_{n})t_{1}(\epsilon \xi _{a}-\eta )r_{a}^{(1-\epsilon )/2}x_{a}^{\epsilon
}+V(\lambda |t_{1}(\lambda |\{x_{i}\}))]\left. =\right. T_{3}(\mu \xi _{n}),
\label{System-SoV-open}
\end{equation}%
for any $\mu =\pm 1,n\in \{1,...,\mathsf{N}\}$, where%
\begin{align}
V(\lambda |t_{1}(\lambda |\{x_{i}\}))& =\frac{8(\lambda ^{2}-\eta
^{2})(\lambda ^{2}-(\eta /2)^{2})(\lambda +3\eta /2)d(\lambda -\eta
)d(\lambda )}{3\eta ^{5}d(\eta )d(0)}t_{2}(0)  \notag \\
& -\sum_{\epsilon =\pm 1}\frac{16\lambda (\lambda ^{2}-\eta ^{2})(\lambda
+\epsilon \eta /2)(\lambda +3\eta /2)d(\lambda -\eta )d(\lambda )}{3\epsilon
(3+\epsilon )\eta ^{5}d((\epsilon -2)\eta /2)d(\eta /2)}t_{2}(\epsilon \eta
/2)  \notag \\
& +\frac{4\lambda (\lambda -\eta )(\lambda ^{2}-(\eta /2)^{2})(\lambda
+3\eta /2)d(\lambda -\eta )d(\lambda )}{3\eta ^{5}d(2\eta )d(\eta )}%
t_{2}(-\eta ),
\end{align}%
and we have defined:%
\begin{eqnarray}
t_{2}(0) &=&r(-\eta )t_{1}(0)t_{1}(-\eta ), \\
t_{2}(-\eta /2) &=&r(-2\eta )t_{1}(-\eta /2)t_{1}(-3\eta /2), \\
t_{2}(\eta /2) &=&T_{2}(\eta /2),\text{ }t_{2}(-\eta )=T_{2}(-\eta ).
\end{eqnarray}%
Moreover, $T(\lambda )$ has simple spectrum and for any $t_{1}(\lambda
)\in \Sigma _{T^{(K)}}$ the associated unique (up-to normalization)
eigenvector $|t\rangle $ has the following wave-function in the left SoV
basis:%
\begin{equation}
\langle h_{1},...,h_{\mathsf{N}}|t\rangle =\prod_{n=1}^{\mathsf{N}%
}t_{1}^{h_{n}}(\xi _{n}).  \label{SoV-Ch-T-eigenV-open-gl3}
\end{equation}
\end{theorem}

\begin{proof}
The proof is done according to the same lines of the proof of the Theorem
5.1 for the fundamental representations of the $Y(gl_{3})$ Yang-Baxter algebra
case in \cite{MaiN18}. In fact, we have explicitly illustrated
this for the rank one case, where the proof of the Theorem \ref{SoV-Sp-Ch-gl2-R} for the fundamental representation of the $Y(gl_{2)}$ reflection algebra follows the same lines of that for the Yang-Baxter algebra.
\end{proof}

From the above discrete characterization of the transfer matrix spectrum in
our SoV basis we can prove the following quantum spectral curve functional
reformulation. Here, we consider explicitly only the case\footnote{Here, we have decided to implement the functional equation construction in this case as it is not covered in the existing literature. Indeed, the eigenvalue Ansatz construction presented in \cite{CaoYSW14} has been developed only for the case $(p_{-}=1,p_{+}=1)$. Note that anyhow we can also derive the quantum spectral curve in this last case and it has the same form of the $(p_{-}=2,p_{+}=1)$ case with just modified coefficients and inhomogeneous term.} $(p_{-}=2,p_{+}=1)$.
Note that denoting by $\alpha $ the boundary parameter introduced in $(\ref%
{Gen-Form-B+-})$, the matrices $\mathcal{M}^{\left( -\right) }$\ and\ $%
\mathcal{M}^{\left( +\right) }$ are non-commuting for $\alpha \neq 0$ mod$%
\,i\pi $. While for $\alpha =0$ we keep the transfer matrix simplicity
asking that it holds:%
\begin{equation}
\mathcal{M}^{\left( -\right) }=W\left( 
\begin{array}{ccc}
1 & 0 & 0 \\ 
0 & 1 & 0 \\ 
0 & 0 & -1%
\end{array}%
\right) W^{-1},\ \mathcal{M}^{\left( +\right) }=W\left( 
\begin{array}{ccc}
1 & 0 & 0 \\ 
0 & -1 & 0 \\ 
0 & 0 & -1%
\end{array}%
\right) W^{-1}\text{.}
\end{equation}%
In the above setting the following quantum spectral curve functional equation characterization 
of the spectrum holds

\begin{theorem}
The entire function $t_{1}(\lambda )$ satisfying the conditions $(\ref%
{Central-T1})$ and $\left( \ref{Inversion-symmetry}\right) $ is a $T(\lambda
)$ eigenvalue if and only if there exists a unique polynomial:%
\begin{equation}
\varphi _{t}(\lambda )=\prod_{a=1}^{\mathsf{M}}(\lambda -\lambda
_{a})(\lambda +\lambda _{a}+\eta )\text{\ \ with }\mathsf{M}\leq \mathsf{N}%
\text{,}  \label{Phi-form}
\end{equation}%
and $\lambda _{a}\neq \xi _{n}$ $\forall (a,n)\in \{1,...,\mathsf{M}\}\times
\{1,...,\mathsf{N}\}$ such that $t_{1}(\lambda )$,%
\begin{equation}
t_{2}(\lambda )=T_{2}^{(\infty )}(\lambda )+\sum_{\epsilon =\pm
1}\sum_{n=1}^{\mathsf{N}}f_{n,\epsilon }(\lambda )t_{1}(\epsilon \xi
_{n}-\eta )t_{1}(\epsilon \xi _{n})+V(\lambda |t_{1}(\lambda )),
\end{equation}%
and $\varphi _{t}(\lambda )$ are solutions of the following quantum spectral
curve functional equation:%
\begin{equation}
\alpha (\lambda )\varphi _{t}(\lambda -3\eta )-\beta (\lambda )t_{1}(\lambda
-2\eta )\varphi _{t}(\lambda -2\eta )-\gamma (\lambda )t_{2}(\lambda -\eta
)\varphi _{t}(\lambda -\eta )+T_{3}(\lambda )\varphi _{t}(\lambda
)=f(\lambda ),
\end{equation}%
where:%
\begin{eqnarray}
f(\lambda ) &=&\left( 1-\cos \alpha \right) v_{3}(\lambda )a(\lambda
)d(\lambda )d(\lambda -\eta )d(\lambda -2\eta ) \\
\alpha (\lambda ) &=&v_{2}(\lambda )\gamma _{0}(\lambda )\gamma _{0}(\lambda
-\eta )\gamma _{0}(\lambda -2\eta ), \\
\beta (\lambda ) &=&v_{1}(\lambda )\gamma _{0}(\lambda )\gamma _{0}(\lambda
-\eta ), \\
\gamma (\lambda ) &=&v_{0}(\lambda )\gamma _{0}(\lambda ),
\end{eqnarray}%
and
\begin{align}
\gamma _{0}(\lambda )& =(\eta /2+\zeta _{+}-\lambda )(\zeta _{-}+\lambda
)a(\lambda ), \\
v_{0}(\lambda )& =2^{4}(\lambda ^{2}-\eta ^{2})(\lambda +3\eta /2)(\lambda
-\eta /2), \\
v_{1}(\lambda )& =2^{6}\lambda (\lambda ^{2}-\eta ^{2})(\lambda ^{2}-(3\eta
/2)^{2})(\lambda +\eta /2), \\
v_{2}(\lambda )& =2^{6}\lambda (\lambda ^{2}-\eta ^{2})(\lambda ^{2}-(\eta
/2)^{2})(\lambda +3\eta /2), \\
v_{3}(\lambda )& =2^{8}\lambda (\lambda +\zeta _{-}-\eta )(\lambda -\eta
)(\lambda -2\eta )(\lambda ^{2}-(\eta /2)^{2})  \notag \\
& \times (\lambda ^{2}-\eta ^{2})(\lambda ^{2}-(3\eta /2)^{2})\gamma
_{0}(\lambda ).
\end{align}%
Moreover, up to a normalization the common transfer matrix eigenvector $%
|t\rangle $ admits the following separate wave-function representation:%
\begin{equation}
\langle h_{1},...,h_{\mathsf{N}}|t\rangle =\prod_{a=1}^{\mathsf{N}}\gamma
^{h_{a}}(\xi _{a})\varphi _{t}^{h_{a}}(\xi _{a}-\eta )\varphi
_{t}^{2-h_{a}}(\xi _{a}).
\end{equation}
\end{theorem}

\begin{proof}
Let us start assuming that the entire function $t_{1}(\lambda )$ satisfies
with the polynomial $t_{2}(\lambda )$ and $\varphi _{t}(\lambda )$ the
functional equation then it is a degree $2\mathsf{N}+2$ polynomial in $%
\lambda $ with leading coefficient $t_{1,2\mathsf{N}+2}$ satisfying the
equation:%
\begin{equation}
-1-\text{ }t_{1,2\mathsf{N}+2}+\text{\thinspace tr}%
_{ab}P_{ab}^{-}\mathcal{M}_{a}^{\left( +\right) }\mathcal{M}_{a}^{\left(
-\right) }\mathcal{M}_{b}^{\left( +\right) }\mathcal{M}_{b}^{\left( -\right)
}P_{ab}^{-}-1=4\left( \cos \alpha -1\right) \text{.}
\end{equation}%
By using the identity $\left( \ref{Gen-Form-B+-}\right) $, we can compute
now the following traces:%
\begin{eqnarray}
\text{tr}_{ab}\mathcal{M}_{a}^{\left( +\right) }\mathcal{M}_{a}^{\left(
-\right) } &=&2\cos \alpha -1, \\
\text{tr}_{ab}P_{ab}^{-}\mathcal{M}_{a}^{\left( +\right) }\mathcal{M}%
_{a}^{\left( -\right) }\mathcal{M}_{b}^{\left( +\right) }\mathcal{M}%
_{b}^{\left( -\right) } &=&2\cos \alpha -1,
\end{eqnarray}%
from which it follows:%
\begin{equation}
t_{1,2\mathsf{N}+2}=1-2\cos \alpha =-\text{tr}_{ab}\mathcal{M}_{a}^{\left(
+\right) }\mathcal{M}_{a}^{\left( -\right) },
\end{equation}%
as it is required for transfer matrix eigenvalues. Let us observe now
that it holds:%
\begin{equation}
\alpha (\pm \xi _{a})=\beta (\pm \xi _{a})=f(\pm \xi _{a})=0,\text{ }\gamma
(\pm \xi _{a})\neq 0,\text{ }T_{3}(\pm \xi _{a})\neq 0,
\end{equation}%
so that the functional equation implies:%
\begin{equation}
\frac{\gamma (\pm \xi _{a})\varphi _{t}(\pm \xi _{a}-\eta )}{\varphi
_{t}(\pm \xi _{a})}=\frac{T_{3}(\pm \xi _{a})}{t_{2}(\pm \xi _{a}-\eta )}.
\label{Ch-1-Q-open}
\end{equation}%
Moreover, we have%
\begin{equation}
\alpha (\pm \xi _{a}+\eta )=T_{3}(\pm \xi _{a}-\eta )=f(\pm \xi _{a}+\eta
)=0,\text{ \ }\beta (\pm \xi _{a}+\eta )\neq 0,\text{ }\gamma (\pm \xi
_{a}+\eta )\neq 0,
\end{equation}%
so that the functional equation implies:%
\begin{equation}
\frac{\beta (\pm \xi _{a}+\eta )\varphi _{t}(\pm \xi _{a}-\eta )}{\gamma
(\pm \xi _{a}+\eta )\varphi _{t}(\pm \xi _{a})}=\frac{t_{2}(\pm \xi _{a})}{%
t_{1}(\pm \xi _{a}-\eta )}.  \label{Ch-2-Q-open}
\end{equation}%
Finally, we have:%
\begin{equation}
t_{2}(\pm \xi _{a}+\eta )=T_{3}(\pm \xi _{a}+2\eta )=f(\pm \xi _{a}+2\eta
)=0,\text{ \ }\beta (\pm \xi _{a}+2\eta )\neq 0,\text{ }\alpha (\pm \xi
_{a}+2\eta )\neq 0,
\end{equation}%
so that the functional equation implies:%
\begin{equation}
\frac{\alpha (\pm \xi _{a}+2\eta )\varphi _{t}(\pm \xi _{a}-\eta )}{\beta
(\pm \xi _{a}+2\eta )\varphi _{t}(\pm \xi _{a})}=t_{1}(\pm \xi _{a}).
\label{Ch-3-Q-open}
\end{equation}%
These identities imply the following ones:%
\begin{eqnarray}
r(\pm 2\xi _{a}-\eta )r(\pm 2\xi _{a}-2\eta )t_{1}(\pm \xi _{a})t_{2}(\pm
\xi _{a}-\eta ) &=&T_{3}(\pm \xi _{a}),\text{ }\forall a\in \{1,...,\mathsf{N%
}\}, \\
r(\pm 2\xi _{a}-\eta )t_{1}(\pm \xi _{a})t_{1}(\pm \xi _{a}-\eta )
&=&t_{2}(\pm \xi _{a}),\text{ }\forall a\in \{1,...,\mathsf{N}\},
\end{eqnarray}%
so that, by the SoV characterization obtained in our previous theorem, we have
that $t_{1}(\lambda )$ and $t_{2}(\lambda )$ are eigenvalues of the transfer
matrices $T(\lambda )$ and $T_{2}(\lambda )$, associated to the same
eigenvector $|t\rangle $.

Let us now prove the reverse statement, i.e. we assume that $t_{1}(\lambda )$
is eigenvalue of the transfer matrix $T(\lambda )$
and we want to show that there exists a polynomial $\varphi _{t}(\lambda )$
which satisfies with $t_{1}(\lambda )$ and $t_{2}(\lambda )$ the functional
equation. Here, we characterize $\varphi _{t}(\lambda )$ by imposing that it
satisfies the following set of conditions:%
\begin{equation}
\gamma (\pm \xi _{a})\frac{\varphi _{t}(\pm \xi _{a}-\eta )}{\varphi
_{t}(\pm \xi _{a})}=t_{1}(\pm \xi _{a}).  \label{Discrete-Ch-Q-open}
\end{equation}%
The fact that this characterizes uniquely a polynomial of the form $\left( %
\ref{Phi-form}\right) $ can be shown just following the general proof given
in \cite{MaiN18a}. Let us show that this characterization of $\varphi
_{t}(\lambda )$ implies that the functional equation is indeed satisfied.
The functional equation is a degree $8\mathsf{N}+12$ polynomial in $\lambda $
so to show it we have just to prove that it is satisfied in $8\mathsf{N}+12$
distinct points as the leading coefficient is zero, as we have shown above.
We use the following $8\mathsf{N}$ points $\pm \xi _{a}+k_{a}\eta $, for any 
$a\in \{1,...,\mathsf{N}\}$ and $k_{a}\in \left\{ -1,0,1,2\right\} $.
Indeed, for $\lambda =\pm \xi _{a}-\eta $ it holds:%
\begin{equation}
\alpha (\pm \xi _{a}-\eta )=\beta (\pm \xi _{a}-\eta )=\gamma (\pm \xi
_{a}-\eta )=T_{3}(\pm \xi _{a}-\eta )=f(\pm \xi _{a}-\eta )=0,
\end{equation}%
from which the functional equation is satisfied for any $a\in \{1,...,%
\mathsf{N}\}$ and in the remaining $6\mathsf{N}$ points the functional
equation reduces to the $6\mathsf{N}$ equations $\left( \ref{Ch-1-Q-open}%
\right) $-$\left( \ref{Ch-3-Q-open}\right) $ which are equivalent to the
discrete characterization $\left( \ref{Discrete-Ch-Q-open}\right) $, thanks to the fusion equations satisfied by the transfer matrix eigenvalues.
Finally, by using the explicit form of the quantum
determinant, we can show that the spectral curve equation factorizes the
following polynomial of degree $6$:%
\begin{equation}
(\lambda +\zeta _{-}-\eta )(\lambda ^{2}-\eta ^{2})(\lambda +3\eta /2)\gamma
_{0}(\lambda ),
\end{equation}%
as indeed it holds:
\begin{equation}
T_{2}(-\zeta _{-})=0\text{, being }K_{\langle ab\rangle }^{-}(-\zeta _{-})=0.
\end{equation}
Moreover, we can prove that this simplified quantum spectral equation (i.e. the one obtained after removing the above 6 common zeros) is satisfied in
the following $6$ points:%
\begin{equation}
\lambda =0,\pm \eta /2,\eta ,3\eta /2,2\eta ,
\end{equation}
just using the known central zeros of the second transfer matrix $\left( \ref{Central-zeros-T2}\right) $ and the transfer matrix properties \rf{Central-T1}
and \rf{Central-T2a}-\rf{T2byT1b}. This completes our proof of the equivalent rewriting of the spectrum in terms of the quantum spectral curve. 

Finally, renormalizing the eigenvector $|t\rangle $ multiplying it by the
non-zero product of the $\varphi _{t}^{2}(\xi _{a})$ over all the $a\in
\{1,...,\mathsf{N}\}$ we get: 
\begin{equation}
\prod_{a=1}^{\mathsf{N}}\varphi _{t}^{2}(\xi _{a})\prod_{a=1}^{\mathsf{N}%
}t_{1}^{h_{a}}(\xi _{a})\overset{\left( \ref{Discrete-Ch-Q-open}\right) }{=}%
\prod_{a=1}^{\mathsf{N}}\gamma ^{h_{a}}(\xi _{a})\varphi _{t}^{h_{a}}(\xi
_{a}-\eta )\varphi _{t}^{2-h_{a}}(\xi _{a}),
\end{equation}%
which proves our statement on the SoV characterization of the transfer
matrix eigenvectors presented in this theorem.
\end{proof}

\section{SoV basis for fundamental representations of $Y(gl_{n})$ reflection
algebra}

Here, we show that the transfer matrices of the fundamental representations
of $Y(gl_{n})$ reflection algebra can be used also for the general higher
rank $n\geq 3$ case as the independent generators of the SoV basis.

Let us consider the $Y(gl_{n})$ $R$-matrix%
\begin{equation}
R_{ab}(\lambda _{a}-\lambda _{b})=(\lambda _{a}-\lambda _{b})I_{ab}+\eta 
\mathbb{P}_{ab}\in \mathrm{{End}(V_{a}\otimes V_{b}),}
\end{equation}%
$\text{with }V_{a}=\mathbb{C}^{n}\text{, }V_{b}=\mathbb{C}^{n}\text{, }n\in 
\mathbb{N}^{\ast }$, solution of the Yang-Baxter equation:%
\begin{equation}
R_{ab}(\lambda _{a}-\lambda _{b})R_{ac}(\lambda _{a}-\lambda
_{c})R_{bc}(\lambda _{b}-\lambda _{c})=R_{bc}(\lambda _{b}-\lambda
_{c})R_{ac}(\lambda _{a}-\lambda _{c})R_{ab}(\lambda _{a}-\lambda _{b})\in 
\mathrm{{End}(V_{a}\otimes V_{b}\otimes V_{c})},
\end{equation}%
where $\mathbb{P}_{ab}$ is the permutation operator on the tensor product $%
V_{a}\otimes V_{b}$ and $\eta $ is an arbitrary complex number. Then, we can
define the bulk monodromy matrix:%
\begin{equation}
M_{a}(\lambda )\equiv R_{a\mathsf{N}}(\lambda _{a}-\xi _{\mathsf{N}})\cdots
R_{a1}(\lambda _{a}-\xi _{1})\in \mathrm{{End}(V_{a}\otimes \mathcal{H})},
\end{equation}%
satisfying the Yang-Baxter algebra:%
\begin{equation}
R_{ab}(\lambda _{a}-\lambda _{b})M_{a}(\lambda _{a})M_{b}(\lambda
_{b})=M_{b}(\lambda _{b})M_{a}(\lambda _{b})R_{ab}(\lambda _{a}-\lambda
_{b})\in \mathrm{{End}(V_{a}\otimes V_{b}\otimes \mathcal{H}),}
\end{equation}%
where $\mathcal{H}\equiv \otimes _{l=1}^{\mathsf{N}}V_{l}$. The boundary
matrices:%
\begin{equation}
K_{\pm }(\lambda )=I\mp \frac{\lambda -n\delta _{\pm 1,1}\eta /2}{\zeta
_{\pm }}\mathcal{M}^{\left( \pm \right) },
\end{equation}%
where%
\begin{equation}
(\mathcal{M}^{\left( \pm \right) })^{2}=r^{\left( \pm \right) }I,r^{\left(
\pm \right) }=1,0
\end{equation}%
define the most general scalar solutions to the reflection and dual
reflection equations:%
\begin{equation}
K_{-,a}(\lambda _{a})R_{ab}(\lambda _{a}-\lambda _{b})K_{-b}(\lambda
_{b})R_{ab}(\lambda _{a}+\lambda _{b})=K_{-,b}(\lambda _{b})R_{ab}(\lambda
_{a}+\lambda _{b})K_{-,a}(\lambda _{a})R_{ab}(\lambda _{a}-\lambda _{b})
\end{equation}%
and%
\begin{equation}
K_{+,a}^{t_{a}}(\lambda _{a})R_{ab}(\lambda _{b}-\lambda
_{a})K_{+,b}^{t_{b}}(\lambda _{b})R_{ab}(\lambda _{a}+\lambda _{b}-n\eta
)=K_{+,b}^{t_{b}}(\lambda _{b})R_{ab}(\lambda _{a}+\lambda _{b}-n\eta
)K_{+,a}^{t_{a}}(\lambda _{a})R_{ab}(\lambda _{b}-\lambda _{a}).
\end{equation}%
By using them we can define the boundary transfer matrix:%
\begin{equation}
T(\lambda )\equiv tr_{V_{a}}[K_{+,a}(\lambda )M_{a}(\lambda
)K_{-,a}(\lambda)\hat{M}_{a}(\lambda )],
\end{equation}%
where
\begin{equation}
\hat{M}_{a}(\lambda )\equiv R_{a1}(\lambda +\xi _{1})\cdots R_{a\mathsf{N}%
}(\lambda +\xi _{\mathsf{N}}),
\end{equation}%
is proven to be a one parameter family of commuting
operators following Sklyanin's paper \cite{Skl88}.

\subsection{Generating the SoV basis by transfer matrix action}

In these fundamental representations for the $Y(gl_{n})$ reflection algebra
the following theorem holds:

\begin{theorem}
\label{General-SoV-basis}The following set of co-vectors:%
\begin{equation}
\langle h_{1},...,h_{\mathsf{N}}|\equiv \langle S|\prod_{a=1}^{\mathsf{N}%
}(T(\xi _{a}))^{h_{a}}\text{\ for any }\{h_{1},...,h_{\mathsf{N}%
}\}\in \{0,...,n-1\}^{\otimes \mathsf{N}},  \label{co-vector-basis-open}
\end{equation}%
is a basis of $\mathcal{H}^{\ast }$ for almost any choice of the co-vector $%
\langle S|$, of the value of $\eta \in \mathbb{C}$ , of the inhomogeneity
parameters satisfying \rf{Inhomog-cond} and of the boundary parameters in $K_{\pm ,a}(\lambda )$.

In particular, for any choice of the boundary parameters such that there exist $\alpha _{\pm }\in \mathbb{C}$ for which the following matrix
\begin{equation}
K_{a}^{(+,-)}(\alpha _{\pm },\mathcal{M}^{\left( \pm \right) })=(\alpha
_{-}I_{a}+a\mathcal{M}_{a}^{\left( -\right) })(\alpha _{+}I_{a}+a\mathcal{M}%
_{a}^{\left( +\right) })\in \mathrm{{End}(V_{a})}
\end{equation}%
has simple spectrum on $V_{a}$ for any $a\in \{1,...,\mathsf{N}\},$
then we can take
\begin{equation}
\langle S|=\bigotimes_{a=1}^{\mathsf{N}}\langle S,a|,\text{ \ \ with \ }%
\langle S,a|\in V_{a}\text{ \ }\forall a\in \{1,...,\mathsf{N}\},
\label{Starting-State-TP}
\end{equation}
such that\footnote{The existence of $\langle S,a|$ is implied by the spectrum simplicity of $K_{a}^{(+,-)}$.}
\begin{equation}
\langle S,a|(K_{a}^{(+,-)})^{h}\text{ with }h\in \{0,...,n-1\},
\end{equation}%
form a co-vector basis for $V_{a}$ for any $a\in \{1,...,\mathsf{N}\}$.
\end{theorem}

\begin{proof}
We can follow the method already presented in \cite{MaiN18} for the
proof of the general Proposition 2.4. Here we use that the transfer matrix
is a polynomial in $\eta $, the inhomogeneities $\{\xi _{a}\}_{a\in \{1,...,%
\mathsf{N}\}}$ and Laurent polynomial in the boundary parameters. So the
determinant of the $n^{\mathsf{N}}\times n^{\mathsf{N}}$ matrix, whose lines
coincides with the component of the co-vectors $\left( \ref%
{co-vector-basis-open}\right) $ in the natural  basis of $\mathcal{H}%
$, is a polynomial in the component of the co-vector $\langle S|\in 
\mathcal{H}^{\ast }$, in $\eta $, in the inhomogeneities $\{\xi _{a}\}_{a\in
\{1,...,\mathsf{N}\}}$ and a Laurent polynomial in the boundary parameters.
Then it is enough to prove that it is nonzero for some special value of
these parameters to prove that it is so for almost any value of these
parameters.

Let us observe now that from \rf{T+ID}, it follows that 
$T(\xi _{l})\zeta _{+}\zeta _{-}$ are polynomials of degree $2$%
\textsf{$N$}$+1$ in $\xi $ for all $l\in \{1,...,\mathsf{N}\}$ with
maximal degree coefficient given by:%
\begin{equation}
d_{l,2\mathsf{N}+1}K_{l}^{(+,-)}(\alpha _{\pm },\mathcal{M}^{\left( \pm
\right) }),\text{ with }d_{l,2\mathsf{N}+1}=\eta (-1)^{\mathsf{N}-l}l(%
\mathsf{N}-l)!(\mathsf{N}+l)!,
\end{equation}
once we impose:%
\begin{equation}
\xi _{a}=a\xi \text{ \ }\forall a\in \{1,...,\mathsf{N}\}\text{ and }\zeta
_{\pm }=\alpha _{\pm }\xi.  \label{First-gl_n-asymp-Cond}
\end{equation}%
So that the co-vectors $\langle h_{1},...,h_{\mathsf{N}}|$ have the following expansion in $\xi$:%
\begin{equation}
\langle h_{1},...,h_{\mathsf{N}}|\equiv \frac{\xi ^{(2\mathsf{N}%
+1)\sum_{a=1}^{\mathsf{N}}h_{a}}\prod_{a=1}^{\mathsf{N}}d_{a,2\mathsf{N}%
+1}^{h_{a}}\langle S|\prod_{a=1}^{\mathsf{N}}(K_{a}^{(+,-)}(\alpha _{\pm },%
\mathcal{M}^{\left( \pm \right) }))^{h_{a}}+O(\xi ^{((2\mathsf{N}%
+1)\sum_{a=1}^{\mathsf{N}}h_{a})-1})}{\prod_{a=1}^{\mathsf{N}}\xi
_{a}^{h_{a}}(\xi _{a}^{h_{a}}-\eta n/4)}.
\end{equation}%
Hence  a sufficient condition to generate a
basis is given by:%
\begin{equation}
\text{det}_{n^{\mathsf{N}}}||\left( \langle S|\left( \prod_{a=1}^{\mathsf{N}%
}(K_{a}^{(+,-)}(\alpha _{\pm },\mathcal{M}^{\left( \pm \right)
}))^{h_{a}(i)}\right) |e_{j}\rangle\right) _{i,j\in \{1,...,n^{\mathsf{N}}\}}||\neq
0,
\end{equation}%
where for any $i\in \{1,...,n^{\mathsf{N}}\}$ the $\mathsf{N}$-tuple $(h_{1}(i),...,h_{\mathsf{N}}(i))\in
\{1,...,n\}^{\otimes \mathsf{N}}$ is uniquely defined by \rf{Isomorph-j} and $|e_{j}\rangle$ is the element $j\in \{1,...,n^{\mathsf{N}}\}$ of the natural
basis in $\mathcal{H}$. If we take $\langle S|$ of the tensor product form $(%
\ref{Starting-State-TP})$\ then the above determinant reduces to:%
\begin{equation}
\prod_{a=1}^{\mathsf{N}}\text{det}_{n}||\left( \langle S,a|\left( (\alpha
_{-}+a\mathcal{M}_{a}{}^{\left( -\right) })^{i-1}(\alpha _{+}+a\mathcal{M}%
_{a}^{\left( +\right) })^{i-1}\right) |e_{j}(a)\rangle\right) _{i,j\in
\{1,...,n\}}||,
\end{equation}%
where $|e_{j}(a)\rangle$ is the element $j\in \{1,...,n\}$ of the natural basis in $%
V_{a}$. Let us now show that for general choice of the boundary parameters
the matrices $K_{a}^{(+,-)}(\alpha _{\pm },\mathcal{M}^{\left( \pm \right)
}) $ can be indeed taken with non-degenerate spectrum. Let us prove it for
the special choice:%
\begin{equation}
K_{a}^{(+,-)}(\alpha _{\pm }=0,\mathcal{M}^{\left( \pm \right) })=\mathcal{M}%
_{a}^{\left( -\right) }\mathcal{M}_{a}^{\left( +\right) },
\end{equation}%
for any $a$. The remaining boundary parameters are indeed contained in the
choice of two matrices $\mathcal{M}_{a}^{\left( -\right) }$ and $\mathcal{M}%
_{a}^{\left( +\right) }$ consistently with the conditions $\mathcal{M}%
_{a}^{\left( \pm \right) 2}=I_{a}$. In particular, we are free to take:%
\begin{equation}
\lbrack \mathcal{M}_{a}^{\left( -\right) },\mathcal{M}_{a}^{\left( +\right)
}]\neq 0,
\end{equation}%
so that we have to determine the conditions on the eigenvalues of the full
matrix $\mathcal{M}_{a}^{\left( -\right) }\mathcal{M}_{a}^{\left( +\right) }$
and prove that it can have simple spectrum. Let us denote by $t_{j}$
an eigenvalue of the matrix $\mathcal{M}_{a}^{\left( -\right) }\mathcal{M}_{a}^{\left( +\right) }$ and $m_{j}$ the corresponding degeneracy, for $j\in \{1,...,\bar{n}\}$ and $n=\sum_{j=1}^{\bar{n}}m_{j}$. Then, by%
\begin{equation}
\text{det}{}_{a}\mathcal{M}_{a}^{\left( -\right) }\mathcal{M}_{a}^{\left(
+\right) }=\text{det}{}_{a}\mathcal{M}_{a}^{\left( -\right) }\text{det}{}_{a}%
\mathcal{M}_{a}^{\left( +\right) },
\end{equation}%
it follows%
\begin{equation}
\prod_{j=1}^{\bar{n}}t_{j}^{m_{j}}=\left( -1\right) ^{s_{-}+s_{+}}\text{
with }\left( -1\right) ^{s_{\pm }}=\text{det}{}_{a}\mathcal{M}_{a}^{\left(
\pm \right) },  \label{Con-1}
\end{equation}%
while the identity%
\begin{equation}
\sum_{j=1}^{\bar{n}}m_{j}t_{j}^{r}=\sum_{j=1}^{\bar{n}}m_{j}t_{j}^{-r}\text{
\ for positive integer }r,  \label{Con-2}
\end{equation}%
follows from the identities:%
\begin{align}
tr_{V_{a}}\left( \left( \mathcal{M}_{a}^{\left( -\right) }\mathcal{M}%
_{a}^{\left( +\right) }\right) ^{r}\right) & =tr_{V_{a}}\left( \mathcal{M}%
_{a}^{\left( -\right) }\mathcal{M}_{a}^{\left( +\right) }\cdots \mathcal{M}%
_{a}^{\left( -\right) }\mathcal{M}_{a}^{\left( +\right) }\right)  \notag \\
& =tr_{V_{a}}\left( \mathcal{M}_{a}^{\left( +\right) }\mathcal{M}%
_{a}^{\left( -\right) }\mathcal{M}_{a}^{\left( +\right) }\cdots \mathcal{M}%
_{a}^{\left( -\right) }\right)  \notag \\
& =tr_{V_{a}}\left( \left( \mathcal{M}_{a}^{\left( +\right) }\mathcal{M}%
_{a}^{\left( -\right) }\right) ^{r}\right)
\end{align}%
and from the identity:%
\begin{equation}
\left( \mathcal{M}_{a}^{\left( -\right) }\mathcal{M}_{a}^{\left( +\right)
}\right) ^{-1}=\mathcal{M}_{a}^{\left( +\right) }\mathcal{M}_{a}^{\left(
-\right) }.
\end{equation}%
The conditions $(\ref{Con-1})$ and $(\ref{Con-2})$ imply that $t_{j}\neq 0$
for any $j\in \{1,...,\bar{n}\}$ and%
\begin{equation}
\forall \text{ }j\in \{1,...,\bar{n}\}:t_{j}\neq \pm 1\rightarrow \exists
!\,h(j)\text{ }\in \{1,...,\bar{n}\}:t_{j}=t_{h(j)}^{-1},m_{j}=m_{h(j)}.
\end{equation}%
It is easy now to show that these conditions are compatible with the simplicity of the 
spectrum of $\mathcal{M}_{a}^{\left( -\right) }\mathcal{M}_{a}^{\left(
+\right) }$. Avoiding the case of non-trivial Jordan blocks for simplicity,
for example we can ask directly $\bar{n}=n$, i.e. $m_{j}=1$ for any $j\in \{1,...,%
\bar{n}\}$. Then we can distinguish the cases, for $n$ odd we can choose the
following solution to the above conditions:%
\begin{equation}
t_{1}=\left( -1\right) ^{s_{-}+s_{+}},\text{ }t_{1+j+(n-1)/2}=t_{1+j}^{-1},%
\text{ }t_{1+j}\neq \pm 1,t_{1+j}\neq t_{1+h}\text{\ }\forall h\neq j\in
\{1,...,(n-1)/2\}
\end{equation}%
for $n$ even and $s_{-}+s_{+}$ even, we can choose the following solution to
the above conditions:%
\begin{equation}
t_{j+n/2}=t_{j}^{-1},\text{ }t_{j}\neq \pm 1,t_{j}\neq t_{h},\text{\ }%
\forall h\neq j\in \{1,...,n/2\},
\end{equation}%
while for $n$ even and $s_{-}+s_{+}$ odd, we can choose:%
\begin{equation}
t_{1}=1,t_{2}=-1,\text{ }t_{2+j+(n-2)/2}=t_{2+j}^{-1},\text{ }t_{2+j}\neq
\pm 1,t_{2+j}\neq t_{2+h}\text{\ }\forall h\neq j\in \{1,...,n/2-1\}.
\end{equation}%
This prove the possibility to choose $\mathcal{M}_{a}^{\left( -\right) }%
\mathcal{M}_{a}^{\left( +\right) }$ with simple spectrum which completes the
proof.
\end{proof}

\textbf{Remark:} Note that the results about the construction of the SoV
basis and the simplicity and diagonalizability of the fundamental transfer matrix
can be extended naturally to the fundamental representations of the $U_{q}(%
\widehat{gl}_{n})$ reflection algebra just using the same arguments
described in the general Proposition 2.5 and Proposition 2.6 of \cite{MaiN18}. In particular, the proof follows mainly the same lines
described for the case $n=2$ in Theorem \ref{SoV-Basis-6v-Open-trigo} of our
current paper.

\subsection{Diagonalizability and simplicity of the transfer matrix}

We want to show that the transfer matrix associated to the fundamental
representation of the $Y(gl_{n})$ reflection algebra is diagonalizable with
simple spectrum under some further requirement on the boundary matrices.

\begin{theorem}
Let the boundary matrices $\mathcal{M}_{a}^{\left( -\right) }$ and $\mathcal{%
M}_{a}^{\left( +\right) }$ be non-commuting while the product matrix $%
\mathcal{M}_{a}^{\left( -\right) }\mathcal{M}_{a}^{\left( +\right) }$ is
diagonalizable and with simple spectrum, then, for almost any value of $\eta
\in \mathbb{C}$ \ and of the inhomogeneity parameters satisfying  \rf{Inhomog-cond}, it holds:%
\begin{equation}
\langle t|t\rangle \neq 0,
\end{equation}%
where $|t\rangle $ and $\langle t|$ are the unique eigenvector and
eigenco-vector associated to $t(\lambda )$, a generic eigenvalue of $%
T(\lambda )$, and $T(\lambda )$ is diagonalizable
with simple spectrum.
\end{theorem}

\begin{proof}
Let us impose here:%
\begin{equation}
\xi _{a}=a\xi \text{ \ }\forall a\in \{1,...,\mathsf{N}\},
\end{equation}%
then it follows that $T(\xi _{l})$ are polynomials of degree $2%
\mathsf{N}+1$ in $\xi $ for all $l\in \{1,...,\mathsf{N}\}$ with maximal
degree coefficient given by:%
\begin{equation}
T_{l,2\mathsf{N}+1}\equiv d_{l,2\mathsf{N}+1}^{\left( +,-\right)
}\mathcal{M}_{l}^{\left( -\right) }\mathcal{M}_{l}^{\left( +\right) },\text{
\ with \ }d_{l,2\mathsf{N}+1}^{\left( +,-\right) }=l^{2}d_{l,2\mathsf{N}%
+1}/(\zeta _{+}\zeta _{-}).
\end{equation}%
The proof now proceed exactly as in the general Proposition 2.5 of  \cite{MaiN18}, for the rank $n-1$ fundamental representations of the $%
Y(gl_{n)}$ Yang-Baxter algebra. Indeed by assumption $T_{l,2\mathsf{N}%
+1}$ is diagonalizable and has simple spectrum and so we can just replace it to the
asymptotic operator $T_{l,\mathsf{N}-1}^{(K)}$ used in the proof of
Proposition 2.5 of \cite{MaiN18}.
\end{proof}

We can also give a more general characterization of the boundary conditions
leading to the diagonalizability and simplicity of the transfer matrix, as
it follows:

\begin{theorem}
Let the boundary matrix product $K_{+a}(\lambda )K_{-a}(\lambda )$ be simple
and diagonalizable, then for almost any choice of $\eta \in \mathbb{C}$, of
the inhomogeneity parameters satisfying  \rf{Inhomog-cond} and of the boundary parameters $\zeta _{\pm }$,
it holds:%
\begin{equation}
\langle t|t\rangle \neq 0,
\end{equation}%
where $|t\rangle $ and $\langle t|$ are the unique eigenvector and
eigenco-vector associated to $t(\lambda )$, a generic eigenvalue of $%
T(\lambda )$, and $T(\lambda )$ is diagonalizable
with simple spectrum.
\end{theorem}

\begin{proof}
Let us start observing that if the boundary matrix product $K_{+a}(\lambda
)K_{-a}(\lambda )$ is simple for a given value of the $\zeta _{\pm }$ then
it stays simple for almost any value of these parameters being $%
K_{+a}(\lambda )K_{-a}(\lambda )$ a polynomial of degree one in $1/\zeta
_{\pm }$. Moreover, for $K_{+a}(\lambda )K_{-a}(\lambda )$ simple and
diagonalizable, we also have that $K_{+,1}(\xi _{1})K_{-,1}(\xi _{1})$ is
simple and diagonalizable for almost all the values of $\xi _{1}$. Let us
now remark that the following identity holds: 
\begin{equation}
K_{1}^{(+,-)}(\alpha _{\pm },\mathcal{M}^{\left( \pm \right) })=\alpha
_{-}\alpha _{+-}K_{+,1}(\xi _{1})K_{-,1}(\xi _{1}),
\end{equation}%
with%
\begin{equation}
\alpha _{-}=\frac{\zeta _{-}}{\xi _{1}},\text{ \ }\alpha _{+}=\frac{\zeta
_{+}}{\xi _{1}-n\eta /2},
\end{equation}%
so that $K_{1}^{(+,-)}(\alpha _{\pm },\mathcal{M}^{\left( \pm \right) })$ is
simple and diagonalizable for almost all the values of $\xi _{1}$, $\alpha
_{\pm }$. This simplicity implies that the determinant%
\begin{equation}
\text{det}_{n}||\langle S,1|(K_{1}^{(+,-)}(\alpha _{\pm },\mathcal{M}%
^{\left( \pm \right) }))^{i}|e_{j}(1)\rangle||_{i,j\in \{1,...,n\}},
\end{equation}%
is nonzero for almost all the values of $\xi _{1}$, $\alpha _{\pm }$ and $%
\langle S,1|\in V_{1}$. Note that this determinant is a nonzero polynomial
in the $\alpha _{\pm }$, and so in the $\zeta _{\pm }$. Then we can always
find values of $\alpha _{\pm }$ such that the following determinants%
\begin{equation}
\text{det}_{n}||\langle S,l|(K_{l}^{(+,-)}(\alpha _{\pm },\mathcal{M}%
^{\left( \pm \right) }))^{i}|e_{j}(l)\rangle||_{i,j\in \{1,...,n\}}=\text{det}%
_{n}||\langle S,1|(K_{1}^{(+,-)}(l\alpha _{\pm },\mathcal{M}^{\left( \pm
\right) }))^{i}|e_{j}(1)\rangle||_{i,j\in \{1,...,n\}},
\end{equation}%
are nonzero for almost any value of the $\alpha _{\pm }$, and so of the $\zeta _{\pm }$,
being also  polynomials.

Let us now impose:%
\begin{equation}
\xi _{a}=a\xi \text{ \ }\forall a\in \{1,...,\mathsf{N}\},
\end{equation}%
the leading coefficient of $T(\xi _{l})\zeta _{+}\zeta _{-}$
reads:%
\begin{equation}
T_{l,2\mathsf{N}+1}=d_{l,2\mathsf{N}+1}K_{l}^{(+,-)}(\alpha
_{\pm },\mathcal{M}^{\left( \pm \right) }),
\end{equation}%
so that following the proof of the general Proposition 2.5 of \cite{MaiN18},
our statements hold being the operators $T_{l,2\mathsf{N}+1}$
diagonalizable and with simple spectrum on $V_{l}$ for any $l\in \{1,...,%
\mathsf{N}\}$ for almost any value of the $\zeta _{\pm }$.
\end{proof}

\section{Conclusion}

In this paper we have solved the longstanding open problem to define the quantum Separation of Variables for the class of integrable quantum models associated to the fundamental representations 
of the $Y(gl_n)$ reflection algebras. We have used the SoV basis to completely characterize the eigenvalue and eigenvector spectrum of the transfer matrix for the rank one and rank two cases and proven its 
equivalence to the so-called quantum spectral curve equation. The result on the construction of the SoV basis for any positive integer rank, indeed, allows us to extend the complete characterization of the transfer 
matrix spectrum as well as to introduce the quantum spectral curve characterization of it also to any higher rank $n$. In this article, we have seen explicitly how the results for the rational fundamental 
representation of the $Y(gl_2)$ reflection algebra can be used to prove the construction of the SoV basis for the general trigonometric case and how our new SoV basis allows for the characterization of the 
transfer matrix spectrum in these representations. This also includes their quantum spectral curve equation. The same results can be similarly derived  also for the fundamental 
representations of the trigonometric $U_q(gl_n)$ reflection algebras for any integer $n$. Our current investigations are both on completing the spectral analysis of other important quantum integrable 
models in our new SoV approach and to implement the analysis of the dynamics for the models already solved in this SoV framework. Here, the first fundamental step is the 
derivation of the scalar product formulae for the separate states of the type that we have derived for the rank one case in the appendix bellow. Such results should give  access to the computation of 
matrix elements of local operators on transfer matrix eigenstates, i.e. the first fundamental step toward the dynamics of quantum models in this higher rank cases.

\section*{Acknowledgements}
J. M. M. and G. N.  are supported by CNRS and ENS de Lyon.


\appendix

\section{Appendix: Scalar products in $Y(gl_{2})$ reflection algebra}

In this appendix we derive the scalar product of separate
states, which contain as particular cases the transfer matrix eigenvectors.
Let us comment that in the main text of the article we have shown that for
fundamental representations of $Y(gl_{2})$ reflection algebra associated to
general non-commuting boundary matrices $K_{+}(\lambda )$ and $K_{-}(\lambda
)$ our SoV basis can be reduced to the generalized Sklyanin's one under a
proper choice of the generating co-vector in the SoV basis. This observation
implies that for this set of representations the "measure" of the left/right
SoV vectors must coincide with the "Sklyanin's neasure" and so the scalar
product of separate states can be computed according to the known literature 
\cite{Nic12,KitMNT17,KitMNT18}. Here, we use this appendix to show how to
compute these scalar products directly in the framework of our new SoV
approach. This has the advantage to prove scalar product formulae also for
the representation associated to commuting boundary matrices, showing that
they keep the same form independently from the applicability of the
Sklyanin's original approach. Using the same type of computations 
presented in the following we can show that this same statement applies also for the
fundamental representations of $U_{q}(\widehat{gl}_{2})$ reflection algebra.
Hence, the results of  \cite{Nic12,KitMNT17,KitMNT18} hold
as well for diagonal boundary conditions and under conditions on the
parameters which make the generalized version of the Sklyanin's approach inapplicable.

\subsection{Construction of the right SoV basis orthogonal to the left one}

The following theorem allows to produce the orthogonal basis to the left SoV
basis and show that it is also of SoV type just using the polynomial form of
the transfer matrix and the fusion equation. Let us denote by $|S\rangle $
the nonzero vector orthogonal to all the SoV co-vectors with the exception of 
$\langle S|$, i.e.%
\begin{equation}
\langle h_{1},...,h_{\mathsf{N}}|S\rangle =\frac{\prod_{n=1}^{\mathsf{N}%
}\delta _{h_{n},0}}{N_{S}\widehat{V}(\xi _{1}^{\left( 0\right) },...,\xi _{%
\mathsf{N}}^{\left( 0\right) })}\text{ }\forall \{h_{1},...,h_{\mathsf{N}%
}\}\in \{0,1\}^{\otimes \mathsf{N}},
\end{equation}%
for some nonzero normalization $N_{S}$ and with 
\begin{equation}
\widehat{V}(x_{1},\ldots ,x_{\mathsf{N}})=\text{det}_{1\leq i,j\leq \mathsf{N%
}}[x_{i}^{2(j-1)}]=\prod_{1\leq k<j\leq \mathsf{N}}(x_{k}^{2}-x_{j}^{2}).
\label{VSQ}
\end{equation}%
Moreover, being the set of SoV co-vectors a basis, then $|S\rangle $ is
uniquely defined by the above normalization.

Similarly, we can introduce the nonzero vector $|\bar{S}\rangle $ orthogonal
to all the SoV co-vectors with the exception of $\langle 1,...,1|$, i.e.%
\begin{equation}
\langle h_{1},...,h_{\mathsf{N}}|\bar{S}\rangle =\frac{\prod_{n=1}^{\mathsf{N%
}}\delta _{h_{n},1}}{N_{S}\widehat{V}(\xi _{1}^{\left( 1\right) },...,\xi _{%
\mathsf{N}}^{\left( 1\right) })}\text{ \ }\forall \{h_{1},...,h_{\mathsf{N}%
}\}\in \{0,1\}^{\otimes \mathsf{N}},
\end{equation}%
which also defines completely $|\bar{S}\rangle $.

\begin{theorem}
Under the same conditions ensuring that the set of SoV co-vectors is a basis,
then the following set of vectors:%
\begin{equation}
|h_{1},...,h_{\mathsf{N}}\rangle =\prod_{a=1}^{\mathsf{N}} [\frac{T(\xi
_{a}+\eta /2)}{\text{\textsc{k}}_{a}\,\mathsf{A}_{\bar{\zeta}_{+},\bar{\zeta}%
_{-}}(\eta /2-\xi _{a})} ]^{h_{a}}|S\rangle \text{ \ }\forall \{h_{1},...,h_{%
\mathsf{N}}\}\in \{0,1\}^{\otimes \mathsf{N}}
\label{First-R-SoV}
\end{equation}%
forms an orthogonal basis to the left SoV basis:%
\begin{equation}
\langle h_{1},...,h_{\mathsf{N}}|k_{1},...,k_{\mathsf{N}}\rangle =\frac{%
\prod_{n=1}^{\mathsf{N}}\delta _{h_{n},k_{n}}}{N_{S}\widehat{V}(\xi
_{1}^{(h_{1})},...,\xi _{\mathsf{N}}^{(h_{\mathsf{N}})})}.
\end{equation}%
Let $t(\lambda )$ be a transfer matrix eigenvalue, $t(\lambda )\in \Sigma
_{T}$, then the uniquely defined eigenvector $|t\rangle $ and co-vectors $%
\langle t|$ admit the following SoV representations:%
\begin{eqnarray}
|t\rangle  &=&\sum_{h_{1},...,h_{\mathsf{N}}=0}^{1}\prod_{a=1}^{\mathsf{N}} [%
\frac{t(\xi _{a}-\eta /2)}{\mathsf{A}_{\bar{\zeta}_{+},\bar{\zeta}_{-}}(\eta
/2-\xi _{n})} ]^{1-h_{a}}\text{ }\widehat{V}(\xi _{1}^{(h_{1})},...,\xi _{%
\mathsf{N}}^{(h_{\mathsf{N}})})|h_{1},...,h_{\mathsf{N}}\rangle , \\
\langle t| &=&\sum_{h_{1},...,h_{\mathsf{N}}=0}^{1}\prod_{a=1}^{\mathsf{N}} [%
\frac{t(\xi _{a}+\eta /2)}{\text{\textsc{k}}_{a}\,\mathsf{A}_{\bar{\zeta}%
_{+},\bar{\zeta}_{-}}(\eta /2-\xi _{n})} ]^{h_{a}}\text{ }\widehat{V}(\xi
_{1}^{(h_{1})},...,\xi _{\mathsf{N}} ^{(h_{\mathsf{N}})})\langle h_{1},...,h_{%
\mathsf{N}}|,
\end{eqnarray}%
where we have set their normalization by imposing:%
\begin{equation}
\langle S|t\rangle =\langle t|S\rangle =1/N_{S}.  \label{Normalization}
\end{equation}
\end{theorem}

\begin{proof}
Let us start proving the orthogonality condition: 
\begin{equation}
\langle h_{1},...,h_{\mathsf{N}}|k_{1},...,k_{\mathsf{N}}\rangle =0\text{ \
for }\forall \{k_{1},...,k_{\mathsf{N}}\}\neq \{h_{1},...,h_{\mathsf{N}%
}\}\in \{0,1\}^{\otimes \mathsf{N}}.
\end{equation}%
The proof is done by induction. For any vector $|k_{1},...,k_{\mathsf{N}}\rangle$ let us denote $l=\sum_{n=1}^{\mathsf{N}}k_{n}$. The property is obviously true for $l=0$.  Assuming that it is true for any vector $%
|k_{1},...,k_{\mathsf{N}}\rangle $ with $\sum_{n=1}^{\mathsf{N}}k_{n}=l$ for some 
$l\leq \mathsf{N}-1$ let us prove it holds for vectors $|k_{1}^{\prime },...,k_{%
\mathsf{N}}^{\prime }\rangle $ with $\sum_{n=1}^{\mathsf{N}}k_{n}^{\prime
}=l+1$. To this aim we fix a vector $|k_{1},...,k_{\mathsf{N}}\rangle $ with 
$\sum_{n=1}^{\mathsf{N}}k_{n}=l$ and we denote by $\pi $ a permutation on
the set $\{1,...,\mathsf{N}\}$ such that:%
\begin{equation}
k_{\pi (a)}=1\text{ for }a\leq l\text{\ \ and }k_{\pi (a)}=0\text{ for }l<a,
\end{equation}%
and then we compute:%
\begin{equation}
\langle h_{1},...,h_{\mathsf{N}}|T(\xi _{\pi (l+1)}^{\left( 0\right)
})|k_{1},...,k_{\mathsf{N}}\rangle =\text{\textsc{k}}_{a}\,\mathsf{A}_{\bar{%
\zeta}_{+},\bar{\zeta}_{-}}(-\xi _{n}^{(1)})\langle h_{1},...,h_{\mathsf{N}%
}|k_{1}^{\prime },...,k_{\mathsf{N}}^{\prime }\rangle\,,
\end{equation}%
where we have defined: 
\begin{equation}
k_{\pi (a)}^{\prime }=k_{\pi (a)}\text{ }\forall a\in \{1,...,\mathsf{N}%
\}\backslash \{l+1\}\text{ and }k_{\pi (l+1)}^{\prime }=1,
\end{equation}%
for any $\{h_{1},...,h_{\mathsf{N}}\}\neq \{k_{1}^{\prime },...,k_{\mathsf{N}%
}^{\prime }\}\in \{0,1\}^{\otimes \mathsf{N}}$. There are two cases, the
first case is $h_{\pi (l+1)}=0$, then it holds:%
\begin{equation}
\langle h_{1},...,h_{\mathsf{N}}|T(\xi _{\pi (l+1)}^{\left( 0\right)
})|k_{1},...,k_{\mathsf{N}}\rangle =\mathsf{A}_{\bar{\zeta}_{+},\bar{\zeta}%
_{-}}(\xi _{n}^{(0)})\langle h_{1}^{\prime },...,h_{\mathsf{N}}^{\prime
}|k_{1},...,k_{\mathsf{N}}\rangle ,  \label{Caso-1}
\end{equation}%
where we have defined: 
\begin{equation}
h_{\pi (a)}^{\prime }=h_{\pi (a)}\text{ }\forall a\in \{1,...,\mathsf{N}%
\}\backslash \{l+1\}\text{ and }h_{\pi (l+1)}^{\prime }=1.
\end{equation}%
Then from $\{h_{1},...,h_{\mathsf{N}}\}\neq \{k_{1}^{\prime },...,k_{\mathsf{%
N}}^{\prime }\}\in \{0,1\}^{\otimes \mathsf{N}}$ it follows also that $%
\{h_{1}^{\prime },...,h_{\mathsf{N}}^{\prime }\}\neq \{k_{1},...,k_{\mathsf{N%
}}\}\in \{0,1\}^{\otimes \mathsf{N}}$ and so the induction hypothesis implies that the
r.h.s. of (\ref{Caso-1}) is zero and so we get:%
\begin{equation}
\langle h_{1},...,h_{\mathsf{N}}|k_{1}^{\prime },...,k_{\mathsf{N}}^{\prime
}\rangle =0.  \label{l+1-Id}
\end{equation}%
The second case is $h_{\pi (l+1)}=1$, then we can use the following
interpolation formula:%
\begin{equation}
T(\xi _{\pi (l+1)}^{\left( 0\right) })=t_{\mathsf{N}+1}u_{\mathbf{h}}(\xi
_{\pi (l+1)}^{\left( 0\right) })+t(\eta /2)s_{\mathbf{h}}(\xi _{\pi
(l+1)}^{\left( 0\right) })+\sum_{a=1}^{\mathsf{N}}r_{a,\mathbf{h}}(\xi _{\pi
(l+1)}^{\left( 0\right) })T(\xi _{a}^{\left( h_{a}\right) }),
\end{equation}%
from which $\langle h_{1},...,h_{\mathsf{N}}|T(\xi _{\pi (l+1)}^{\left(
0\right) })|k_{1},...,k_{\mathsf{N}}\rangle $ reduces to the following sum:%
\begin{align}
& \left( t_{\mathsf{N}+1}u_{\mathbf{h}}(\xi _{\pi (l+1)}^{\left( 0\right)
})+t(\eta /2)s_{\mathbf{h}}(\xi _{\pi (l+1)}^{\left( 0\right) })\right)
\langle h_{1},...,h_{\mathsf{N}}|k_{1},...,k_{\mathsf{N}}\rangle  \notag \\
& +\sum_{a=1}^{\mathsf{N}}r_{a,\mathbf{h}}(\xi _{\pi (l+1)}^{\left( 0\right)
})\left( \mathsf{A}_{\bar{\zeta}_{+},\bar{\zeta}_{-}}(\xi _{n}^{(0)})\right)
^{1-h_{\pi (a)}}\left( \mathsf{A}_{\bar{\zeta}_{+},\bar{\zeta}_{-}}(-\xi
_{n}^{(1)})\right) ^{h_{\pi (a)}}\langle h_{1}^{(a)},...,h_{\mathsf{N}%
}^{(a)}|k_{1},...,k_{\mathsf{N}}\rangle ,
\end{align}%
where we have defined:%
\begin{equation}
h_{\pi (j)}^{(a)}=h_{\pi (j)}\text{ }\forall j\in \{1,...,\mathsf{N}%
\}\backslash \{a\}\text{ and }h_{\pi (a)}^{(a)}=1-h_{\pi (a)}.
\end{equation}%
Let us now note that from $h_{\pi (l+1)}=1$ it follows that $\{h_{1},...,h_{%
\mathsf{N}}\}\neq \{k_{1},...,k_{\mathsf{N}}\}$ as $k_{\pi (l+1)}=0$ by
definition and similarly $\{h_{1}^{(a)},...,h_{\mathsf{N}}^{(a)}\}\neq
\{k_{1},...,k_{\mathsf{N}}\}$ being by definition $h_{\pi
(l+1)}^{(a)}=h_{\pi (l+1)}=1$ for any $a\in \{1,...,\mathsf{N}\}\backslash
\{l+1\}$. Finally from $\{h_{1},...,h_{\mathsf{N}}\}\neq \{k_{1}^{\prime
},...,k_{\mathsf{N}}^{\prime }\}$ with $h_{\pi (l+1)}=k_{\pi (l+1)}^{\prime
}=1$ clearly it follows that $\{h_{1}^{(l+1)},...,h_{\mathsf{N}%
}^{(l+1)}\}\neq \{k_{1},...,k_{\mathsf{N}}\}$. So by using the induction
argument we get that all the terms in the above sum are equal to zero. So that also
in the case $h_{\pi (l+1)}=1$, we get that (\ref{l+1-Id}) is satisfied, and
so it is satisfied for any $\{h_{1},...,h_{\mathsf{N}}\}\neq \{k_{1}^{\prime
},...,k_{\mathsf{N}}^{\prime }\}$ which proves the induction of the
orthogonality to $l+1$. Indeed, by changing the permutation $\pi $ we can
both take for $\{\pi (1),...,\pi (l)\}$ any subset of cardinality $l$ in $%
\{1,...,\mathsf{N}\}$ and with $\pi (l+1)$ any element in its complement $%
\{1,...,\mathsf{N}\}\backslash \{\pi (1),...,\pi (l)\}$.

We can compute now the left/right normalization, and to do this we just need
to compute the following type of ratio:%
\begin{equation}
\frac{\langle h_{1}^{(a)},...,h_{\mathsf{N}}^{(a)}|h_{1}^{(a)},...,h_{%
\mathsf{N}}^{(a)}\rangle }{\langle \bar{h}_{1}^{(a)},...,\bar{h}_{\mathsf{N}%
}^{(a)}|\bar{h}_{1}^{(a)},...,\bar{h}_{\mathsf{N}}^{(a)}\rangle }=\mathsf{A}%
_{\bar{\zeta}_{+},\bar{\zeta}_{-}}(-\xi _{n}^{(1)})\frac{\langle
h_{1}^{(a)},...,h_{\mathsf{N}}^{(a)}|h_{1}^{(a)},...,h_{\mathsf{N}%
}^{(a)}\rangle }{\langle h_{1}^{(a)},...,h_{\mathsf{N}}^{(a)}|T(\xi
_{a}^{(1)})|\bar{h}_{1}^{(a)},...,\bar{h}_{\mathsf{N}}^{(a)}\rangle }
\end{equation}%
with $\bar{h}_{j}^{(a)}=h_{j}^{(a)}$ for any $j\in \{1,...,\mathsf{N}%
\}\backslash \{a\}$ while $\bar{h}_{a}^{(a)}=0$ and $h_{a}^{(a)}=1$. We can
use now once again the interpolation formula $(\ref{Interp-T-open})$
computed in $\lambda =\xi _{a}^{(1)}$ which by the orthogonality condition
produces only one non-zero term, the one associate to $T(\xi _{a}^{(0)})$,
i.e. it holds:%
\begin{equation}
\frac{\langle h_{1}^{(a)},...,h_{\mathsf{N}}^{(a)}|h_{1}^{(a)},...,h_{%
\mathsf{N}}^{(a)}\rangle }{\langle \bar{h}_{1}^{(a)},...,\bar{h}_{\mathsf{N}%
}^{(a)}|\bar{h}_{1}^{(a)},...,\bar{h}_{\mathsf{N}}^{(a)}\rangle }=\frac{1}{%
\text{\textsc{k}}_{a}r_{a,\mathbf{\bar{h}}}(\xi _{a}^{(1)})}=\prod_{b\neq
a,b=1}^{\mathsf{N}}\frac{(\xi _{a}^{(0)})^{2}-(\xi _{b}^{(h_{b})})^{2}}{(\xi
_{a}^{(1)})^{2}-(\xi _{b}^{(h_{b})})^{2}}.
\end{equation}%
It is now standard \cite{GroMN12,Nic13} to get the Vandermonde
determinant for the normalization once we use the above result.

Let us note thatg the set of SoV co-vectors and vectors  being both basis, it
follows that for any transfer matrix eigenstates $|t\rangle $ and $\langle t|
$ there exist at least a $\{r_{1},...,r_{\mathsf{N}}\}\in \{0,1\}^{\otimes 
\mathsf{N}}$ and a $\{s_{1},...,s_{\mathsf{N}}\}\in \{0,1\}^{\otimes \mathsf{%
N}}$ such that:%
\begin{equation}
\langle r_{1},...,r_{\mathsf{N}}|t\rangle \neq 0,\text{ }\langle
t|s_{1},...,s_{\mathsf{N}}\rangle \neq 0,
\end{equation}%
which together with the identities:%
\begin{equation}
\langle h_{1},...,h_{\mathsf{N}}|t\rangle \propto \langle S|t\rangle ,\text{ 
}\langle t|h_{1},...,h_{\mathsf{N}}\rangle \propto \langle t|\bar{S}\rangle 
\text{ \ }\forall \{h_{1},...,h_{\mathsf{N}}\}\in \{0,1\}^{\otimes \mathsf{N}%
},
\end{equation}%
imply that:%
\begin{equation}
\langle S|t\rangle \neq 0,\text{ }\langle t|\bar{S}\rangle \neq 0\text{,}
\end{equation}%
so that we are free to fix the normalization of the eigenstates $|t\rangle $
and $\langle t|$ by (\ref{Normalization}). Finally, the representations for
these left and right transfer matrix eigenvectors follow from the use of the
SoV decomposition of the identity:%
\begin{equation}
\mathbb{I}=\text{ }N_{S}\sum_{h_{1},...,h_{\mathsf{N}}=0}^{1}\widehat{V}(\xi
_{1}^{(h_{1})},...,\xi _{\mathsf{N}}^{(h_{\mathsf{N}})})|h_{1},...,h_{%
\mathsf{N}}\rangle \langle h_{1},...,h_{\mathsf{N}}|.
\end{equation}
\end{proof}

\begin{corollary}
Under the same conditions ensuring that the set of SoV co-vectors is a basis,
then the vectors of the right SoV basis admit also the following
representations:%
\begin{equation}
|h_{1},...,h_{\mathsf{N}}\rangle =\prod_{a=1}^{\mathsf{N}}(\frac{\text{%
\textsc{k}}_{a}\,T(\xi _{a}-\eta /2)}{\mathsf{A}_{\bar{\zeta}_{+},\bar{\zeta}%
_{-}}(\xi _{a}+\eta /2)})^{1-h_{a}}|\bar{S}\rangle \text{ \ }\forall
\{h_{1},...,h_{\mathsf{N}}\}\in \{0,1\}^{\otimes \mathsf{N}},
\end{equation}%
as well as for any element of the spectrum of $T(\lambda ,\{\xi \})$ the
unique associated eigenco-vector $\langle t|$ admit the following SoV
representations:%
\begin{equation}
\langle t|=\mathsf{N}_{t}\sum_{h_{1},...,h_{\mathsf{N}}=0}^{1}\prod_{a=1}^{%
\mathsf{N}}(\frac{\text{\textsc{k}}_{a}\,t(\xi _{a}-\eta /2)}{\mathsf{A}_{%
\bar{\zeta}_{+},\bar{\zeta}_{-}}(\xi _{a}+\eta /2)})^{1-h_{a}}\widehat{V}%
(\xi _{1}^{(h_{1})},...,\xi _{\mathsf{N}}^{(h_{\mathsf{N}})})\langle
h_{1},...,h_{\mathsf{N}}|,
\end{equation}%
where we have defined:%
\begin{equation}
\mathsf{N}_{t}=\langle t|\bar{S}\rangle =\frac{1}{N_{S}}\prod_{a=1}^{\mathsf{%
N}}\frac{\mathsf{A}_{\bar{\zeta}_{+},\bar{\zeta}_{-}}(\xi _{a}+\eta /2)}{%
\text{\textsc{k}}_{a}\,t(\xi _{a}-\eta /2)}\neq 0,  \label{Normalization+}
\end{equation}%
once we fix the normalization by (\ref{Normalization}).
\end{corollary}

\begin{proof}
Taking into account the chosen normalizations clearly it holds:%
\begin{equation}
|\bar{S}\rangle =|h_{1}=1,...,h_{\mathsf{N}}=1\rangle =\prod_{a=1}^{\mathsf{N%
}}\frac{T(\xi _{a}+\eta /2)}{\text{\textsc{k}}_{a}\,\mathsf{A}_{\bar{\zeta}%
_{+},\bar{\zeta}_{-}}(\eta /2-\xi _{n})}|S\rangle ,
\end{equation}%
so that:%
\begin{align}
\prod_{a=1}^{\mathsf{N}}(\frac{\text{\textsc{k}}_{a}\,T(\xi _{a}-\eta /2)}{%
\mathsf{A}_{\bar{\zeta}_{+},\bar{\zeta}_{-}}(\xi _{n}+\eta /2)})^{1-h_{a}}|%
\bar{S}\rangle & =\prod_{a=1}^{\mathsf{N}}(\frac{\text{\textsc{k}}%
_{a}\,T(\xi _{a}-\eta /2)}{\mathsf{A}_{\bar{\zeta}_{+},\bar{\zeta}_{-}}(\xi
_{n}+\eta /2)})^{1-h_{a}}\frac{T(\xi _{a}+\eta /2)}{\text{\textsc{k}}_{a}\,%
\mathsf{A}_{\bar{\zeta}_{+},\bar{\zeta}_{-}}(\eta /2-\xi _{n})}|S\rangle 
\notag \\
& =\prod_{a=1}^{\mathsf{N}}(\frac{T(\xi _{a}-\eta /2)T(\xi _{a}+\eta /2)}{%
\mathsf{A}_{\bar{\zeta}_{+},\bar{\zeta}_{-}}(\xi _{n}+\eta /2)\mathsf{A}_{%
\bar{\zeta}_{+},\bar{\zeta}_{-}}(\eta /2-\xi _{n})})^{1-h_{a}}  \notag \\
& \times (\frac{T(\xi _{a}+\eta /2)}{\text{\textsc{k}}_{a}\,\mathsf{A}_{\bar{%
\zeta}_{+},\bar{\zeta}_{-}}(\eta /2-\xi _{n})})^{h_{a}}|S\rangle  \notag \\
& =|h_{1},...,h_{\mathsf{N}}\rangle ,
\end{align}%
by the fusion identities (\ref{gl_2-fusion}). From this representation of
the right SoV vectors and from the original one in (\ref{First-R-SoV}), it
follows that: 
\begin{equation}
\langle t|\bar{S}\rangle \equiv \langle t|1,...,1\rangle =\prod_{a=1}^{%
\mathsf{N}}\frac{t(\xi _{a}+\eta /2)}{\text{\textsc{k}}_{a}\,\mathsf{A}_{%
\bar{\zeta}_{+},\bar{\zeta}_{-}}(\eta /2-\xi _{n})}\langle t|S\rangle ,
\end{equation}%
from which our result follows.
\end{proof}

\subsection{Algebraic Bethe Ansatz form of separate states}

Let us rewrite the left and right transfer matrix eigenstates in terms of
the $Q$-functions. The following corollary holds:

\begin{corollary}
Under the same conditions ensuring that the set of SoV co-vectors is a basis,
then for any element of the spectrum of $T(\lambda )$ the unique associated
eigenvector $|t\rangle $ admit the following SoV representations:%
\begin{align}
& |t\rangle =\sum_{\mathbf{h}\in \{0,1\}^{\mathsf{N}}}\prod_{n=1}^{\mathsf{N}%
}Q_{t}(\xi _{n}^{(h_{n})})\ \widehat{V}(\xi _{1}^{(h_{1})},\ldots ,\xi _{%
\mathsf{N}}^{(h_{\mathsf{N}})})\,\,|\,h_{1},...,h_{\mathsf{N}}\,\rangle ,
\label{eigenT-right} \\
& \langle \,t\,|=\sum_{\mathbf{h}\in \{0,1\}^{\mathsf{N}}}\prod_{n=1}^{%
\mathsf{N}}\left[ \left( \frac{\xi _{n}-\eta }{\xi _{n}+\eta }\frac{\mathsf{A%
}_{\bar{\zeta}_{+},\bar{\zeta}_{-}}(\xi _{n}^{(0)})}{\mathsf{A}_{\bar{\zeta}%
_{+},\bar{\zeta}_{-}}(-\xi _{n}^{(1)})}\right) ^{\!h_{n}}Q_{t}(\xi
_{n}^{(h_{n})})\right]   \notag \\
& \hspace{7cm}\times \widehat{V}(\xi _{1}^{(h_{1})},\ldots ,\xi _{\mathsf{N}%
}^{(h_{\mathsf{N}})})\,\langle h_{1},...,h_{\mathsf{N}}\,|\,.
\label{eigenT-left}
\end{align}
\end{corollary}

Let us remark that these representations for the left and right transfer
matrix eigenstates formally coincide with that obtained in the schema of the
generalized Sklyanin's SoV approach \cite{KitMNT16} even when this last
approach does not apply and without any requirement on the form of the
co-vector $\langle \,S\,|$. One can introduce the following class of left and
right separate states:%
\begin{align}
\langle \,\alpha \,|& =\sum_{h_{1},...,h_{\mathsf{N}}=0}^{1}\prod_{a=1}^{%
\mathsf{N}}\alpha (\xi _{a}^{(h_{a})})\widehat{V}(\xi _{1}^{(h_{1})},...,\xi
_{\mathsf{N}}^{(h_{\mathsf{N}})})\langle h_{1},...,h_{\mathsf{N}}|,
\label{L_S-state} \\
|\,\beta \,\rangle & =\sum_{h_{1},...,h_{\mathsf{N}}=0}^{1}\prod_{a=1}^{%
\mathsf{N}}\left( \frac{\xi _{n}-\eta }{\xi _{n}+\eta }\frac{\mathsf{A}_{%
\bar{\zeta}_{+},\bar{\zeta}_{-}}(\xi _{n}^{(0)})}{\mathsf{A}_{\bar{\zeta}%
_{+},\bar{\zeta}_{-}}(-\xi _{n}^{(1)})}\right) ^{\!h_{n}}\beta (\xi
_{a}^{(h_{a})})\text{ }\widehat{V}(\xi _{1}^{(h_{1})},...,\xi _{\mathsf{N}%
}^{(h_{\mathsf{N}})})|h_{1},...,h_{\mathsf{N}}\rangle ,  \label{R_S-state}
\end{align}%
where $\alpha (\lambda )$ and $\beta (\lambda )$ are generic functions. It
is then clear by the previous corollary that the left and right transfer
matrix eigenstates are special elements in these classes.

Let us now introduce the one parameter family of commuting operators by:%
\begin{equation}
\mathbb{B}\left( \lambda \right) =\text{ }N_{S}\sum_{h_{1},...,h_{\mathsf{N}%
}=0}^{1}b_{h_{1},...,h_{\mathsf{N}}}(\lambda )\widehat{V}(\xi
_{1}^{(h_{1})},...,\xi _{\mathsf{N}}^{(h_{\mathsf{N}})})|h_{1},...,h_{%
\mathsf{N}}\rangle \langle h_{1},...,h_{\mathsf{N}}|,
\end{equation}%
where we have defined:%
\begin{equation}
b_{h_{1},...,h_{\mathsf{N}}}(\lambda )=\prod_{a=1}^{\mathsf{N}}(\lambda
^{2}-(\xi _{a}^{(h_{a})})^{2}).
\end{equation}%
Clearly, if the two boundary matrices are non simultaneously diagonalizable
and we take the special choice $\langle S|=\langle 0|\mathcal{W}_{K}^{-1}$
then it holds:%
\begin{equation}
\widehat{\mathcal{B}}_{-}(\lambda )=(-1)^{\mathsf{N}}\,\bar{\mathsf{b}}_{-}%
\frac{\lambda -\eta /2}{\bar{\zeta}_{-}}\mathbb{B}\left( \lambda \right) .
\end{equation}%
Let us assume that $\alpha (\lambda )$ be the following polynomial: 
\begin{equation}
\alpha (\lambda )=\prod_{k=1}^{R}(\lambda ^{2}-\alpha _{k}^{2}),
\end{equation}%
then the left and right separate states $\langle \,\alpha \,|$ and $%
|\,\alpha \,\rangle $ associated admits the following Algebraic Bethe Ansatz
form: 
\begin{equation}
\langle \,\alpha \,|=(-1)^{R\mathsf{N}}\langle \,1\,|\prod_{k=1}^{R}\mathbb{B%
}(\alpha _{k}),\qquad |\,\alpha \,\rangle =(-1)^{R\mathsf{N}}\prod_{k=1}^{R}%
\mathbb{B}(\alpha _{k})|\,1\,\rangle ,  \label{ABA_SoV-form}
\end{equation}%
where we have defined $\langle \,1\,|$ and $|\,1\,\rangle $ to be the
separate co-vector and vector associated to the identity polynomial:%
\begin{align}
\langle \,1\,|& =\sum_{h_{1},...,h_{\mathsf{N}}=0}^{1}\widehat{V}(\xi
_{1}^{(h_{1})},...,\xi _{\mathsf{N}}^{(h_{\mathsf{N}})})\langle h_{1},...,h_{%
\mathsf{N}}|, \\
|\,1\,\rangle & =\sum_{h_{1},...,h_{\mathsf{N}}=0}^{1}\prod_{a=1}^{\mathsf{N}%
}\left( \frac{\xi _{n}-\eta }{\xi _{n}+\eta }\frac{\mathsf{A}_{\bar{\zeta}%
_{+},\bar{\zeta}_{-}}(\xi _{n}^{(0)})}{\mathsf{A}_{\bar{\zeta}_{+},\bar{\zeta%
}_{-}}(-\xi _{n}^{(1)})}\right) ^{\!h_{n}}\text{ }\widehat{V}(\xi
_{1}^{(h_{1})},...,\xi _{\mathsf{N}}^{(h_{\mathsf{N}})})|h_{1},...,h_{%
\mathsf{N}}\rangle .
\end{align}

\subsection{Scalar product of separate states}

Let us consider a couple of separate states $\langle \,\alpha \,|$ and $%
|\,\beta \,\rangle $, then it holds:%
\begin{align}
\langle \,\alpha \,|\,\beta \,\rangle & =\sum_{h_{1},...,h_{\mathsf{N}%
}=0}^{1}\prod_{a=1}^{\mathsf{N}}\left( \frac{\xi _{n}-\eta }{\xi _{n}+\eta }%
\frac{\mathsf{A}_{\bar{\zeta}_{+},\bar{\zeta}_{-}}(\xi _{n}^{(0)})}{\mathsf{A%
}_{\bar{\zeta}_{+},\bar{\zeta}_{-}}(-\xi _{n}^{(1)})}\right)
^{\!h_{n}}\alpha (\xi _{a}^{(h_{a})})\beta (\xi _{a}^{(h_{a})})\frac{%
\widehat{V}(\xi _{1}^{(h_{1})},...,\xi _{\mathsf{N}}^{(h_{\mathsf{N}})})}{%
N_{S}} \\
& =\sum_{h_{1},...,h_{\mathsf{N}}=0}^{1}\prod_{a=1}^{\mathsf{N}}\left(
g_{n}f_{n}\right) ^{\!h_{n}}\alpha (\xi _{a}^{(h_{a})})\beta (\xi
_{a}^{(h_{a})})\frac{\widehat{V}(\xi _{1}^{(h_{1})},...,\xi _{\mathsf{N}%
}^{(h_{\mathsf{N}})})}{N_{S}},  \label{Sp-al-be-1}
\end{align}%
where:%
\begin{equation}
g_{n}\equiv g_{\bar{\zeta}_{+},\bar{\zeta}_{-}}(\xi _{n})=\frac{(\xi _{n}+%
\bar{\zeta}_{+})(\xi _{n}+\bar{\zeta}_{-})}{(\xi _{n}-\bar{\zeta}_{+})(\xi
_{n}-\bar{\zeta}_{-})},  \label{g_n}
\end{equation}%
and 
\begin{align}
f_{n}\equiv f(\xi _{n},\{\xi \})& =-\prod_{\substack{ a=1 \\ a\neq n}}^{%
\mathsf{N}}\frac{(\xi _{n}-\xi _{a}+\eta )(\xi _{n}+\xi _{a}+\eta )}{(\xi
_{n}-\xi _{a}-\eta )(\xi _{n}+\xi _{a}-\eta )}  \notag \\
& =-\prod_{\substack{ a=1 \\ a\neq n}}^{\mathsf{N}}\frac{\left[ (\xi
_{n}^{(0)})^{2}-(\xi _{a}^{(1)})^{2}\right] \left[ (\xi _{n}^{(0)})^{2}-(\xi
_{a}^{(0)})^{2}\right] }{\left[ (\xi _{n}^{(1)})^{2}-(\xi _{a}^{(1)})^{2}%
\right] \left[ (\xi _{n}^{(1)})^{2}-(\xi _{a}^{(0)})^{2}\right] }.
\label{f_n}
\end{align}%
So that in our general SoV approach the scalar product of separate states
admits the same representations which hold for the separate states in the
generalized Sklyanin's approach, as one can directly infer comparing $\left( %
\ref{Sp-al-be-1}\right) $ with the formula (4.12) of \cite{KitMNT17}.
Moreover, our current result is not limited to the cases of non-commuting
boundary matrices, where the generalized Sklyanin's approach applies. In
particular, setting the normalization as:%
\begin{equation}
N_{S}=\widehat{V}(\xi _{1},...,\xi _{\mathsf{N}})\frac{\widehat{V}(\xi
_{1}^{\left( 0\right) },...,\xi _{\mathsf{N}}^{\left( 0\right) })}{\widehat{V%
}(\xi _{1}^{\left( 1\right) },...,\xi _{\mathsf{N}}^{\left( 1\right) })}%
\prod_{n=1}^{\mathsf{N}}\frac{\xi _{n}}{\xi _{n}-\bar{\zeta}_{-}},
\end{equation}%
we obtain 
\begin{equation}
\langle \,\alpha \,|\,\beta \,\rangle =\prod_{n=1}^{\mathsf{N}}\frac{\xi
_{n}-\bar{\zeta}_{-}}{\xi _{n}}\sum_{h_{1},...,h_{\mathsf{N}%
}=0}^{1}\prod_{a=1}^{\mathsf{N}}\left( -g_{n}\right) ^{\!h_{n}}\alpha (\xi
_{a}^{(h_{a})})\beta (\xi _{a}^{(h_{a})})\frac{\hat{V}(\xi
_{1}^{(1-h_{1})},...,\xi _{\mathsf{N}}^{(1-h_{\mathsf{N}})})}{\hat{V}(\xi
_{1},...,\xi _{\mathsf{N}})},
\end{equation}%
which coincides with the formula (4.13) of \cite{KitMNT17}, up to the non
required prefactor $1/\bar{\mathsf{b}}_{-}$. This means that we can use in
our more general SoV framework the manipulation of these scalar product
formulae to obtain Izergin and Slavnov type scalar products \cite%
{Tsu98,Ize87,Sla89,FodW12} and the generalized Gaudin type formula as done
in \cite{KitMNT17}. The same statements apply as well to the trigonometric
case in comparison with the results in \cite{KitMNT18}.


{\small
\begin{thebibliography}{999}

\bibitem{MaiN18}
J.~M. Maillet, and G.~Niccoli.
\newblock {\em On quantum separation of variables}
\newblock J. Math. Phys. {\bf 59} (2018) 091417.

\bibitem{MaiN18a} J.~M. Maillet, and G.~Niccoli. \newblock {\em Complete spectrum
of quantum integrable lattice models associated to $Y(gl_n)$ by separation of variables} 
\newblock arXiv:1810.11885.

\bibitem{MaiN18b} 
J.~M. Maillet, and G.~Niccoli.
\newblock {\em Complete spectrum of quantum integrable lattice models associated to $U_q(\hat{gl}_n)$ by separation
of variables.}
\newblock arXiv:1811.08405.

\bibitem{MaiN19} J.~M. Maillet, and G.~Niccoli. \newblock {\em On quantum
separation of variables: beyond fundamental representations}, \newblock arXiv:1903.06618.


\bibitem{FadS78}
L.~D. Faddeev and E.~K. Sklyanin.
\newblock {\em Quantum-mechanical approach to completely integrable field theory
  models.}
\newblock Sov. Phys. Dokl., 23:902--904, 1978.

\bibitem{FadST79}
L.~D. Faddeev, E.~K. Sklyanin, and L.~A. Takhtajan.
\newblock {\em Quantum inverse problem method {I}.}
\newblock Theor. Math. Phys., 40:688--706, 1979.
\newblock Translated from Teor. Mat. Fiz. 40 (1979) 194-220.

\bibitem{FadT79} L.~A. Takhtadzhan and L.~D. Faddeev, \newblock \emph{The
quantum method of the inverse problem and the {H}eisenberg {XYZ} model}, %
\newblock \href{http://dx.doi.org/10.1070/RM1979v034n05ABEH003909}{Russ.
Math. Surveys \textbf{34}(5), 11 (1979)}.

\bibitem{Skl79}
E.~K. Sklyanin.
\newblock {\em Method of the inverse scattering problem and the non-linear quantum
  {S}chr{\"o}dinger equation.}
\newblock Sov. Phys. Dokl. , 24:107--109, 1979.

\bibitem{Skl79a}
E.~K. Sklyanin.
\newblock {\em On complete integrability of the {L}andau-{L}ifshitz equation.}
\newblock Preprint LOMI E-3-79, 1979.

\bibitem{FadT81}
L.~D. Faddeev and L.~A. Takhtajan.
\newblock {\em Quantum inverse scattering method.}
\newblock Sov. Sci. Rev. Math., C 1:107, 1981.

\bibitem{Skl82}
E.~K. Sklyanin.
\newblock {\em Quantum version of the inverse scattering problem method.}
\newblock J. Sov. Math., 19:1546--1595, 1982.

\bibitem{Fad82}
L.~D. Faddeev.
\newblock {\em Integrable models in $(1 + 1)$-dimensional quantum field theory.}
\newblock In J.~B. Zuber and R.~Stora, editors, {\em Les Houches 1982, Recent
  advances in field theory and statistical mechanics}, pages 561--608. Elsevier
  Science Publ., 1984.

\bibitem{Fad96}
L.~D. Faddeev.
\newblock {\em How algebraic {B}ethe ansatz works for integrable model.}
\newblock Les Houches Lectures, 1996.


\bibitem{Che84} I.~V. Cherednik, \newblock \emph{Factorizing particles on a
half line and root systems}, \newblock  \href{http://dx.doi.org/10.1007/BF01038545%
}{Theor. Math. Phys. \textbf{61}, 977 (1984)}.

\bibitem{Skl88} E.~K. Sklyanin, \newblock \emph{Boundary conditions for
integrable quantum systems}, \newblock \href{http://dx.doi.org/10.1088/0305-4470/21/10/015%
}{J. Phys. A: Math. Gen. \textbf{21}, 2375 (1988)}.



\bibitem{Gau71a} M.~Gaudin, \newblock \emph{Boundary energy of a {B}ose gas
in one dimension}, \newblock \href{http://dx.doi.org/10.1103/PhysRevA.4.386}{%
Phys. Rev. A \textbf{4}, 386 (1971)}.

\bibitem{Bar79} R.~Z. Bariev, \newblock \emph{Correlation functions of the
semi-infinite two-dimensional {I}sing model. {I}. {L}ocal magnetization}, %
\newblock \href{http://dx.doi.org/10.1007/BF01019245}{Teoret. Mat. Fiz. 
\textbf{40}(1), 95 (1979)}.

\bibitem{Bar80} R.~Z. Bariev, \newblock \emph{Local magnetization of the
semi-infinite {XY}-chain}, \newblock  \href{http://dx.doi.org/10.1016/0378-4371(80)90224-1%
}{Physica A: Stat. Mech. Appl. \textbf{103}(1), 363 (1980)}.

\bibitem{Bar80a} R.~Z. Bariev, \newblock \emph{Correlation functions of the
semi-infinite two-dimensional {I}sing model. {II}. {T}wo-point correlation
functions}, \newblock \href{http://dx.doi.org/10.1007/BF01032121}{Theor.
Math. Phys. \textbf{42}(2), 173 (1980)}.

\bibitem{Sch85} H.~Schulz, \newblock \emph{Hubbard chain with reflecting ends}, \newblock \href{http://dx.doi.org/10.1088/0022-3719/18/3/010}{J. Phys. C:
Solid State Phys. \textbf{18}(3), 581 (1985)}.

\bibitem{AlcBBBQ87} F.~Alcaraz, M.~Barber, M.~Batchelor, R.~Baxter and
G.~Quispel, \newblock \emph{Surface exponents of the quantum {XXZ}, {A}shkin-{T}eller and {P}otts models}, \newblock  \href{http://dx.doi.org/10.1088/0305-4470/20/18/038}{J. Phys. A: Math. Gen. \textbf{20}, 6397 (1987)}.


\bibitem{Bar88} R.~Z. Bariev, \newblock \emph{Correlation functions of the
semi-infinite two-dimensional {I}sing model. {III}. {I}nfluence of a
``fixed'' boundary}, \newblock \href{http://dx.doi.org/10.1007/BF01028685}{%
Teor. Mat. Fiz. \textbf{77}(1), 127 (1988)}.

\bibitem{MezNR90} L.~Mezincescu, R.~I. Nepomechie and V.~Rittenberg, %
\newblock \emph{Bethe ansatz solution of the {F}ateev-{Z}amolodchikov
quantum spin chain with boundary terms}, \newblock \href{http://dx.doi.org/10.1016/0375-9601(90)90016-H%
}{Phys. Lett. A \textbf{147}(1), 70 (1990)}.

\bibitem{PasS90} V.~Pasquier and H.~Saleur, \newblock \emph{Common
structures between finite systems and conformal field theories through
quantum groups}, \newblock  \href{http://dx.doi.org/10.1016/0550-3213(90)90122-T%
}{Nucl. Phys. B \textbf{330}, 523 (1990)}.

\bibitem{BatMNR90} M.~T. Batchelor, L.~Mezincescu, R.~I. Nepomechie and
V.~Rittenberg, \newblock \emph{{q}-deformations of the {O}(3) symmetric
spin-1 heisenberg chain}, \newblock  \href{http://dx.doi.org/10.1088/0305-4470/23/4/003%
}{J. Phys. A: Math. Gen. \textbf{23}(4), L141 (1990)}.

\bibitem{KulS91} P.~P. Kulish and E.~K. Sklyanin, \newblock \emph{The
general ${U}_q (sl(2))$ invariant {XXZ} integrable quantum spin chain}, %
\newblock  \href{http://dx.doi.org/10.1088/0305-4470/24/8/009}{J. Phys. A:
Math. Gen. \textbf{24}(8) L435 (1991)}.

\bibitem{FriM91}
{L. Freidel and J.M. Maillet}.
\newblock {\em Quadratic algebras and integrable systems}.
\newblock {\em {Physics Letters B}}, {262}:{278--284}, {1991}.

\bibitem{MezN91} L.~Mezincescu and R.~I. Nepomechie, \newblock \emph{%
Integrability of open spin chains with quantum algebra symmetry}, \newblock 
\href{http://dx.doi.org/10.1142/S0217751X91002458}{Int. J. Mod. Phys. A 
\textbf{06}(29), 5231 (1991)}.

\bibitem{MezN91b} L. Mezincescu and R. Nepomechie, \newblock{\em Analytical Bethe Ansatz for quantum-algebra-invariant spin chains}. \newblock Int. J. Mod. Phys. A 6 (1991) 5231 with addendum


\bibitem{deVG93} H.~J. de~Vega and A.~G. Ruiz, \newblock \emph{Boundary {$K$}%
-matrices for the six vertex and the $n(2n-1)$ {$A_{n-1}$} vertex models}, %
\newblock \href{http://dx.doi.org/10.1088/0305-4470/26/12/007}{J. Phys. A :
Math. Gen. \textbf{26}, L519 (1993)}.

\bibitem{deVG93a}
{H. J. de Vega and A. Gonz\'alez-Ruiz}.
\newblock {\em The highest weight property for the SUq(n) invariant spin chains}.
\newblock {Journal of Physics A: Mathematical and General}, {26}:{L519},
  {1993}.

\bibitem{deVG94} H.~J. de~Vega and A.~Gonzalez-Ruiz, \newblock \emph{%
Boundary {K}-matrices for the {XYZ}, {XXZ} and {XXX} spin chains}, \newblock
\href{http://dx.doi.org/10.1088/0305-4470/27/18/021}{J. Phys. A: Math. Gen. 
\textbf{27}(18), 6129 (1994)}.

\bibitem{GhoZ94a} S.~Ghoshal and A.~Zamolodchikov, \newblock \emph{Boundary {$%
S$}-matrix and boundary state in two-dimensional integrable quantum field
theory}, \newblock \href{http://dx.doi.org/10.1142/S0217751X94001552}{Int.
J. Mod. Phys. \textbf{A9}, 3841 (1994)}.

\bibitem{GhoZ94b}
{S. Ghoshal and A. B. Zamolodchikov}.
\newblock {\em Boundary s matrix and boundary state in two-dimensional integrable
  quantum field theory}.
\newblock {International Journal of Modern Physics A}, {9}:{4353},
  {1994}.

\bibitem{ArtMN95} S.~Artz, L.~Mezincescu and R.~I. Nepomechie, \newblock 
\emph{Analytical {B}ethe ansatz for a ${A}_{2n-1}^{(2)}$ , ${B}_n^{(1)}$ , ${%
C}_n^{(1)}$ , ${D}_n^{(1)}$ quantum-algebra-invariant open spin chains}, %
\newblock \href{http://dx.doi.org/10.1088/0305-4470/28/18/006}{J. Phys. A:
Math. Gen. \textbf{28}(18), 5131 (1995)}.

\bibitem{JimKKKM95} M.~Jimbo, R.~Kedem, T.~Kojima, H.~Konno and T.~Miwa, %
\newblock \emph{{XXZ} chain with a boundary}, \newblock  \href{http://dx.doi.org/10.1016/0550-3213(95)00062-W%
}{Nucl. Phys. B \textbf{441}, 437 (1995)}.

\bibitem{JimKKMW95} M.~Jimbo, R.~Kedem, H.~Konno, T.~Miwa and R.~Weston, %
\newblock \emph{Difference equations in spin chains with a boundary}, %
\newblock \href{http://dx.doi.org/10.1016/0550-3213(95)00218-H}{Nucl. Phys.
B \textbf{448}, 429 (1995)}.
  
\bibitem{Kul96}  P.P. Kulish \newblock {\em Yang-Baxter equation and reflection equations in integrable models.} \newblock (1996) \newblock In: Grosse H., Pittner L. (eds) Low-Dimensional Models in Statistical Physics and Quantum Field Theory. \href{}{Lecture Notes in Physics, vol 469. Springer, Berlin, Heidelberg}.

\bibitem{FanHSY96} H.~Fan, B.-Y. Hou, K.-J. Shi and Z.-X. Yang, \newblock 
\emph{Algebraic {B}ethe ansatz for the eight-vertex model with general open
boundary conditions}, \newblock  \href{http://dx.doi.org/10.1016/0550-3213(96)00398-7%
}{Nucl. Phys. B \textbf{478}(3), 723 (1996)}.


\bibitem{Zho96} H.-Q. Zhou, \newblock \emph{Quantum integrability for the
one-dimensional {H}ubbard open chain}, \newblock \href{http://dx.doi.org/10.1103/PhysRevB.54.41%
}{Phys. Rev. B \textbf{54}, 41 (1996)}.

\bibitem{AsaS96} H.~Asakawa and M.~Suzuki, \newblock \emph{Finite-size
corrections in the {XXZ} model and the {H}ubbard model with boundary fields}%
, \newblock \href{http://dx.doi.org/10.1088/0305-4470/29/2/004}{J. Phys. A:
Math. Gen. \textbf{29}(2), 225 (1996)}.

\bibitem{Zho97} H.-Q. Zhou, \newblock \emph{Graded reflection equations and
the one-dimensional {H}ubbard open chain}, \newblock \href{http://dx.doi.org/10.1016/S0375-9601(97)00035-2%
}{Phys. Lett. A \textbf{228}(1), 48 (1997)}.

\bibitem{GuaWY97} X.-W. Guan, M.-S. Wang and S.-D. Yang, \newblock \emph{Lax
pair and boundary {K}-matrices for the one-dimensional {H}ubbard model}, %
\newblock \href{http://dx.doi.org/10.1016/S0550-3213(96)00630-X}{Nucl. Phys.
B \textbf{485}(3), 685 (1997)}. 

\bibitem{ShiW97} M.~Shiroishi and M.~Wadati, \newblock \emph{Bethe ansatz
equation for the hubbard model with boundary fields}, \newblock \href{http://dx.doi.org/10.1143/JPSJ.66.1%
}{J. Phys. Soc. Jpn. \textbf{66}(1), 1 (1997)}.

\bibitem{Tsu98} O.~Tsuchiya, \newblock \emph{Determinant formula for the
six-vertex model with reflecting end}, \newblock  \href{http://dx.doi.org/10.1063/1.532606%
}{J. Math. Phys. \textbf{39}, 5946 (1998)}.


\bibitem{Gua00} X.-W. Guan, \newblock \emph{Algebraic {B}ethe ansatz for the
one-dimensional {H}ubbard model with open boundaries}, \newblock \href{http://dx.doi.org/10.1088/0305-4470/33/30/309%
}{J. Phys. A: Math. Gen. \textbf{33}(30), 5391 (2000)}.

\bibitem{MinRS01} M. Mintchev, E. Ragoucy and P. Sorba. \newblock {\em Spontaneous symmetry breaking in the gl(N)-NLS hierarchy on the half line} \newblock  \href{}{ J. Phys A : Math. Gen. 34 (2001) 8345.}

\bibitem{Nep02} R.~I. Nepomechie, \newblock \emph{Solving the open {XXZ}
spin chain with nondiagonal boundary terms at roots of unity}, \newblock 
\href{http://dx.doi.org/10.1016/S0550-3213(01)00585-5}{ Nucl. Phys. B 
\textbf{622}, 615 (2002)}.

\bibitem{NepR03} R.~I. Nepomechie and F.~Ravanini, \newblock \emph{%
Completeness of the {B}ethe {A}nsatz solution of the open {XXZ} chain with
nondiagonal boundary terms}, \newblock  \href{http://dx.doi.org/10.1088/0305-4470/36/45/003%
}{J. Phys. A: Math. Gen. \textbf{36}(45), 11391 (2003)}.

\bibitem{CaoLSW03} J.~Cao, H.-Q. Lin, K.-J. Shi and Y.~Wang, \newblock \emph{%
Exact solution of {XXZ} spin chain with unparallel boundary fields}, %
\newblock  \href{http://dx.doi.org/10.1016/S0550-3213(03)00372-9}{Nucl.
Phys. B \textbf{663}, 487 (2003)}.

\bibitem{Doi03} A.~Doikou, \newblock \emph{Fused integrable lattice models
with quantum impurities and open boundaries}, \newblock \href{http://dx.doi.org/10.1016/j.nuclphysb.2003.07.001%
}{Nucl. Phys. B \textbf{668}(3), 447 (2003)}.

\bibitem{Nep04} R.~I. Nepomechie, \newblock \emph{Bethe ansatz solution of
the open {XXZ} chain with nondiagonal boundary terms}, \newblock  \href{http://dx.doi.org/10.1088/0305-4470/37/2/012%
}{J. Phys. A: Math. Gen. \textbf{37}, 433 (2004)}.

\bibitem{deGP04} J. de Gier and P. Pyatov, \newblock {\em Bethe ansatz for the Temperley–Lieb loop model with open boundaries}.  \newblock  \href{}{JSTAT 03 (2004) P002}

\bibitem{ArnACDFR04} D.~Arnaudon, J Avan, N.~Cramp\'e, A.~Doikou, L.~Frappat and
E.~Ragoucy, \newblock \emph{General boundary conditions for the $sl(N)$ and $sl(M|N)$
open spin chains}, \newblock \href{https://iopscience.iop.org/article/10.1088/1742-5468/2004/08/P08005/meta}{JSTAT 08 (2004) P005}.


\bibitem{GalM05} W. Galleas and M.J. Martins, \newblock {\em Solution of the SU(N) vertex model with non-diagonal open
boundaries}, \newblock Phys. Lett. A 335 (2005) 167

\bibitem{ArnCDFR05} D.~Arnaudon, N.~Cramp\'e, A.~Doikou, L.~Frappat and
E.~Ragoucy, \newblock \emph{Analytical {B}ethe ansatz for closed and open $%
gl (\mathcal{N})$-spin chains in any representation}, \newblock \href{http://dx.doi.org/10.1088/1742-5468/2005/02/P02007%
}{J. Stat. Mech. \textbf{2005}, P02007 (2005)}.

\bibitem{ArnCDFR06} D.~Arnaudon, N.~Cramp{\'e}, A.~Doikou, L.~Frappat and
E.~Ragoucy, \newblock \emph{Spectrum and {B}ethe ansatz equations for the ${U%
}_q \left( gl(\mathcal{N}) \right)$ closed and open spin chains in any
representation}, \newblock  \href{http://dx.doi.org/10.1007/s00023-006-0280-x%
}{Annales Henri Poincar{\'e} \textbf{7}(7), 1217 (2006)}.

\bibitem{MurNS06} R.~Murgan, R.~I. Nepomechie and C.~Shi, \newblock \emph{%
Exact solution of the open {XXZ} chain with general integrable boundary
terms at roots of unity}, \newblock  \href{http://dx.doi.org/10.1088/1742-5468/2006/08/P08006%
}{J. Stat. Mech. \textbf{2006}(08), P08006 (2006)}.

\bibitem{Doi06} A.~Doikou, \newblock \emph{The open {XXZ} and associated
models at q root of unity}, \newblock  \href{http://dx.doi.org/10.1088/1742-5468/2006/09/P09010%
}{J. Stat. Mech. Theory Exp. \textbf{2006}(09), P09010 (2006)}.

\bibitem{YanNZ06} W.-L. Yang, R.~I. Nepomechie and Y.-Z. Zhang, \newblock 
\emph{{Q}-operator and {T}-{Q} relation from the fusion hierarchy}, %
\newblock  \href{http://dx.doi.org/10.1016/j.physletb.2005.12.022}{Phys.
Lett. B \textbf{633}(4), 664 (2006)}.

\bibitem{KitKMNST07} N.~Kitanine, K.~K. Kozlowski, J.~M. Maillet,
G.~Niccoli, N.~A. Slavnov and  V.~Terras, \newblock \emph{Correlation
functions of the open XXZ chain: I}, \newblock  \href{http://dx.doi.org/10.1088/1742-5468/2007/10/P10009%
}{J. Stat. Mech. \textbf{2007}, P10009 (2007)}.

\bibitem{BasK07} P.~Baseilhac and K.~Koizumi, \newblock \emph{Exact spectrum
of the {XXZ} open spin chain from the $q$-{O}nsager algebra representation
theory}, \newblock \href{http://dx.doi.org/10.1088/1742-5468/2007/09/P09006}{%
J. Stat. Mech. \textbf{2007}, P09006 (2007)}.

\bibitem{YanZ07} W.-L. Yang and Y.-Z. Zhang, \newblock \emph{On the second
reference state and complete eigenstates of the open {XXZ} chain}, \newblock 
\href{http://dx.doi.org/10.1088/1126-6708/2007/04/044}{JHEP \textbf{04}, 044
(2007)}. \newblock 

\bibitem{RagS07} E.~Ragoucy and G.~Satta, \newblock \emph{Analytical bethe
ansatz for closed and open $gl (\mathcal{M}|\mathcal{N})$ super-spin chains
in arbitrary representations and for any Dynkin diagram}, \newblock \href{http://dx.doi.org/10.1088/1126-6708/2007/09/001%
}{JHEP \textbf{09}, 001 (2007)}.

\bibitem{FraNR07} L.~Frappat, R.~I. Nepomechie and E.~Ragoucy, \newblock 
\emph{A complete {B}ethe ansatz solution for the open spin-s {XXZ} chain
with general integrable boundary terms}, \newblock  \href{http://dx.doi.org/10.1088/1742-5468/2007/09/P09009%
}{J. Stat. Mech. \textbf{2007}(09), P09009 (2007)}.

\bibitem{AmiFOR07} L.~Amico, H.~Frahm, A.~Osterloh and G.~Ribeiro, \newblock 
\emph{Integrable spin-boson models descending from rational six-vertex models%
}, \newblock  \href{http://dx.doi.org/10.1016/j.nuclphysb.2007.07.022}{Nucl.
Phys. B \textbf{787}(3), 283 (2007)}.

\bibitem{KitKMNST08} N.~Kitanine, K.~K. Kozlowski, J.~M. Maillet,
G.~Niccoli, N.~A. Slavnov and  V.~Terras, \newblock \emph{Correlation
functions of the open XXZ chain: II}, \newblock \href{http://dx.doi.org/10.1088/1742-5468/2008/07/P07010%
}{J. Stat. Mech. \textbf{2008}, P07010 (2008)}.

\bibitem{Gal08} W.~Galleas, \newblock \emph{Functional relations from the
Yang-Baxter algebra: Eigenvalues of the XXZ model with non-diagonal twisted
and open boundary conditions}, \newblock \href{http://dx.doi.org/10.1016/j.nuclphysb.2007.09.011%
}{ Nucl. Phys. B \textbf{790}, 524 (2008)}.

\bibitem{BelR09} S.~Belliard and E.~Ragoucy, \newblock \emph{The nested {B}%
ethe ansatz for 'all' open spin chains with diagonal boundary conditions}, %
\newblock  \href{http://dx.doi.org/10.1088/1751-8113/42/20/205203}{J. Phys.
A: Math. Theor. \textbf{42}(20), 205203 (2009)}.

\bibitem{Nep10} R.I. Nepomechie, \newblock {\em Nested algebraic Bethe ansatz for open GL(N) spin chains with projected
K-matrices}, \newblock Nucl. Phys. B 831 (2010) 429

\bibitem{FilK10} G.~Filali and N.~Kitanine, \newblock \emph{Partition
function of the trigonometric {SOS} model with reflecting end}, \newblock  
\href{http://dx.doi.org/10.1088/1742-5468/2010/06/L06001}{J. Stat. Mech.
L06001 (2010)}.

\bibitem{FilK11} G.~Filali and N.~Kitanine, \newblock \emph{Spin chains with
non-diagonal boundaries and trigonometric {SOS} model with reflecting end}, %
\newblock  \href{http://dx.doi.org/10.3842/SIGMA.2011.012}{SIGMA \textbf{7},
012 (2011)}.

\bibitem{Fil11} G.~Filali, \newblock \emph{Elliptic dynamical reflection
algebra and partition function of SOS model with reflecting end}, \newblock  
\href{http://dx.doi.org/10.1016/j.geomphys.2011.01.002}{J. Geom. Phys. 
\textbf{61}, 1789 (2011)}.

\bibitem{CraR12} N.~Crampe and E.~Ragoucy, \newblock \emph{Generalized
coordinate {B}ethe ansatz for non-diagonal boundaries}, \newblock \href{http://dx.doi.org/10.1016/j.nuclphysb.2012.01.020%
}{Nucl. Phys. B \textbf{858}(3), 502 (2012)}.

\bibitem{CaoYSW13b} J.~Cao, W.-L. Yang, K.~Shi and Y.~Wang, \newblock \emph{Off-diagonal {B}ethe ansatz solutions of the anisotropic spin-1/2 chains
with arbitrary boundary fields}, \newblock  \href{http://dx.doi.org/10.1016/j.nuclphysb.2013.10.001%
}{Nucl. Phys. B \textbf{877}, 152 (2013)}.

\bibitem{BasB13} P.~Baseilhac and S.~Belliard, \newblock \emph{The
half-infinite {XXZ} chain in {O}nsager's approach}, \newblock \href{http://dx.doi.org/10.1016/j.nuclphysb.2013.05.003%
}{Nucl. Phys. B \textbf{873}(3), 550 (2013)}.

\bibitem{BelC13} S.~Belliard and N.~Cramp{\'e}, \newblock \emph{Heisenberg {%
XXX} model with general boundaries: Eigenvectors from algebraic {B}ethe
ansatz}, \newblock  \href{http://dx.doi.org/10.3842/SIGMA.2013.072}{SIGMA 
\textbf{9}, 072 (2013)}.

\bibitem{BelCR13} S.~Belliard, N.~Cramp{\'e} and E.~Ragoucy, \newblock \emph{
Algebraic {B}ethe ansatz for open {XXX} model with triangular boundary
matrices}, \newblock  \href{http://dx.doi.org/10.1007/s11005-012-0601-6}{%
Lett. Math. Phys. \textbf{103}, 493 (2013)}.


\bibitem{CaoYSW14} Junpeng Cao, Wen-Li Yang, Kangjie Shi and Yupeng Wang, \newblock  \emph{Nested off-diagonal Bethe ansatz and exact solutions
of the su(n) spin chain with generic integrable boundaries} JHEP04(2014)143.

\bibitem{Bel15} S.~Belliard, \newblock \emph{Modified algebraic {B}ethe
ansatz for {XXZ} chain on the segment - {I} - {T}riangular cases}, \newblock
\href{http://dx.doi.org/10.1016/j.nuclphysb.2015.01.003}{Nucl. Phys. B 
\textbf{892}, 1 (2015)}. 

\bibitem{BelP15} S.~Belliard and R.~A. Pimenta, \newblock \emph{Modified
algebraic {B}ethe ansatz for {XXZ} chain on the segment - {II} - general
cases}, \newblock \href{http://dx.doi.org/10.1016/j.nuclphysb.2015.03.016}{%
Nucl. Phys. B \textbf{894}, 527 (2015)}.

\bibitem{AvaBGP15} J.~Avan, S.~Belliard, N.~Grosjean and R.~Pimenta, %
\newblock \emph{Modified algebraic {B}ethe ansatz for {XXZ} chain on the
segment -- {III} -- {P}roof}, \newblock  \href{http://dx.doi.org/10.1016/j.nuclphysb.2015.08.006%
}{Nucl. Phys. B \textbf{899}, 229 (2015)}.

\bibitem{XuHYCYS16} X.~Xu, K.~Hao, T.~Yang, J.~Cao, W.-L. Yang and K.-J.
Shi, \newblock \emph{Bethe ansatz solutions of the $\tau_2$-model with
arbitrary boundary fields}, \newblock  \href{http://dx.doi.org/10.1007/JHEP11(2016)080%
}{JHEP \textbf{11}, 80 (2016)}.

\bibitem{BasB17} P.~Baseilhac and S.~Belliard, \newblock \emph{Non-{A}belian
symmetries of the half-infinite {XXZ} spin chain}, \newblock \href{http://dx.doi.org/10.1016/j.nuclphysb.2017.01.012%
}{Nucl. Phys. B \textbf{916}, 373 (2017)}.


\bibitem{DerKM03b} S.~E. Derkachov, G.~P. Korchemsky and A.~N. Manashov, %
\newblock \emph{Baxter {Q}-operator and separation of variables for the open 
{SL(2,$\mathbb{R}$)} spin chain}, \newblock  \href{http://dx.doi.org/10.1088/1126-6708/2003/10/053%
}{JHEP \textbf{10}, 053 (2003)}.

\bibitem{FraSW08} H.~Frahm, A.~Seel and T.~Wirth, \newblock \emph{Separation
of variables in the open {XXX} chain}, \newblock  \href{http://dx.doi.org/10.1016/j.nuclphysb.2008.04.008%
}{Nucl. Phys. B \textbf{802}, 351 (2008)}.

\bibitem{AmiFOW10} L.~Amico, H.~Frahm, A.~Osterloh and T.~Wirth, \newblock 
\emph{Separation of variables for integrable spin-boson models}, \newblock 
\href{http://dx.doi.org/10.1016/j.nuclphysb.2010.07.005}{Nucl. Phys. B 
\textbf{839}(3), 604 (2010)}.

\bibitem{FraGSW11} H.~Frahm, J.~H. Grelik, A.~Seel and T.~Wirth, \newblock 
\emph{Functional {B}ethe ansatz methods for the open {XXX} chain}, \newblock
\href{http://dx.doi.org/10.1088/1751-8113/44/1/015001}{J. Phys A: Math.
Theor. \textbf{44}, 015001 (2011)}.

\bibitem{Nic12} G.~Niccoli, \newblock \emph{Non-diagonal open spin-1/2 {XXZ}
quantum chains by separation of variables: Complete spectrum and matrix
elements of some quasi-local operators}, \newblock \href{http://dx.doi.org/10.1088/1742-5468/2012/10/P10025%
}{J. Stat. Mech. \textbf{2012}, P10025 (2012)}.

\bibitem{FalN14} S.~Faldella and G.~Niccoli, \newblock \emph{{SOV} approach
for integrable quantum models associated with general representations on
spin-1/2 chains of the 8-vertex reflection algebra}, \newblock  \href{http://dx.doi.org/10.1088/1751-8113/47/11/115202%
}{J. Phys. A: Math. Theor. \textbf{47}, 115202 (2014)}.

\bibitem{FalKN14} S.~Faldella, N.~Kitanine and G.~Niccoli, \newblock \emph{%
Complete spectrum and scalar products for the open spin-1/2 {XXZ} quantum
chains with non-diagonal boundary terms}, \newblock  \href{http://dx.doi.org/10.1088/1742-5468/2014/01/P01011%
}{J. Stat. Mech. \textbf{2014}, P01011 (2014)}.

\bibitem{KitMN14} N.~Kitanine, J.~M. Maillet and G.~Niccoli, \newblock \emph{%
Open spin chains with generic integrable boundaries: Baxter equation and {B}%
ethe ansatz completeness from separation of variables}, \newblock  \href{http://dx.doi.org/10.1088/1742-5468/2014/05/P05015%
}{J. Stat. Mech. \textbf{2014}, P05015 (2014)}.

\bibitem{MaiNP17} J.~M. Maillet, G.~Niccoli and B.~Pezelier, \newblock \emph{%
Transfer matrix spectrum for cyclic representations of the 6-vertex
reflection algebra I}, \newblock  \href{http://dx.doi.org/10.21468/SciPostPhys.2.1.009%
}{SciPost Phys. \textbf{2}, 009 (2017)}.

\bibitem{KitMNT17} N.~Kitanine, J.~M. Maillet, G.~Niccoli and V.~Terras, %
\newblock \emph{The open {XXX} spin chain in the SoV framework: scalar
product of separate states}, \newblock \href{http://dx.doi.org/10.1088/1751-8121/aa6cc9}{J. Phys. A: Math. Theor. \textbf{50}(22), 224001 (2017)}.

\bibitem{MaiNP18} J.~M. Maillet, G.~Niccoli and B.~Pezelier, \newblock \emph{%
Transfer matrix spectrum for cyclic representations of the 6-vertex
reflection algebra {II}}, \newblock  \href{https://scipost.org/10.21468/SciPostPhys.5.3.026}{SciPost Phys. 5, 026 (2018)}.

\bibitem{KitMNT18} N.~Kitanine, J.~M. Maillet, G.~Niccoli and V.~Terras, %
\newblock \emph{The open {XXZ} spin chain in the SoV framework: scalar
product of separate states}, \newblock \href{https://iopscience.iop.org/article/10.1088/1751-8121/aae76f/meta}{J. Phys. A: Math. Theor. \textbf{51}(48), 485201 (2018)}.


\bibitem{deGE05}
{J. de Gier and F. H. L. Essler}.
\newblock {\em Bethe Ansatz solution of the asymmetric exclusion process with open
  boundaries}.
\newblock {Physical Review Letters}, {95}:{240601}, {2005}.

\bibitem{deGE06}
{J. de Gier and F. H. L. Essler}.
\newblock {\em Exact spectral gaps of the asymmetric exclusion process with open
  boundaries}.
\newblock {Journal of Statistical Mechanics: Theory and Experiment}, page
  {P12011}, {2006}.

\bibitem{KinWW06} T.~Kinoshita, T.~Wenger and D.~S. Weiss, \newblock \emph{A
quantum newton's cradle}, \newblock  \href{http://dx.doi.org/10.1038/nature04693%
}{Nature \textbf{440}, 900 EP (2006)}.

\bibitem{HofLFSS07} S.~Hofferberth, I.~Lesanovsky, B.~Fischer, T.~Schumm and
J.~Schmiedmayer, \newblock \emph{Non-equilibrium coherence dynamics in
one-dimensional {B}ose gases}, \newblock \href{http://dx.doi.org/10.1038/nature06149%
}{Nature \textbf{449}, 324 EP (2007)}.

\bibitem{BloDZ08} I.~Bloch, J.~Dalibard and W.~Zwerger, \newblock \emph{%
Many-body physics with ultracold gases}, \newblock \href{http://dx.doi.org/10.1103/RevModPhys.80.885%
}{Rev. Mod. Phys. \textbf{80}, 885 (2008)}.

\bibitem{CraRS10} N.~Cramp{\'e}, E.~Ragoucy and D.~Simon, \newblock \emph{%
Eigenvectors of open {XXZ} and {ASEP} models for a class of non-diagonal
boundary conditions}, \newblock \href{http://dx.doi.org/10.1088/1742-5468/2010/11/P11038%
}{J. Stat. Mech. \textbf{2010}, P11038 (2010)}.

\bibitem{Pro11} T.~Prosen, \newblock \emph{Open {XXZ} spin chain:
Nonequilibrium steady state and a strict bound on ballistic transport}, %
\newblock \href{http://dx.doi.org/10.1103/PhysRevLett.106.217206}{Phys. Rev.
Lett. \textbf{106}, 217206 (2011)}.


\bibitem{TroCFMSEB12} S.~Trotzky, Y.-A. Chen, A.~Flesch, I.~P. McCulloch,
U.~Schollw{\"o}ck,  J.~Eisert and I.~Bloch, \newblock \emph{Probing the
relaxation towards equilibrium in an isolated strongly correlated
one-dimensional bose gas}, \newblock \href{http://dx.doi.org/10.1038/nphys2232%
}{Nature Physics \textbf{8}, 325 EP (2012)}.

\bibitem{SchHRWBBBDMRR12} U.~Schneider, L.~Hackerm{\"u}ller, J.~P.
Ronzheimer, S.~Will, S.~Braun,  T.~Best, I.~Bloch, E.~Demler, S.~Mandt,
D.~Rasch and A.~Rosch, \newblock \emph{Fermionic transport and
out-of-equilibrium dynamics in a homogeneous {H}ubbard model with ultracold
atoms}, \newblock  \href{http://dx.doi.org/10.1038/nphys2205}{Nature Physics 
\textbf{8}, 213 EP (2012)}.

\bibitem{RonSBHLMHBS13} J.~P. Ronzheimer, M.~Schreiber, S.~Braun, S.~S.
Hodgman, S.~Langer, I.~P.  McCulloch, F.~Heidrich-Meisner, I.~Bloch and
U.~Schneider, \newblock \emph{Expansion dynamics of interacting bosons in
homogeneous lattices in one and two dimensions}, \newblock \href{http://dx.doi.org/10.1103/PhysRevLett.110.205301%
}{Phys. Rev. Lett. \textbf{110}, 205301 (2013)}.


\bibitem{EisFG15} J.~Eisert, M.~Friesdorf and C.~Gogolin, \newblock \emph{%
Quantum many-body systems out of equilibrium}, \newblock\href{http://dx.doi.org/10.1038/nphys3215%
}{Nature Physics \textbf{11}, 124 EP (2015)}.

\bibitem{Skl85} E.~K. Sklyanin, \newblock \emph{The quantum {T}oda chain}, %
\newblock in N.~Sanchez, ed., \emph{Non-Linear Equations in Classical and
Quantum Field Theory}, \href{http://dx.doi.org/10.1007/3-540-15213-X_80}{
pp. 196--233. (1985) Springer Berlin Heidelberg, Berlin, Heidelberg}, %
\newblock ISBN 978-3-540-39352-8.

\bibitem{Skl85a} E.~K. Sklyanin, \newblock \emph{Goryachev-{C}haplygin top
and the inverse scattering method}, \newblock  \href{http://dx.doi.org/10.1007/BF02107243%
}{J. Soviet Math. \textbf{31}(6), 3417 (1985)}.

\bibitem{Skl92} E.~K. Sklyanin, \newblock \emph{Quantum inverse scattering
method. {S}elected topics}, \newblock In M.-L. Ge, ed., \emph{Quantum Group
and Quantum Integrable Systems}, pp. 63--97. Nankai Lectures in Mathematical
Physics, World  Scientific (1992), \href{https://arxiv.org/abs/hep-th/9211111%
}{ArXiv:hep-th/9211111}.

\bibitem{Skl95} E.~K. Sklyanin, \newblock \emph{Separation of variables. {N}%
ew trends}, \newblock \href{http://dx.doi.org/10.1143/PTPS.118.35}{Prog.
Theor. Phys. \textbf{118}, 35 (1995)}.

\bibitem{Bax73a}
{R. J. Baxter}.
\newblock {\em Eight-vertex model in lattice statistics and one-dimensional
  anisotropic Heisenberg chain. I, II, III}.
\newblock\  \href{http://dx.doi.org/10.1016/0003-4916(73)90440-5} { {Annals of Physics}, {76}:{1--24 ; 25--47 ; 48--71}, {1973}}.

\bibitem{Bax82B} R.~J. Baxter, \newblock \emph{Exactly solved models in
statistical mechanics}, \newblock Academic Press, London, \newblock ISBN
0-12-083180-5 (1982),  


\bibitem{BabBS96} O.~Babelon, D.~Bernard and F.~A. Smirnov, \newblock \emph{%
Quantization of solitons and the restricted sine-gordon model}, \newblock  
\href{http://dx.doi.org/10.1007/BF02517893}{Commun. Math. Phys. \textbf{182}%
(2), 319 (1996)}. 

\bibitem{BabBS97} O.~Babelon, D.~Bernard and F.~A. Smirnov, \newblock \emph{%
Null-vectors in integrable field theory}, \newblock \href{http://dx.doi.org/10.1007/s002200050122%
}{Commun. Math. Phys. \textbf{186}(3), 601 (1997)}.

\bibitem{Smi98a} F.~A. Smirnov, \newblock \emph{Structure of matrix elements
in the quantum {T}oda chain}, \newblock \href{http://dx.doi.org/10.1088/0305-4470/31/44/019%
}{J. Phys. A: Math. Gen. \textbf{31}(44), 8953 (1998)}.

\bibitem{DerKM03} S.~E. Derkachov, G.~P. Korchemsky, and A.~N. Manashov, %
\newblock \emph{Separation of variables for the quantum {SL(2,$\mathbb{\ R}$)%
} spin chain}, \newblock \href{http://dx.doi.org/10.1088/1126-6708/2003/10/053%
}{JHEP \textbf{07}, 047 (2003)}.

\bibitem{BytT06} A.~Bytsko and J.~Teschner, \newblock \emph{Quantization of
models with non--compact quantum group symmetry. Modular XXZ magnet and
lattice sinh--Gordon model}, \newblock  \href{http://dx.doi.org/10.1088/0305-4470/39/41/S11%
}{J. Phys. A: Mat. Gen. \textbf{39}, 12927 (2006)}.

\bibitem{vonGIPS06} G.~von Gehlen, N.~Iorgov, S.~Pakuliak and V.~Shadura, %
\newblock \emph{The Baxter--Bazhanov--Stroganov model: separation of
variables and the Baxter equation}, \newblock \href{http://dx.doi.org/10.1088/0305-4470/39/23/006%
}{J. Phys. A: Math. Gen. \textbf{39}, 7257 (2006)}.

\bibitem{vonGIPS09} G.~von Gehlen, N.~Iorgov, S.~Pakuliak and V.~Shadura, %
\newblock \emph{Factorized finite-size Ising model spin matrix elements from
separation of variables}, \newblock \href{http://dx.doi.org/10.1088/1751-8113/42/30/304026%
}{J. Phys. A: Math. Theor. \textbf{42}, 304026 (2009)}.

\bibitem{vonGIPST07} G.~von Gehlen, N.~Iorgov, S.~Pakuliak, V.~Shadura and
Y.~Tykhyy, \newblock \emph{Form-factors in the Baxter--Bazhanov--Stroganov
model I: norms and matrix elements}, \newblock\href{http://dx.doi.org/10.1088/1751-8113/40/47/006%
}{J. Phys. A: Math. Theor. \textbf{40}, 14117 (2007)}.

\bibitem{vonGIPST08} G.~von Gehlen, N.~Iorgov, S.~Pakuliak, V.~Shadura and
Y.~Tykhyy, \newblock \emph{Form-factors in the Baxter--Bazhanov--Stroganov
model II: Ising model on the finite lattice}, \newblock  \href{http://dx.doi.org/10.1088/1751-8113/41/9/095003%
}{J. Phys. A: Math. Theor. \textbf{41}, 095003 (2008)}.

\bibitem{NicT10} G.~Niccoli and J.~Teschner, \newblock \emph{The Sine-Gordon
model revisited I}, \newblock  \href{http://dx.doi.org/10.1088/1742-5468/2010/09/P09014%
}{J. Stat. Mech. \textbf{2010}, P09014 (2010)}.

\bibitem{Nic10a} G.~Niccoli, \newblock \emph{Reconstruction of Baxter
Q-operator from Sklyanin SOV for cyclic representations of integrable
quantum models}, \newblock \href{http://dx.doi.org/10.1016/j.nuclphysb.2010.03.009%
}{Nucl. Phys. B \textbf{835}, 263 (2010)}.

\bibitem{GroN12} N.~Grosjean and G.~Niccoli, \newblock \emph{The $\tau_2$%
-model and the chiral Potts model revisited: completeness of Bethe equations
from Sklyanin's SOV method}, \newblock  \href{http://dx.doi.org/10.1088/1742-5468/2012/11/P11005%
}{J. Stat. Mech. \textbf{2012}, P11005 (2012)}.

\bibitem{GroMN12} N.~Grosjean, J.~M. Maillet and G.~Niccoli, \newblock \emph{%
On the form factors of local operators in the lattice sine-{G}ordon model}, %
\newblock \href{http://dx.doi.org/10.1088/1742-5468/2012/10/P10006}{J. Stat.
Mech. \textbf{2012}, P10006 (2012)}.

\bibitem{Nic13} G.~Niccoli, \newblock \emph{Antiperiodic spin-1/2 {XXZ}
quantum chains by separation of variables: Complete spectrum and form factors%
}, \newblock \href{http://dx.doi.org/10.1016/j.nuclphysb.2013.01.017}{Nucl.
Phys. B \textbf{870}, 397 (2013)}. 

\bibitem{Nic13b} G.~Niccoli, \newblock \emph{Form factors and complete
spectrum of {XXX} antiperiodic higher spin chains by quantum separation of
variables}, \newblock \href{http://dx.doi.org/10.1063/1.4807078}{J. Math.
Phys. \textbf{54}, 053516 (2013)}.

\bibitem{Nic13a} G.~Niccoli, \newblock \emph{An antiperiodic dynamical
six-vertex model: {I}. {C}omplete spectrum by {SOV}, matrix elements of the
identity on separate states and connections to the periodic eight-vertex
model}, \newblock  \href{http://dx.doi.org/10.1088/1751-8113/46/7/075003}{J.
Phys. A: Math. Theor. \textbf{46}, 075003 (2013)}. 

\bibitem{GroMN14} N.~Grosjean, J.~M. Maillet and G.~Niccoli, \newblock \emph{
On the form factors of local operators in the Bazhanov-Stroganov and chiral
Potts models}, \newblock \href{http://dx.doi.org/10.1007/s00023-014-0358-9}{Annales Henri Poincar{\'e} \textbf{16}, 1103 (2015)}.

\bibitem{NicT15} G.~Niccoli and V.~Terras, \newblock \emph{Antiperiodic {XXZ}
chains with arbitrary spins: Complete eigenstate construction by functional
equations in separation of variables}, \newblock  \href{http://dx.doi.org/10.1007/s11005-015-0759-9%
}{Lett. Math. Phys. \textbf{105}, 989 (2015)}. 

\bibitem{NicT16} G.~Niccoli and V.~Terras, \newblock \emph{The 8-vertex
model with quasi-periodic boundary conditions}, \newblock \href{http://dx.doi.org/10.1088/1751-8113/49/4/044001%
}{J. Phys. A: Math. Theor. \textbf{49}, 044001 (2016)}. 

\bibitem{LevNT16} D.~Levy-Bencheton, G.~Niccoli and V.~Terras, \newblock 
\emph{Antiperiodic dynamical 6-vertex model by separation of variables {II}
: Functional equations and form factors}, \newblock \href{http://dx.doi.org/10.1088/1742-5468/2016/03/033110%
}{J. Stat. Mech. \textbf{2016}, 033110 (2016)}.

\bibitem{KitMNT16} N.~Kitanine, J.~M. Maillet, G.~Niccoli and V.~Terras, %
\newblock \emph{On determinant representations of scalar products and form
factors in the {SoV} approach: the {XXX} case}, \newblock \href{http://dx.doi.org/10.1088/1751-8113/49/10/104002%
}{J. Phys. A: Math. Theor. \textbf{49}, 104002 (2016)}.


\bibitem{Smi98} F.~A. Smirnov, \newblock \emph{Quasi-classical study of form
factors in finite volume}, \newblock \href{https://arxiv.org/abs/hep-th/9802132%
}{ArXiv:hep-th/9802132 (1998)}.

\bibitem{NieWF09} S.~Niekamp, T.~Wirth and H.~Frahm, \newblock \emph{The {XXZ} model with anti-periodic twisted boundary conditions}, \newblock  \href{http://dx.doi.org/10.1088/1751-8113/42/19/195008%
}{J. Phys. A: Math. Theor. \textbf{42}, 195008 (2009)}.

\bibitem{JiaKKS16}
Y. Jiang, S. Komatsu, I. Kostov and D. Serban
\newblock {\em The hexagon in the mirror: the three-point function in the SoV representation}
\newblock J. Phys A 49:174007, 2016.

\bibitem{RyaV18} P.~Ryan, D.~Volin 
\newblock {\em Separated variables and wave functions for rational gl(N) spin chains in the companion twist frame }
\newblock arXiv:1810.10996.







\bibitem{KulRS81} P.~P. Kulish, N.~Yu. Reshetikhin, and E.~K. Sklyanin. %
\newblock {\em {Y}ang-{B}axter equation and representation theory {I}.} \newblock
Lett. Math. Phys., 5:393--403, 1981.

\bibitem{KulR82} P.~P. Kulish and N.~{Yu}. Reshetikhin. \newblock {\em {$Gl_3$}%
-invariant solutions of the {Y}ang-{B}axter equation and associated quantum
systems.} \newblock J. Sov. Math. , 34:1948--1971, 1986. \newblock %
translated from Zap. Nauch. Sem. LOMI 120, 92-121 (1982).

\bibitem{Res83} N.~Yu. Reshetikhin, \newblock {\em A method of functional
equations in the theory of exactly solvable quantum systems}
\newblock Lett. Math. Phys. {\bf 7} (1983) 205.


\bibitem{KirR86} A. N. Kirillov and N. Yu. Reshetikhin. \newblock {\em Exact solution of the {H}eisenberg {XXZ} model of spin s.} \newblock J. Sov. Math. {\bf 35}, 2627--2643 (1986), Translated from Zap. Nauch. Sem. LOMI {\bf145}, 109--133, (1985).

\bibitem{Ize87} A.~G. Izergin, \newblock \emph{Partition function of the
six-vertex model in a finite volume}, \newblock Sov. Phys. Dokl. \textbf{32}%
, 878 (1987).

\bibitem{Sla89}
N.~A. Slavnov.
\newblock {\em Calculation of scalar products of wave functions and form factors in
  the framework of the algebraic {B}ethe {A}nsatz.}
\newblock Theor. Math. Phys. , 79:502--508, 1989.


\bibitem{FodW12} O.~Foda and M.~Wheeler, \newblock \emph{Variations on {S}lavnov's scalar product}, \newblock  \href{http://dx.doi.org/10.1007/JHEP10(2012)096}{JHEP \textbf{10}, 096 (2012)}.


\end{thebibliography}}
\end{document}